\pdfoutput=1
\RequirePackage{ifpdf}
\ifpdf 
\documentclass[pdftex]{sigma}
\else
\documentclass{sigma}
\fi

\usepackage{mathrsfs,cancel}

\def\q{\quad}
\def\eps{\epsilon}%
\def\tensor{\,\raise2pt\hbox{${}_{\otimes}$}\,}
\def\ptl{\partial}

\def\cal#1{\mathcal{#1}}
\def\mbo#1{\boldsymbol{#1}}

\def\ip#1#2{\langle#1,#2\rangle}

\def\olin#1{\overline{#1}{}}

\def\grad{{\nabla}}
\newcommand{\leftexp}[2]{{\vphantom{#2}}^{#1}{#2}}

\def\halb{\frac{1}{2}}

\def \gm{\gamma}

\def \a{\alpha}
\def \b {\beta}

\newcommand{\ba}{\begin{array}}
\newcommand{\ea}{\end{array}}
\newcommand{\bea}{\begin{eqnarray}}
\newcommand{\eea}{\end{eqnarray}}
\newcommand{\bee}{\begin{eqnarray*}}
\newcommand{\eee}{\end{eqnarray*}}

\renewcommand{\gg}{{\bf g}}
\newcommand{\bgg}{{\bar{\bf g}}}
\newcommand{\tgg}{{\tilde{\bf g}}}

\renewcommand{\H}{\mathcal{H}}

\renewcommand{\a}{\alpha}
\renewcommand{\b}{\beta}
\renewcommand{\d}{\delta}

\renewcommand{\r}{\rho}
\newcommand{\la}{\lambda}

\numberwithin{equation}{section}

\newtheorem{Theorem}{Theorem}[section]
\newtheorem*{Theorem*}{Theorem}
\newtheorem{Corollary}[Theorem]{Corollary}
\newtheorem{Lemma}[Theorem]{Lemma}
\newtheorem{Proposition}[Theorem]{Proposition}
\newtheorem{claim}[Theorem]{Claim}
 { \theoremstyle{definition}
\newtheorem{Definition}[Theorem]{Definition}

 }

\begin{document}

\allowdisplaybreaks

\newcommand{\arXivNumber}{2507.??????}

\renewcommand{\PaperNumber}{054}

\FirstPageHeading

\ShortArticleName{On Axially Symmetric Perturbations of Kerr Black Hole Spacetimes}

\ArticleName{On Axially Symmetric Perturbations\\ of Kerr Black Hole Spacetimes}

\Author{Nishanth GUDAPATI}

\AuthorNameForHeading{N.~Gudapati}

\Address{Department of Mathematics, College of the Holy Cross, \\ 1 College Street, Worcester, MA-01610, USA}
\Email{\href{mailto:ngudapati@holycross.edu}{ngudapati@holycross.edu}}

\ArticleDates{Received February 13, 2024, in final form June 17, 2025; Published online July 11, 2025}

\Abstract{The lack of a positive-definite and conserved energy is a serious obstacle in the black hole stability problem. In this work, we will show that there exists a positive-definite and conserved Hamiltonian energy for axially symmetric linear perturbations of the exterior of Kerr black hole spacetimes. In the first part, based on the Hamiltonian dimensional reduction of $3+1$ axially symmetric, Ricci-flat Lorentzian spacetimes to a~${2+1}$ Einstein-wave map system with the negatively curved hyperbolic 2-plane target, we construct a positive-definite, spacetime gauge-invariant energy functional for linear axially symmetric perturbations in the exterior of Kerr black holes, in a manner that is also gauge-independent on the target manifold. In the construction of the positive-definite energy, various dynamical terms at the boundary of the orbit space occur critically. In the second part, after setting up the initial value problem in harmonic coordinates, we prove that the positive energy for the axially symmetric linear perturbative theory of Kerr black holes is strictly conserved in time, by establishing that all the boundary terms dynamically vanish for all times. This result implies a form of dynamical linear stability of the exterior of Kerr black hole spacetimes.}

\Keywords{Kerr black holes; black hole stability problem; ergo-region; Hamiltonian mechanics; wave maps; Poisson equation}

\Classification{83C57; 58E20; 58E30; 35C15; 35Q75}

{\small
\tableofcontents

}

\section{Geometric mass-energy and perturbations of black holes}

The stability of stationary solutions of a physical law serves as an impetus to the validity of the law. In the context of Einstein's equations for general relativity, an important stationary solution is the Kerr family of black holes which is also an asymptotically flat, axially symmetric family of solutions of the $(3+1)$-dimensional vacuum Einstein equations for general relativity
\begin{align}\label{EVE}
\bar{R}_{\mu \nu} =0 \qquad \text{for a Lorentzian manifold} \ \bigl(\bar{M}, \bar{g}\bigr).
\end{align}
In part due to the physical relevance and the mathematical beauty arising from its multiple miraculous properties (see, e.g., \cite{Chandrasekhar_83, Teukolsky_15}), the problem of stability of Kerr black hole spacetimes within the class of Einstein's equations \eqref{EVE} has been a subject of active research interest since their discovery by R.~Kerr in~1963. However, geometric properties of Kerr black hole spacetimes such as stationarity (as opposed to staticity), trapping of null geodesics and the general issue of gauge dependence of metric perturbations cause significant obstacles in the resolution of this `black hole stability' problem.
In this work, we focus on the serious issue caused by the stationarity of the Kerr metric:
\begin{enumerate}\setlength{\leftskip}{0.03cm}
\item [(P1)] The problem of the ergo-region, the lack of a positive-definite and conserved energy caused by the shift vector of the Kerr metric.
\end{enumerate}

It may be noted that an asymptotically flat spacelike Riemannian hypersurface $\bigl(\olin{\Sigma}, \bar{q}\bigr)$ such that \smash{$\bar{M} = \olin{\Sigma} \times \mathbb{R}$} satisfies the Einstein's equations for general relativity \eqref{EVE}, has a positive-definite total (ADM) mass $m_{\text{ADM}}$:
\begin{align*}
m_{\text{ADM}} := \lim_{r \to \infty} \int_{\mathbb{S}^2(r)} \sum^3_{i,j,k=1} (\ptl_k \bar{q}_{i \ell} - \ptl_i \bar{q}_{\ell k}) \frac{x^i}{r} \bar{\mu}_{\mathbb{S}^2}, \qquad \bar{q} \ \text{is asymptotically Euclidean},
\end{align*}
from the celebrated positive-mass theorems of Schoen--Yau and Witten \cite{schoen-yau-1,schoen-yau-2,witten-pmt}. However, it is not necessary that the positivity of energy carries forward to the perturbative theory of the Einstein equations. In a general asymptotically flat manifold, it is a priori not determinate whether the mass-energy at infinity increases or decreases for small perturbations.
 This outcome can be seen in the energy of a (linear) scalar wave equation propagating in the exterior of Kerr black holes -- an illustrative, albeit a special `test' case of the perturbations of the Kerr metric.

 However, as we already alluded to, the difficulty of constructing a positive-definite energy for the perturbative theory is not only due to the shift vector (or the ergo-region) of the Kerr metric. Even if one considers the Schwarzschild metric (the special case of vanishing angular momentum of Kerr),
 \begin{align}\label{sch}
 \bar{g} =- f {\rm d}t^2+ f^{-1} {\rm d}r^2 + r^2 {\rm d}\omega^2_{\mathbb{S}^2},
 \end{align}
 where $f := \bigl(1-2mr^{-1}\bigr)$, it is not immediate that there exists a positive-definite energy for the perturbative theory of \eqref{sch}. In 1974, Moncrief had devised a `Hamiltonian' for the perturbative theory of Schwarzschild based on the ADM formalism of Einstein's equations \cite{Moncrief_74}. Suppose the Lorentzian spacetime \smash{$\bigl(\bar{M}, \bar{g}\bigr)$} admits a $3+1$ ADM decomposition
\begin{align*}
\bar{g} = - N^2 {\rm d}t^2 + \bar{q}_{ij} \bigl({\rm d}x^i + N^i {\rm d}t\bigr) \otimes \bigl({\rm d}x^j + N^j {\rm d}t\bigr),
\end{align*}
then the ADM constraint and evolution equations are given by the variational principle for the phase space $X_{\text{ADM}} := \bigl\{ \bigl(\mbo{\bar{\pi}}^{ij}, \bar{q}_{ij}\bigr),\, i, j= 1, 2, 3 \bigr\}$:
\begin{subequations}
\begin{align}
&I_{\text{ADM}} := \int \bigl(\bar{\mbo{\pi}}^{ij} \ptl_t \bar{q}_{ij} - N H -N^i H_i \bigr){\rm d}^4x, \label{ADM-var-first}
\intertext{where $\mbo{\bar{\pi}}^{ij}$ is the momentum conjugate to $\bar{q}_{ij}$,}
&H := \bar{\mu}^{-1}_{\bar{q}} \biggl(\Vert \mbo{\bar{\pi}} \Vert^2_{\bar{q}} - \halb {\rm Tr}_{\bar{q}} (\mbo{\bar{\pi}})^2 \biggr) - \bar{\mu}_{\bar{q}} R_{\bar{q}}, \\
&H_i := -2 \leftexp{(\bar{q})}{\grad}_j \mbo{\bar{\pi}}^j_i,
\end{align}
\end{subequations}
$R_{\bar{q}}$ is the scalar curvature of $\bigl(\olin{\Sigma}, \bar{q}\bigr)$,
and $\{ N, N^i \}$ are the Lagrange multipliers. We say that $\bigl(\olin{\Sigma}, \bar{q}, \bar{\mbo{\pi}}\bigr)$ are asymptotically flat if diffeomorphic to $\mathbb{R}^3 \setminus \olin{B}_R(0)$, where $\olin{B}_R (0)$ is a closed ball of radius $R$ centered at the origin and for $\a >0$
\begin{align*}
&\bar{q}_{ij} = \biggl(1+ \frac{M}{r} \biggr) \bar{\delta}_{ij} + \mathcal{O}\bigl(r^{-1-\a}\bigr), \\
&\bar{\mbo{\pi}}^{ij} = \mathcal{O}\bigl(r^{-2-\a}\bigr),
\end{align*}
as $ r= \vert x \vert \to 0$. As a consequence, we have $M =m_{\text{ADM}}$.
It follows from the variational principle~\eqref{ADM-var-first} that the evolution equations are given by
\begin{gather*}
 \ptl_t \bar{q}_{ij} = 2 \bar{N} \bar{\mu}_{\bar{q}} \biggl(\bar{\mbo{\pi}}_{ij} - \halb q_{ij} {\rm Tr}_{\bar{q}}(\bar{\mbo{\pi}}) \biggr) + \leftexp{(\bar{q})}{\grad}_j \bar{N}_i+ \leftexp{(\bar{q})}{\grad}_i \bar{N}_j, \\
\ptl_t \bar{\mbo{\pi}}^{ij} = - \bar{N} \bar{\mu}_{\bar{q}} \biggl(\leftexp{(\bar{q})}{R}^{ij} - \halb \bar{q}^{ij} R_{\bar{q}} \biggr) + \halb \bar{N} \bar{\mu}^{-1}_{\bar{q}} q^{ij} \biggl({\rm Tr}_{\bar{q}} \bigl(\bar{\mbo{\pi}}^2\bigr) - \halb {\rm Tr}_{\bar{q}}(\bar{\mbo{\pi}})^2 \biggr)
\\ \hphantom{\ptl_t \bar{\mbo{\pi}}^{ij} =}{}
 -2 \bar{N} \bar{\mu}^{-1}_{\bar{q}} \biggl(\bar{\pi}^{im} \bar{\mbo{\pi}}^j_m - \halb \bar{\mbo{\pi}}^{ij} {\rm Tr}_{\bar{q}}(\bar{\mbo{\pi}})\biggr) + \bar{\mu}_q \bigl( \leftexp{(\bar{q})}{\grad}^
i \leftexp{(\bar{q})}{\grad}^
j \bar{N} - \bar{q}^{ij} \leftexp{(\bar{q})}{\grad}
^m \leftexp{(\bar{q})}{\grad}_m \bar{N}\bigr)
\\ \hphantom{\ptl_t \bar{\mbo{\pi}}^{ij} =}{}
+ \leftexp{(\bar{q})}{\grad}_m \bigl(\bar{\mbo{\pi}}^{ij} \bar{N} ^m\bigr) - \leftexp{(\bar{q})}{\grad}_m \bar{N}^i \bar{\mbo{\pi}}^{mj}
- \leftexp{(\bar{q})}{\grad}_m \bar{N}^j \bar{\mbo{\pi}}^{mi}.
\end{gather*}
Suppose we consider the small perturbations of the initial data of Schwarzschild black hole spacetimes: $\bar{q}= \bar{q}_{s} + \eps \bar{q}'$ and $\bar{\mbo{\pi}}= \bar{\mbo{\pi}}_s + \eps \bar{\mbo{\pi}}'$, Moncrief's Hamiltonian energy formula is
\begin{align}\label{vm-74}
H_{\text{pert}} := \int_{\olin{\Sigma}} \biggl\{& \bar{N} \bar{\mu}^{-1}_{\bar{q}} \biggl(\Vert \bar{\mbo{\pi}}' \Vert_{\bar{q}}^2 - \halb {\rm Tr}_{\bar{q}}\bigl(\bar{\mbo{\pi}}'\bigr)^2 \biggr)\\
& {} + \halb \bar{N} \bar{\mu}_{\bar{q}} \biggl(\halb \leftexp{(\bar{q})}{\grad}_k \bar{q}'_{ij} \leftexp{(\bar{q})}{\grad}^k \bar{q}'^{ij}- \leftexp{(\bar{q})}{\grad}_k \bar{q}'_{ij} \leftexp{(\bar{q})}{\grad}^j \bar{q}'^{ik}
- \halb \leftexp{(\bar{q})}{\grad}_i {\rm Tr} \bigl(\bar{q}'\bigr) \leftexp{(\bar{q})}{\grad}^i {\rm Tr} \bigl(\bar{q}'\bigr)
\notag\\
&\hphantom{+ \halb \bar{N} \bar{\mu}_{\bar{q}} \biggl(}{}\ + 2\leftexp{(\bar{q})}{\grad}_i {\rm Tr}\bigl(\bar{q}'\bigr) \leftexp{(\bar{q})}{\grad}_j \bar{q}'^{ij} + {\rm Tr} \bigl(\bar{q}'\bigr) \leftexp{(\bar{q})}{\grad}^2_{ij} \bar{q}'^{ij} - {\rm Tr} \bigl(\bar{q}'\bigr) \bar{q}'_{ij} \bar{R}^{ij}_{\bar{q}} \biggr) \biggr\} {\rm d}^3x,\notag
\end{align}
which is a volume integral on the hypersurface $\olin{\Sigma}$.
Moncrief used the Hamiltonian formulation to decompose the metric perturbations into gauge-dependent, gauge-independent and constraints; and ultimately reconciled with the Regge--Wheeler--Zerilli results \cite{Regge-Wheeler_57,Zerilli_70}. An important feature of these results is that the energy functional \eqref{vm-74} can be realized to be positive-definite for both odd and even parity perturbations. Using tensor harmonics, positive-definite energy functionals for both odd and even parity perturbations of Schwarzschild black holes were constructed in~\cite{Moncrief_74}. In this spirit, a number of pioneering articles on the perturbations of static black holes were written by Moncrief \cite{Moncrief_74_1,Moncrief_74_2,Moncrief_74_3}.

The subject of this article is to focus on axially symmetric perturbations of the Kerr metric. In precise terms, the Kerr metric $\bigl(\bar{M}, \bar{g}\bigr)$ can be represented in Boyer--Lindquist coordinates $(t, r, \theta, \phi)$ as
\begin{align}
\bar{g} &{}= - \biggl(\frac{\Delta - a^2 \sin^2 \theta}{\Sigma} \biggr) {\rm d}t^2 - \frac{2a \sin^2 \theta \bigl(r^2 + a^2 -\Delta\bigr)}{\Sigma} {\rm d}t {\rm d}\phi \notag\\
&\quad{}+ \biggl(\frac{\bigl(r^2 +a^2\bigr)^2 -\Delta a^2 \sin^2 \theta}{\Sigma}\biggr) \sin^2 \theta {\rm d}\phi^2 + \frac{\Sigma}{\Delta} {\rm d}r^2 + \Sigma {\rm d}\theta^2, \label{BL-Kerr}
\end{align}
where
\begin{align*}
&\Sigma := r^2 + a^2 \cos^2 \theta, \\
&\Delta := r^2 -2Mr + a^2 \quad \textnormal{with the real roots} \ \{r_-,r_+\}, \qquad
 r_{+} := M + \sqrt{M^2- a^2} > r_{-}
\end{align*}
and
\[
\theta \in [0, \pi],\qquad r \in (r_+, \infty),\qquad \phi \in [0, 2\pi).
\]
It is well known (cf.\ \cite{YCB-VM96, Ycb-Mon_01, diss_13}) that the Einstein equations \eqref{EVE} on spacetimes $\bigl(\bar{M}, \bar{g}\bigr)$ with one isometry $\bigl(\frac{\ptl}{\ptl \phi}\bigr)$, represented in Weyl--Papapetrou coordinates,
\begin{align*}
\bar{g} = {\rm e}^{-2\gamma} g + {\rm e}^{2\gamma} ({\rm d}\phi + \mathcal{A}_\nu {\rm d}x^\nu)^2, \qquad \nu= 0, 1, 2,
\end{align*}
where ${\rm e}^{2 \gamma}$ is the square of the norm of the rotational Killing vector $\ptl_\phi$, $\mathcal{A}$ is a $1$-form and $g$ is the induced metric on orbit space $M := \bar{M}/{\rm SO}(2)$, admit a dimensional reduction to a $(2 + 1)$-dimensional Einstein wave map system
\begin{subequations} \label{ewm-system}
\begin{align}
&E_{\mu \nu} = T_{\mu \nu}, \\
&\square_g U^A + \leftexp{(h)}{\Gamma}^A_{BC} g^{\mu \nu} \ptl_\mu U^B \ptl_\nu U^C= 0 \qquad \text{on} \ (M, g),
\end{align}
\end{subequations}
where $\square_g$ is the covariant wave operator, $E_{\mu \nu}$ is the Einstein tensor in the interior of the quotient $(M, g) := \bigl(\bar{M}, \bar{g}\bigr)/ {\rm SO}(2)$ and $T$ is the stress energy tensor of the wave map $U \colon (M, g) \to (\mathbb{N}, h)$, $\mathbb{N}$ is the negatively curved hyperbolic $2$-plane,
\begin{align*}
T_{\mu \nu} = \ip{\ptl_\mu U} {\ptl_\nu U }_{h(U)} - \halb g_{\mu \nu} \ip{\ptl_\sigma U}{ \ptl^\sigma U}_{h(U)}.
\end{align*}
Introducing the coordinates $(\r, z)$ such that $\r = R \sin \theta$, $z= R \cos
\theta$, where $R := \halb \bigl(r-M+ \sqrt{\Delta}\bigr)$, the Kerr metric \eqref{BL-Kerr} can be represented in the Weyl--Papapetrou form as
\begin{align*}
\bar{g} = \Sigma \zeta^{-1} \bigl(-\Delta {\rm d}t^2 + \zeta R^{-2} \bigl({\rm d}\r^2 + {\rm d}z^2\bigr)\bigr) + \sin^2\theta \Sigma^{-1} \zeta
\bigl({\rm d}\phi - 2aMr \zeta^{-1} {\rm d}t\bigr)\bigr)^2,
\end{align*}
where \smash{$\zeta = \bigl(r^2 +a^2\bigr)^2 - a^2 \Delta \sin^2 \theta$}. Furthermore, the Kerr metric can also be represented in the Weyl--Papapetrou form using functions $(\bar{\r}, \bar{z})$
such that
\begin{align*}
\bar{\r} = \r - \frac{\bigl(M^2 -a^2\bigr)}{4 R^2}\r, \qquad \text{and} \qquad \bar{z} = z + \frac{\bigl(M^2 -a^2\bigr)}{4 R^2}z
\end{align*}
(cf.\ \cite[Appendix A]{GM17_gentitle} for details). Now we shall turn to the axially symmetric perturbation theory of the Kerr metric.

In view of the peculiar behaviour of the $2 + 1$ Einstein-wave map system, a detailed discussion of our methods is relevant for our article and perhaps also interesting to the reader. Consider the Hamiltonian energy of an axially symmetric linear wave equation propagating on the Kerr metric ($\square_g u =0$),
\begin{align*}
H^{\text{LW}} := \int_{\olin{\Sigma}} \biggl(\halb \bar{N}\bar{\mu}^{-1}_q v^2 + vN^i \ptl_i u+ \halb N \bar{\mu}_q \bar{q}^{ij} \ptl_i u \ptl_j u \biggr) {\rm d}^3 x,
\end{align*}
where $v$ is the conjugate momentum of $u$, the energy is directly positive-definite.
However, this simplification does not carry forward to the Maxwell equations on the Kerr metric
\[
H^{{\rm Max}} := \int \biggl(\halb N \bar{q}_{ij} \bar{\mu}_q (\mathfrak{E}^i \mathfrak{E}^j + \mathfrak{B}^i \mathfrak{B}^j) + N^i \eps_{ijk}\mathfrak{E}^j \mathfrak{B}^k \biggr) {\rm d}^3 x,
\]
where
\[
\mathfrak{E}^i := \halb \eps^{ijk} \leftexp{*}F_{jk}, \qquad
\mathfrak{B}^i := \halb \eps^{ijk} F_{jk}
\]
are the electric and magnetic field densities respectively and $F$ is the Faraday tensor which satisfies the Maxwell equations
\begin{align*}
 {\rm d} \star F =0, \qquad {\rm d} F =0.
\end{align*}

Actually, one can construct counter examples of positivity of energy density, for instance, using time-symmetric Maxwell fields (cf.\ the discussion in \cite[Section 2]{GM17_gentitle}). In a crucial work, Dain--de Austria \cite{DA_14} had arrived at a positive-definite energy for the gravitational perturbations of extremal Kerr black holes using the Brill mass formula \cite{D09} and subsequent use of Carter's identity \cite{Car_71}, originally developed for black hole uniqueness theorems. In a Weyl coordinate system for the spacelike hypersurface $\bigl(\olin{\Sigma}, q\bigr)$ in extremal Kerr spacetime, their positive-definite energy for axially symmetric perturbations is obtained
from perturbations of the Brill mass formula, which in turn is obtained from multiplying a factor with the Hamiltonian constraint that conveniently results in a volume form (in the chosen Weyl coordinate system) that is useful in its representation.

In order to construct a positive-definite energy for the perturbations of Kerr--Newman metric for the full-subextremal range, we delve into the variational structure of the relevant field equations. The beautiful linearization stability framework, developed by V.~Moncrief, J.~Marsden and A.~Fischer~\cite{FM_75,FMM_80,Mon_75}, provides a
natural mechanism to construct an energy-functional based on the kernel of the adjoint of the deformations around the Kerr metric of the dimensionally reduced constraint map. This recognition allows us to extend results to the full sub-extremal range $(\vert a \vert, \vert Q \vert <M)$ of the perturbations of the Kerr--Newman metric \cite{GM17_gentitle}, which is a solution of Einstein--Maxwell equations of general relativity.

Let us briefly outline our construction of a positive-definite energy functional for axially symmetric perturbations of Kerr black hole spacetimes. Consider the ADM decomposition of $\bar{M}=\olin{\Sigma} \times \mathbb{R}$. Suppose the group ${\rm SO}(2)$ acts on $\olin{\Sigma}$ through isometries. Let $\Gamma$ be the non-empty fixed-point set. Suppose the norm squared of the Killing vector generating the rotational isometry is denoted by ${\rm e}^{2\gamma}$. In the dimensional reduction ansatz, the metric $\bar{g}$ is
\begin{align}\label{KK-ADM}
\bar{g} = {\rm e}^{-2\gamma} \bigl(-N^2 {\rm d}t^2 + q_{ab} ({\rm d}x^a + N^a {\rm d}t) \otimes \bigl({\rm d}x^b + N^b {\rm d}t\bigr)\bigr) + {\rm e}^{2\gamma} ({\rm d}\phi + \mathcal{A}_0 {\rm d}t + \mathcal{A}_a {\rm d}x^a)^2.
\end{align}
It may be noted that this metric form is combination of ADM formalism and Weyl--Papapetrou coordinates \cite{Moncrief_74}.
In the dimensional reduction framework, identifying the reduced conjugate momenta, which form the reduced phase space in $(M, g);$ and the corresponding reduced Hamiltonian formalism is nontrivial. This construction was done in \cite{kaluz1}. Define the conjugate momentum corresponding to the metric $q_{ab}$ as follows:
\begin{align*}
\mbo{\pi}^{ab} = {\rm e}^{-2\gamma} \bar{\mbo{\pi}}^{ab}, \qquad \bar{q}_{ab} = {\rm e}^{-2\gamma} q_{ab} + {\rm e}^{2\gamma} \mathcal{A}_a \mathcal{A}_b.
\end{align*}
As a consequence, the ADM action principle transforms to
\begin{align*}
J = \int^{t_2}_{t_1}\int_{\Sigma} \bigl(\mbo{\pi}^{ab} \ptl_t q_{ab} + \mathcal{E}^a \ptl_t \mathcal{A}_a + p \ptl_t \gamma - N H - N^aH_a+ \mathcal{A}_0 \ptl_a \mathcal{E}^a \bigr) {\rm d}^2x{\rm d}t,
\end{align*}
where the phase-space is now
\[
 \{ (q, \mbo{\pi}), (\mathcal{A}_a, \mathcal{E}^a), (\gamma, p) \}
\]
with the Lagrange multipliers
\[
 \{ N, N^a, \mathcal{A}_0 \}
\]
and the constraints
\begin{subequations} \label{2+1constraints}
\begin{align}
&H= \bar{\mu}^{-1}_q \bigl(\Vert \mbo{\pi} \Vert^2_q - {\rm Tr}_q(\mbo{\pi})^2\bigr) + \frac{1}{8} p^2 + \halb {\rm e}^{-4\gamma}
q_{ab} \mathcal{E}^a \mathcal{E}^b + \bar{\mu}_q \bigl(-R_q + 2 q^{ab} \ptl_a \gamma \ptl_b \gamma\bigr) \notag\\
&\hphantom{H=}{}+ \frac{1}{4} {\rm e}^{4\gamma} q^{ab}q^{bd} \ptl_{[b} \mathcal{A}_{a]} \ptl_{[d} \mathcal{A}_{c]}, \\
&H_a= -2 \leftexp{(q)}{ \grad}_b \mbo{\pi}^b_a + p \ptl_a \gamma + \mathcal{E}^b \bigl(\ptl_{[a}\mathcal{A}_{b]} \bigr),\\
&\ptl_a \mathcal{E}^a =0.
\end{align}
\end{subequations}
It may be noted that $\mathcal{E}$ and $p$ are the momenta conjugate to $\mathcal{A} $ and $\gamma$ respectively.
After applying the Poincar\'e lemma on $\mathcal{E}$ and introducing the twist potential such that $\mathcal{E}^a =: \eps^{ab} \ptl_b \omega$,
we transform into the phase space
\[ X_{{\rm EWM}} = \bigl\{ (\gamma, p), (\omega, \mbo{r}), \bigl(q_{ab}, \mbo{\pi}^{ab}\bigr) \bigr\},
\]
where $\mbo{r}$ is the momentum conjugate to $\omega$. The variational principle reduces to
\begin{subequations} \label{2+1-no-constraints}
\begin{align}
&\tilde{J} := \int^{t_2}_{t_1}\int_{\Sigma} \bigl( \mbo{\pi}^{ab} \ptl_t q_{ab} + p \ptl_t \gamma + r \ptl_t \omega - N H - N^a H_a \bigr) {\rm d}^2x {\rm d}t,
\intertext{where $H$ and $H_a$ are now}
&H= \bar{\mu}^{-1}_q \biggl( \Vert \mbo{\pi} \Vert^2_q - {\rm Tr}_q (\mbo{\pi})^2 + \frac{1}{8} p^2 + \frac{1}{2} {\rm e}^{4\gamma} \mbo{r}^2 \biggr) \notag\\
&\hphantom{H=}{} + \bar{\mu}^{-1}_q \biggl(-R_q + 2 q^{ab} \ptl_a \gamma \ptl_b \gamma + \halb {\rm e}^{-4\gamma} q^{ab} \ptl_a \omega \ptl_b \omega \biggr), \\
&H_a= -2 \leftexp{(q)}{\grad}_b \mbo{\pi}^b_a + p \ptl_a \gamma + \mbo{r} \ptl_a \omega.
\end{align}
\end{subequations}
with the Lagrange multipliers $N$, $N_a$. After computing the linearized field equations involving the perturbed phase space
\[
X' := \bigl\{(\gamma', p'), (\omega', \mbo{r}'),\bigl(q'_{ab}, \mbo{\pi}'^{ ab} \bigr) \bigr\},
\]
it can be noted that $ (N, 0)^{\mathsf{T}}$ is the element of the kernel of the adjoint of the perturbed constraint map. This in turn provides a candidate for the energy, analogous to \eqref{vm-74}. The resulting expression has the potential energy
\begin{align*}
D^2 \cdot \mathcal{V} = \bar{\mu}_q q^{ab} \bigl(4 \ptl_a \gamma'
\ptl_b \gamma' + {\rm e}^{-4\gamma} \ptl_a \omega' \ptl_b \omega' + 8{\rm e}^{-4\gamma} \gamma'^2 \ptl_a \omega \ptl_b \omega-8{\rm e}^{-4\gamma} \gamma' \ptl_a \omega \ptl_b \omega'\bigr),
\end{align*}
where $D^2 \cdot$ is the second variational derivative.
This expression is then transformed to a positive-definite form using the Carter--Robinson identities. Firstly, it may be noted that the original Carter--Robinson identities are not restrictive to the choice of the function `$\r$' (in \cite{Car_71} and in~\mbox{\cite[equation~(5)]{Rob_74}}) and can thus be generalized as follows:
\begin{gather*}
\bar{\mu}_q q^{ab} \bigl(4 \ptl_a \gamma'
\ptl_b \gamma' + {\rm e}^{-4\gamma} \ptl_a \omega' \ptl_b \omega' + 8{\rm e}^{-4\gamma} \gamma'^2 \ptl_a \omega \ptl_b \omega-8{\rm e}^{-4\gamma} \gamma' \ptl_a \omega \ptl_b \omega'\bigr) \\
\qquad\quad{}+ \ptl_b \bigl(N \bar{\mu}_q q^{ab} \bigl(-2 {\rm e}^{-4\gamma} \ptl_a \gamma \omega' + {\rm e}^{-4\gamma} \omega' + 4 {\rm e}^{-4\gamma} \gamma' \ptl_a \omega\bigr)\bigr) \\
\qquad\quad{}+ \halb \bar{\mu}_q {\rm e}^{-4\gamma} L_1 \bigl({\rm e}^{-2\gamma} \omega'\bigr) + \bar{\mu}_q L_2 \bigl(-4\gamma' \omega'\bigr) \\
\qquad{}= N \bar{\mu}_q q^{ab} \bigl(\leftexp{(1)}{V}_a \leftexp{(1)}{V}_b + \leftexp{(2)}{V}_a \leftexp{(2)}{V}_b + \leftexp{(3)}{V}_a \leftexp{(3)}{V}_b\bigr),
\end{gather*}
where
\begin{align*}
&\leftexp{(1)}{V}_a = 2 \ptl_a \gamma' + {\rm e}^{-4\gamma} \omega' \ptl_a \omega, \\
&\leftexp{(2)}{V}_a= -\ptl_a \bigl({\rm e}^{-2\gamma}\omega'\bigr) + 2 {\rm e}^{-2\gamma}\gamma' \ptl_a \omega, \\
&\leftexp{(3)}{V}_a = 2 \ptl_a \gamma \omega' -2\gamma' \ptl_a \omega,
\end{align*}
and
\begin{align*}
&L_1 := {\rm e}^{-2 \gamma} \bigl(\ptl_b \bigl(N \bar{\mu}_q q^{ab} \ptl_a \gamma\bigr) + N {\rm e}^{-4\gamma} \bar{\mu}_q q^{ab} \ptl_a \omega \ptl_b \omega\bigr), \\
&L_2 := - \ptl_b\bigl(N \bar{\mu}_q q^{ab} {\rm e}^{-4\gamma} \ptl_a \omega\bigr).
\end{align*}
Upon substitution into the potential energy, this results in a positive-definite energy of the form
\begin{align*}
H^{{\rm Reg}} = \int_{\Sigma} \biggl\{& N \bar{\mu}^{-1}_q \biggl(
\varrho'^{b}_a \varrho'^{a}_b + \frac{1}{8} p'^2+ \halb
{\rm e}^{4\gamma} \mbo{r}'^2 \biggr) - \halb \bar{\mu}_q \tau'^2 \notag\\
&{}+ N \bar{\mu}_q q^{ab} \biggl(2 \biggl(\ptl_a \gamma' + \halb {\rm e}^{-4\gamma} \omega' \ptl_a \omega\biggr) \biggl(\ptl_b \gamma' + \halb {\rm e}^{-4\gamma} \omega' \ptl_b \omega\biggr) \notag\\
&\hphantom{+ N \bar{\mu}_q q^{ab} \biggl(\ }{}+ 2\bigl(\gamma' {\rm e}^{-2\gamma}\ptl_a \omega - \ptl_a \bigl({\rm e}^{-2\gamma} \omega'\bigr)\bigr)\bigl(\gamma' {\rm e}^{-2\gamma}\ptl_b \omega - \ptl_b \bigl({\rm e}^{-2\gamma} \omega'\bigr)\bigr) \notag\\
 &\hphantom{+ N \bar{\mu}_q q^{ab} \biggl(\ }{}+ 2{\rm e}^{-4\gamma} \bigl(\ptl_a \gamma \omega' - \gamma' \ptl_a \omega\bigr)\bigl(\ptl_b \gamma \omega' - \gamma' \ptl_b \omega\bigr) \biggr)
\biggr\}{\rm d}^2 x
\end{align*}
modulo a time-coordinate gauge condition $\tau'=0$.
It is then shown that this energy functional is a Hamiltonian for the dynamics of the reduced Einstein equations in the perturbative phase-space and a spacetime divergence-free vector field density is constructed,
\[
J^{{\rm Reg}} = \bigl(J^{{\rm Reg}}\bigr)^t \ptl_t + \bigl(J^{{\rm Reg}}\bigr)^a \ptl_a,
\]
where $\bigl(J^{{\rm Reg}}\bigr)^t = \mathbf{e}^{{\rm Reg}}$, $\mathbf{e}^{{\rm Reg}}$ is the integrand of $H^{{\rm Reg}}$, and
\begin{align*}
\bigl(J^{{\rm Reg}}\bigr)^a &{}= N^2 \bar{\mu}^{-1}_q \bigl(\bigl(p' \bar{\mu}_q q^{ab} \ptl_b \gamma'\bigr) + {\rm e}^{4\gamma} \mbo{r}' \bigl({\rm e}^{-4\gamma} \bar{\mu}_q q^{ab} \ptl_b \omega'\bigr) \bigr) + \gamma' \mathcal{L}_{N'} \bigl(4N \bar{\mu}_q q^{ab} \ptl_b \gamma\bigr) \notag\\
&\quad{}+ \omega' \mathcal{L}_{N'} \bigl(N {\rm e}^{-4\gamma} \bar{\mu}_q q^{ab} \ptl_b \omega\bigr) + 2 \mathcal{L}_{N'} N \bigl(\bar{\mu}_q q^{ab} \ptl_b \mbo{\nu}'\bigr) + 2 \mathcal{L}_{X'} \mbo{\nu}' \bar{\mu}_q q^{ab} \ptl_b N \notag\\
&\quad{}- 2 X^a \bigl(\bar{\mu}_q q^{bc} \ptl_b \mbo{\nu}' \ptl_c N\bigr) + 2 N q_0^{ac} \varrho'^{b}_c {\rm e}^{-2\nu} \ptl_b N' + \bigl(N \ptl_b N' - N' \ptl_ b N\bigr) \tau' \bar{\mu}_q q^{ab} \notag\\
&\quad{} -2N' q_0^{ac} \varrho'^{b}_c {\rm e}^{-2\gamma} \ptl_b N
\end{align*}
after the imposition of the linearly perturbed constraints.
 Following the use of the Carter--Robinson identities for the sub-extremal case, there was still the lingering question of why do these transformations magically solve the issues of the ergo-region and the positivity and conservation of energy in the stability theory of Kerr and Kerr--Newman black holes, even in the axially symmetric case, which, as we just discussed, is nontrivial because the energy density can in principle be locally negative. Our results demonstrate that the reason for these transformations holding in such a way that one can obtain positivity and yet preserve the `symplectic structure', is not `by fluke', but that there are well-defined geometric and variational underlying structures, namely, the covariant (in target) nature of the dimensionally reduced system, the negative curvature of the target and the linearization stability methods. In the context of the black hole uniqueness theorems, the associated generalizations of Carter--Robinson identities were constructed by Bunting and Mazur.\footnote{It is indeed remarkable that the rather ingenious identity of Carter (Robinson for the Einstein--Maxwell case) that was seemingly constructed from `trial and error', would later have a natural geometric interpretation. These results also hold for higher-dimensional black holes and have been used for black hole uniqueness theorems in higher dimensions (see \cite{Hollands-Ishibashi_12}).} We adapted these results for our present problem of (dynamical) black hole stability.

It may be noted that the analogous transformations also resolve the positivity problem for the energy of axially symmetric Maxwell's equations if one does the dimensional reduction to introduce the twist potentials $\lambda$, $\eta$ corresponding to the $\mathfrak{E}$ and $\mathfrak{B}$ fields (cf.\ \cite{GM17_gentitle}).

In the axially symmetric case, even though the original Maxwell equations are linear, a~nonlinear transformation
is used to reduce the $3 + 1$ Einstein--Maxwell equations to an Einstein-wave map system \cite{kaluz2}, which introduces nonlinear coupling within the Maxwell `twist' fields. However, if we turn off the background $\mathfrak{E}$ and $\mathfrak{B}$ fields (e.g., restrict attention to the Kerr metric), the Maxwell equations in twist potential variables reduce to \emph{linear} hyperbolic PDE.

Somewhat interestingly, it appears that the construction of a positive-definite energy for the axially symmetric Maxwell equations on Kerr black hole spacetimes does not easily follow from the Carter's identity, but can be realized a special case of the full Robinson's identity.

In \cite{NG_17_2}, a positive-definite Hamiltonian energy functional for axially symmetric Maxwell equations propagating on Kerr--de Sitter black hole spacetimes was constructed, using modified Einstein-wave maps for the Lorentzian Einstein manifolds with one rotational isometry.

In this work, we shall extend this result and construct a positive-definite energy in a way that is gauge-invariant on the target manifold $(\mathbb{N}, h)$. As we shall see, this is based on negative curvature of the target manifold $(\mathbb{H}^2, h)$ and the convexity of $2+1$ wave maps. The construction of an energy-functional based on the convexity of wave maps, together with our application of the linearization stability methods, suggests why the positivity of our (global) energy for the perturbative theory is to be expected in general, not relying on the insightful and elaborate identity of Carter, which relies on a specific gauge on the target. In such a formulation, the intrinsic geometry within the $2 + 1$ Einstein wave map system becomes more transparent.

 In the context of black hole uniqueness theorems, extensions along these lines, from the initial Carter--Robinson results, were done by Bunting \cite{Bunt_83} and Mazur (\cite{Maz_00} and references therein). In~the~mathematics literature, convexity of harmonic maps for axially symmetric (Brill) initial data was established by Schoen--Zhou \cite{schoen-zhou_13}, which is often used in geometric inequalities between the area of the horizon, angular-momentum and the mass.

 In general, due to the geometric nature of the construction, the linearization stability machinery provides a robust mechanism to deal with the stability problems of black holes within a~symmetry class, including the initial value problem on hypersurfaces that intersect null infinity. The linearization stability machinery is also equipped for dealing with projections from higher-dimensional $(n+1$, $n>3)$ black holes with suitable symmetries \big(toroidal $\mathbb{T}^{n-2}$ spacelike symmetries\big), including 5D Myers--Perry black holes, the stability of which is the main open problem in the stability of higher-dimensional black holes (see, e.g., \cite{EmRe_08}). Indeed, most of our current work, especially the local aspects, readily extend to perturbations within the aforementioned symmetry class of higher-dimensional black holes (see below). However, we propose to carefully address the global aspects of this problem, based on the methods developed in the current work and \cite{GM17_gentitle}, in a future work.

We would like to point out there are related and independent works, based on the `canonical energy' of Hollands--Wald \cite{WH_13}. In \cite{WP_13}, Prabhu--Wald have extended \cite{WH_13} by associating the axisymmetric stability to the existence of a positive-definite `canonical energy'. A positive-definite energy functional was constructed by Dafermos--Holzegel--Rodnianski \cite{HDR_16} in the context of their proof of linear stability of Schwarzschild black holes (see also \cite{GH_16}). Subsequently, a~positive-definite energy was constructed by Prabhu--Wald using the canonical energy methods, that is consistent with both \cite{HDR_16} and \cite{Moncrief_74}. Their approach is based on the construction of metric perturbations using the Teukolsky variable as the Hertz potential.

From a PDE perspective, a suitable notion of (positive-definite) energy is crucial to control the dynamics of a given system of PDEs. In case the scaling symmetries of a nonlinear hyperbolic PDE and its corresponding energy match, powerful techniques come into play that characterize blow up (concentration) and scattering categorically. This problem is referred to as `energy critical'. In the context of $2 + 1$ critical flat-space wave maps
\[
U \colon\ \mathbb{R}^{2+1} \to (\mathbb{N}, h),
\]
the fact that this characterization can be made was demonstrated in the Landmark works \mbox{\cite{chris_tah1,krieg_schlag_ccwm,jal_tah1, jal_tah,sterb_tata_long,sterb_tata_main, struwe_sswm,struwe_equi,tao_all-I,tao_all-II,tao_all-III,tao_all-IV,tao_all-V,tao_all-VI,tao_all-VII}} in the analysis of geometric wave equations. It may be noted that $3 + 1$ Einstein’s equations with one translational isometry can be reduced to the $2 + 1$ Einstein-wave map system \eqref{ewm-system}. In this case, the notion of a positive-definite, gauge-invariant Hamiltonian mass-energy is provided by Ashtekar--Varadarajan \cite{ash_var} (see also Thorne's C-energy \cite{thorne_cenergy})
\begin{align*}
q_{ab} = r^{-m_{\rm AV}} \bigl(\delta_{ab} + \mathcal{O}\bigl(r^{-1}\bigr)\bigr)
\end{align*}
in the asymptotic region of asymptotically flat $(\Sigma, q)$. In a previous work \cite{diss_13}, it was noted that the aforementioned fundamental results on flat space wave maps can be extended to the~${2 + 1}$~Einstein-wave map system resulting from $3 + 1$ Einstein’s equations with translational symmetry using the AV-mass, which in turn is related to the energy of $2 + 1$ wave maps arising from the energy-momentum tensor and a local conservation law in the equivariant case. We would like to point out that, even though the dimensional reduction of $3 + 1$ dimensional axially symmetric, asymptotically flat spacetimes
results in the same $2 + 1$ Einstein-wave map system locally, the axisymmetric problem is not a (geometric) mass-energy-critical problem \cite{NG17}. There is yet another dimensional reduction, based on the ‘time-translational’ Killing vector of stationary class of spacetimes, to which the Kerr metric also belongs, that results in harmonic maps. This distinction between each of the three cases, which is relevant for the applicable methods therein, is explained in~\cite{NG17} for the interested reader.

Without the energy-criticality of the $2 + 1$ Einstein-wave map system, a direct consideration of the nonlinear problem, analogous to \cite{AGS_15, diss_13}, is infeasible. A long-standing approach that is commonly used in the stability problems of Einstein's equations, is to first consider the linear perturbations and hope to control the nonlinear (higher-order perturbations) using the linear perturbation theory.
However, the problem of what is the natural notion of energy for the linear perturbative theory, that is consistent with the dimensional reduction and wave map structure of field equations remains open:
\begin{enumerate}\setlength{\leftskip}{0.03cm}
\item[(P2)] Is there a natural notion of mass-energy for the axially symmetric linear perturbative theory of Kerr black hole spacetimes that is consistent with the dimensional reduction and the wave map structure of the equations?
\end{enumerate}

This question is closely related to whether there exists a natural factor that multiplies the dimensionally reduced Hamiltonian constraint
of the system (compare with the discussion in~\mbox{\cite[pp.~3--4]{NG17}}), which provides a natural notion of energy for our linearized problem.
We point out that the linearization stability methods employed in our works provide a natural mechanism that resolves both (P1) $(\vert a \vert < m)$ and (P2).

Nevertheless, dealing with a plethora of boundary terms that arise in the construction of the positive-definite energy, in connection with the gauge-conditions and the dimensional reduction, is nontrivial.\footnote{This is in contrast with the Maxwell perturbations on Kerr black hole spacetimes, which is a (locally) gauge-invariant problem.} These aspects are dealt with in detail in the latter half of this article, starting from Section~\ref{section5}.

In the current work, after establishing that the constraints for our system are scleronomic,
we prove that our energy functional drives the constrained Hamiltonian dynamics of our system
and that it forms a (spacetime) divergence-free vector field density, after the imposition of the constraints.
In the process of obtaining our results, we construct several variational principles from both Lagrangian and Hamiltonian perspectives, for the nonlinear (exact) and linear perturbative theories. These may be of interest in their own right.

As we already remarked, the black hole stability problem is a very active research area. The decay of Maxwell equations on Schwarzschild was proved in~\cite{PB_08}. The linear stability of Schwarzschild was established in~\cite{HDR_16}. Likewise, the linear stability of Schwarzschild black hole spacetimes using the Cauchy problem for metric coefficients was established in~\cite{HKW_16_2}. These results build on the classic results \cite{Moncrief_74, Regge-Wheeler_57,Zerilli_70,Zerilli_74}.

The important case of linear wave perturbations of Kerr black holes was studied in several fundamental works for small angular momentum \cite{LB_15_1,DR_11,Tato_11}. Likewise, the decay of Maxwell perturbations of Kerr was established in \cite{LB_15_2}. A uniform energy bound and Morawetz estimate for the $ \vert s \vert =1, 2$ Teukolsky equations was established in \cite{SMa_17_1,SMa_17_2}. Boundedness and decay for the $\vert s \vert=2$ Teukolsky equation was established in \cite{HDR_17}. A positive-definite energy for axially symmetric NP-Maxwell scalars was constructed in \cite{NG_19_1}, extending our aforementioned results on Maxwell equations.
The linear stability of Kerr black holes was established in \cite{ABBM_19,HHV_19} for small~$\vert a \vert$. Nonlinear stability of Schwarzschild black holes was announced in \cite{HDRT_21}. Recently, the full proof of nonlinear stability of slowly rotating Kerr black holes was announced in \cite{GKS_22} (see also~\mbox{\cite{GKS_20,KS_21}}).

The effects of the ergo-region become more subtle for rapidly rotating (but $\vert a \vert < M$) Kerr black holes. The decay of the scalar wave for fixed azimuthal modes was established in \cite{FKSY_06,FKSY_08_E,FKSY_08} using spectral methods. The decay of a general linear wave equation was established in the remarkable work \cite{DRS_16}. We would like to point out that the global behaviour, especially the decay estimates, of Maxwell and linearized Einstein perturbations of Kerr black holes, is relatively less understood for the large, but sub-extremal $(\vert a \vert < M)$ case. We expect that our work will be useful in this regard.

Let us briefly summarize the layout of the article. In Section~\ref{section2}, we discuss the Hamiltonian formulation of axially symmetric spacetimes. We start with the Hamiltonian reduction of the wave map Lagrangian and proceed to perform Hamiltonian dimensional reduction of $3+1$ axisymmetric spacetimes into $(2 + 1)$-dimensional Einstein-wave map system. In Section~\ref{section3}, we construct the field equations for linear perturbations of the Einstein equations in axially symmetric spacetimes. These Hamilton field equations are also constructed in the dimensionally reduced Einstein-wave map formalism. Section~\ref{section4} contains the construction of a positive-definite Hamiltonian energy functional for linearized perturbations of the Einstein equations in axially symmetric spacetimes. The construction uses linearization stability methods for the $(2 + 1)$-dimensional Einstein-wave map system. Section~\ref{section4} also contains the proof that the linearized constraint equations are preserved in time for this system.

For stability purposes, it is crucial to establish that this positive-definite Hamiltonian energy is strictly conserved in time. The remainder of this article is dedicated to this purpose. The~strict conservation of the positive-definite-definite energy is closely connected to the boundary behaviour of the fields. Indeed, the boundary behaviour of the fields is critical for the whole construction of this work. This is discussed in detail in Section~\ref{section5}. In Section~\ref{section6}, we formulate the Cauchy problem of linearized Einstein equations in harmonic coordinates and discuss the propagation of regularity and causal structure.

In Section~\ref{section7}, we transform the phase-space variables and the Lagrange multipliers into the Weyl--Papapetrou gauge and discuss their boundary behaviour. Finally, in Section~\ref{section8}, we use this construction to establish that the integrated fluxes of the energy vanish at the horizon, the axes and at the spatial infinity, thereby establishing the strict conservation of the positive-definite energy we constructed earlier. In establishing these results, we classify these boundary flux terms into the dynamical terms, the kinematic terms, and the conformal terms.
This article is a merger of the articles \cite{NG_19_2} and \cite{NG_21_1}.

\section{A Hamiltonian formalism for axially symmetric spacetimes}\label{section2}
Recall that $\bar{M} = \olin{\Sigma} \times \mathbb{R} $ is a $3+1$ Lorentzian spacetime, such that the rotational vector field $\Phi$ acts on $\olin{\Sigma}$ as an isometry with the fixed point set $\Gamma$. In the case of Kerr black hole spacetime, $\Gamma$ is a union of two disjoint sets (the `axes'). It follows that the quotient $\Sigma := \olin{\Sigma}/{\rm SO}(2)$ and $M := \Sigma \times \mathbb{R}$ are manifolds with boundary $\Gamma$.

Consider the Einstein--Hilbert action on $\bigl(\bar{M}, \bar{g}\bigr)$
\begin{align} \label{EH}
S_{\text{EH}} := \int \bar{R}_{\bar{g}} \, \bar{\mu}_{\bar{g}}.
\end{align}
Suppose the axially symmetric $\bigl(\bar{M}, \bar{g}\bigr)$ is a critical point of \eqref{EH}. In the Weyl--Papapetrou coordinates,
\begin{align*}
\bar{g} = \vert \Phi \vert^{-1} \bar{g} + \vert \Phi \vert ({\rm d}\phi + A_\nu {\rm d}x^\nu)^2,
\end{align*}
$ \vert \Phi \vert$ is the norm squared of the Killing vector $ \Phi := \ptl_\phi$, and $g$ is the metric on the quotient $M := \bar{M}/{\rm SO}(2)$. Suppose $II$ is the second fundamental form of the embedding $(M, \tilde{g}) \hookrightarrow \bigl(\bar{M}, \bar{g}\bigr)$, $\tilde{g} = \vert \Phi \vert^{-1} g$, then following the Gauss--Kodazzi equations and the conformal transformation
\begin{align*}
\tilde{R}_{\tilde{g}} = \vert \Phi \vert^{-1} \bigl(R_g -4 g^{\mu \nu} \grad_\mu \grad_\nu \log \vert \Phi \vert^{1/2} - 2 g^{\mu \nu} \grad_\mu \log \vert \Phi \vert^{1/2} \grad_\nu \log \vert \Phi \vert^{1/2}\bigr).
\end{align*}
The Einstein--Hilbert action \eqref{EH} can be reduced to
\begin{align}\label{EWM-Lag}
L_{{\rm EWM}} := \halb \int \biggl(\frac{1}{\kappa} R_g - h_{AB} (U) g^{\a \b} \ptl_\a U^A \ptl_\b U^B \biggr) \bar{\mu}_g
\end{align}
for $\kappa =2$ and $U$ is a wave map
\begin{align*}
U \colon\ (M, g) \to (\mathbb{N}, h)
\end{align*}
to a hyperbolic 2-plane target $(\mathbb{N}, h)$, whose components are associated to the norm and the twist (potential) of the Killing vector. The tangent bundle of the configuration space of \eqref{EWM-Lag} is~now
\begin{align*}
C_{{\rm EWM}} := \bigl\{ (g, \dot{g}), \bigl(U^A, \dot{U}^A\bigr) \bigr\},
\end{align*}
where the dot \big(e.g., $\dot{U}$\big) denotes derivative with respect to a time-coordinate function $t$.
We would like to perform the Hamiltonian reduction of the system \eqref{EWM-Lag}. Recall the ADM decomposition of $(M, g) = (\Sigma, q) \times \mathbb{R}$
\begin{align*}
g = -N^2 {\rm d}t^2 + q_{ab} ({\rm d}x^a + N^a {\rm d}t) \otimes \bigl({\rm d}x^b + N^b {\rm d}t\bigr)
\end{align*}

 Let us split the geometric part and the wave map part of the variational principle \eqref{EWM-Lag} as $L_{{\rm EWM}} = L_{{\rm geom}} + L_{{\rm WM}}$. 
 Let us now start with the Hamiltonian reduction of the wave map Lagrangian $L_{{\rm WM}}$,
 \begin{align}\label{WM-lag}
L_{{\rm WM}} := -\halb\int \bigl(h_{AB} (U) g^{\a \b} \ptl_\a U^A \ptl_\b U^B\bigr) \bar{\mu}_g
\end{align}
 over the tangent bundle of the configuration space of wave maps, $C_{{\rm WM}} = \bigl\{ \bigl(U^A, \dot{U}^A\bigr) \bigr\}$.

Suppose we denote the Lagrangian density of \eqref{WM-lag} as $\mathcal{L}$ and conjugate momenta as $p_A$, we~have
\begin{align*}
p_A = \frac{1}{N} \bar{\mu}_q h_{AB}(U) \ptl_t U^B - \frac{1}{N} \bar{\mu}_q h_{AB} (U) \mathcal{L}_N U^B,
\end{align*}
where $\mathcal{L}_N$ is the Lie derivative with respect to the shift $N^a$. As a consequence, we have
\begin{align*}
h_{AB} (U)\ptl_t U^B &{}= \frac{1}{\bar{\mu}_q} N p_A + h_{AB} (U) \mathcal{L}_N U^B,
\end{align*}
and the Lagrangian density $\mathcal{L}_{{\rm WM}}$ can be expressed in terms of the wave map phase space $X_{{\rm WM}} := \bigl\{ \bigl(U^A, p_A\bigr) \bigr\}$ as
\begin{align*}
\mathcal{L}_{{\rm WM}} &{}= \halb p_B \ptl_t U^B - \halb p_B \mathcal{L}_N U^B - \halb N \bar{\mu}_q h_{AB} (U) q^{ab} \ptl_a U^a \ptl_b U^B.
\end{align*}
Let us now define the Hamiltonian density as follows:
\begin{align*}
\mathcal{H}_{{\rm WM}} := \halb p_B \ptl_t U^B + \halb p_B \mathcal{L}_N U^B + \halb N \bar{\mu}_q h_{AB} (U) q^{ab} \ptl_a U^A \ptl_b U^B.
\end{align*}
As a consequence, we formulate the ADM variational principle for the Hamiltonian dynamics of the wave map phase space $X_{{\rm WM}}$ as
\begin{align*}
L_{{\rm WM}} [X_{{\rm WM}}] := \halb \int \bigl(p_A \ptl_t U^A - p_B \mathcal{L}_N U^B - N \bar{\mu}_q h_{AB} (U) q^{ab} \ptl_a U^a \ptl_b U^B\bigr){\rm d}^3x,
\end{align*}
which has the field equations
\begin{align} \label{u-dot}
p_A = \frac{1}{N} \bar{\mu}_q h_{AB} (U) \ptl_t U^B - \frac{1}{N} \bar{\mu}_q h_{AB} (U) \mathcal{L}_N U^B
\end{align}
and the critical point with respect to \big(the first variation $ D_{U^A} \cdot L_{{\rm WM}}=0$\big) $U^a$ gives
\begin{align}
\ptl_t p_A&{}= -N \bar{\mu}^{-1}_q \frac{\ptl}{\ptl U^A}h^{BC} p_B p_C + h_{AB} \ptl_a\bigl(N \bar{\mu}_q q^{ab} \ptl_b U^B\bigr) \notag\\
&\quad{}+ N\bar{\mu}_q h_{AB} \leftexp{(h)}{\Gamma}^B_{CD}(U)q^{ab}\ptl_a U^C \ptl_b U^D
 + \mathcal{L}_N p_A,\label{p-dot}
\end{align}
where $\leftexp{(h)}{\Gamma}$ are the Christoffel symbols,
\[
\leftexp{(h)}{\Gamma}^A_{BC} := \halb h^{AD}(U) (\ptl_C h_{BD} + \ptl_B h_{DC} - \ptl_D h_{BC}).
\]
It is straight-forward to verify that the canonical equations
\begin{align*}
D_{p_A}\cdot H_{{\rm WM}} = \ptl_t U^A \qquad \text{and} \qquad D_{U^A}\cdot H_{{\rm WM}} = - \ptl_t p_A
\end{align*}
correspond to \eqref{u-dot} and \eqref{p-dot}, respectively,
where $H_{{\rm WM}} := \int \mathcal{H}_{{\rm WM}} {\rm d}^2x $ is the (total) Hamiltonian. Subsequently, if we use the Gauss--Kodazzi equation for the ADM $2+1$ decomposition and defining the (geometric) phase space
\begin{align*}
X_{{\rm geom}} := \bigl\{ \bigl(q_{ab}, \mbo{\pi}^{ab}\bigr) \bigr\},
\end{align*}
 we can represent the gravitational Lagrangian density as follows:
\begin{align*}
\mathcal{L}^{{\rm Alt}}_{{\rm geom}} := \bigl(- q_{ab}\ptl_t \pi^{ab} - N H_{{\rm geom}} - N_a H_{{\rm geom}}^a\bigr),
\end{align*}
where
\begin{align*}
&H_{{\rm geom}}:= \bar{\mu}^{-1}_q \bigl(\Vert \mbo{\pi} \Vert_q^2 - {\rm Tr}_q(\mbo{\pi})^2 \bigr) - \bar{\mu}_q R_q, \\
&H_{{\rm geom}}^a := -2\, \leftexp{(q)}{\grad}_b \mbo{\pi}^{ab}.
\end{align*}

It may be noted that the conjugate momentum tensor of the reduced metric and the corresponding components of the conjugate momentum of the $3$ metric are related as follows:
\begin{align*}
 \vert \Phi \vert \bar{\mbo{\pi}}^{ab} = \mbo{\pi}^{ab}, \qquad \text{where} \ \mbo{\pi}^{ab} = \bar{\mu}_q \bigl(q^{ab} {\rm Tr}_q(K) - K^{ab} \bigr)
\end{align*}
Consequently, we have the variational principle for Hamiltonian dynamics of the reduced Einstein wave map system
\begin{align} \label{ham-var}
J_{{\rm EWM}} := \int^{t_2}_{t_1}\int_{\Sigma} \bigl(\mbo{\pi}^{ab} \ptl_t q_{ab} + p_A \ptl_t U^A - N H - N^a H_a \bigr) {\rm d}^2x {\rm d}t,
\end{align}
where now the reduced $H$ and $H_a$ are
\begin{align*}
&H= \bar{\mu}^{-1}_q \biggl(\bigl(\Vert \mbo{\pi} \Vert^2_q - {\rm Tr}_q (\mbo{\pi})^2 \bigr) + \halb p_A p^A \biggr) + \bar{\mu}_q \biggl(- R_q + \halb h_{AB} q^{ab} \ptl_a U^A \ptl_b U^B \biggr), \\
&H_a= - 2\, \leftexp{(q)}{\grad}_b \mbo{\pi}^b_a + p_A \ptl_a U^A.
\end{align*}
Therefore, we have proved the following theorem.

\begin{Theorem}
Suppose $\bigl(\bar{M}, \bar{g}\bigr)$ is an axially symmetric, Ricci-flat, globally hyperbolic Lorentzian spacetime and that $\bar{g}$ admits the decomposition \eqref{KK-ADM} in a local coordinate system, then the dimensionally reduced field equations in the interior of $M = \Sigma \times \mathbb{R}$, where $\Sigma = \olin{\Sigma}/{\rm SO}(2)$, are derivable from the variational principle \eqref{ham-var} for the reduced phase space
\begin{align*}
X_{{\rm EWM}} := \bigl\{ \bigl(q_{ab}, \mbo{\pi}^{ab}\bigr), \bigl(U^A, p_A\bigr) \bigr\}
\end{align*}
with the Lagrange multipliers $\{ N, N^a \}$.
\end{Theorem}

As a consequence, we have the field equations for Hamiltonian dynamics in $X_{{\rm EWM}}$
\begin{subequations}
\begin{gather}
h_{AB} \ptl_t U^B = N \bar{\mu}^{-1}_q p_A + h_{AB} \mathcal{L}_N U^B, \\
\ptl_t p_A= - N \bar{\mu}^{-1}_q\frac{\ptl}{\ptl U^A}h^{BC} (U) p_B p_C + h_{AB} \ptl_a\bigl(N \bar{\mu}_q q^{ab} \ptl_b U^B\bigr) \notag\\
 \hphantom{\ptl_t p_A=}{}
 + N \bar{\mu}_q h_{AB} \leftexp{(h)}{\Gamma}^B_{CD}(U)q^{ab}\ptl_a U^C \ptl_b U^D + \mathcal{L}_N p_A,\\
\ptl_t q_{ab} = 2 N \bar{\mu}^{-1}_q (\mbo{\pi}_{ab} - q_{ab} {\rm Tr} (\mbo{\pi})) + \leftexp{(q)}{\grad}_a N_b + \leftexp{(q)}{\grad}_b N_a, \label{qdot}\\
\ptl_t \mbo{\pi}^{ab}= \halb N \bar{\mu}^{-1}_q q^{ab} \bigl(\Vert \mbo{\pi} \Vert^2_q - {\rm Tr}(\mbo{\pi})^2\bigr)- 2N \bar{\mu}^{-1}_q \bigl( \mbo{\pi}^{ac} \mbo{\pi}^{b}_c - \mbo{\pi}^{ab} {\rm Tr}(\mbo{\pi}) \bigr) \notag\\
\hphantom{\ptl_t \mbo{\pi}^{ab}=}{}
 +\bar{\mu}_q \bigl(\leftexp{(q)}{\grad}^b \,\leftexp{(q)}{\grad}^a N -q^{ab}\, \leftexp{(q)}{\grad}_c \leftexp{(q)}{\grad}^c N\bigr), \notag\\
\hphantom{\ptl_t \mbo{\pi}^{ab}=}{}
 + \leftexp{(q)}{\grad}_c\bigl(\mbo{\pi}^{ab} N^c\bigr) - \leftexp{(q)}{\grad}_c N^a \mbo{\pi}^{cb} - \leftexp{(q)}{\grad}_c N^b \mbo{\pi}^{ca}\notag\\
\hphantom{\ptl_t \mbo{\pi}^{ab}=}{}
 + \frac{1}{4}\bar{\mu}_q^{-1}N q^{ab} p_A p^A + \halb N \bar{\mu}_q h_{AB} \biggl(q^{ac}q^{bd} - \halb q^{ab} q^{cd}\biggr)\ptl_c U^A \ptl_d U^B
\end{gather}
\end{subequations}
and the constraint equations
\begin{align}\label{constraints}
H=0, \qquad H_a =0.
\end{align}
It should be pointed out that, analogous to original ADM formulation~\cite{ADM_62}, we have made a~simplification with the coupling constant (see also the discussion in~\cite[pp.~520--521]{MTW}). In case the precise coupling between the $2 + 1$ Einstein's equations and its wave map source is relevant, the original coupling can be reinstated by simply substituting the following formulas throughout our~work:
\begin{align*}
&\mbo{\pi}_{{\rm true}} := \frac{1}{2\kappa} \mbo{\pi}= \frac{1}{2\kappa} \bar{\mu}_q \bigl(q^{ab} {\rm Tr}(K) - K^{ab} \bigr), \\
&H_{{\rm true}} := \frac{1}{2\kappa} H_{{\rm geom}} = \bar{\mu}^{-1}_q \frac{1}{2\kappa} \bigl(\Vert \mbo{\pi} \Vert_q^2 - {\rm Tr}(\mbo{\pi})^2 \bigr) - \frac{1}{2\kappa} \bar{\mu}_q R_q, \\
&(H_{{\rm true}})_a := \frac{1}{2 \kappa} (H_{{\rm geom}})_a = -\frac{1}{\kappa} \leftexp{(q)}{\grad}_b \mbo{\pi}^{b}_a.
\end{align*}
We would like to remind the reader that, in the dimensional reduction process, we introduce the closed 1-form $G$ such that
\begin{align*}
\vert \Phi \vert^{-2} \eps_{\mu \nu \b} g^{\b \a} G_\a = F_{\mu \nu},
\end{align*}
where $F = {\rm d}A$. In our simply connected domain, $G= {\rm d}w$, where $\omega$ is the gravitational twist potential and one of the components of the wave map $U$.

\subsection*{Nonlinear conservation laws}
Following Komar's definition of angular momentum,
\begin{align*}
J = \frac{1}{16\pi} \int_{\Sigma} \star {\rm d}\Phi, 
\end{align*}
it follows that for the Kerr metric $J = a M$. In view of the well-known fact that the angular momentum is conserved for our vacuum axisymmetric problem, without effective loss of generality, we shall assume that the perturbation of the angular-momentum is zero.

The dimensional reduction provides additional structure for the original field equations. As~noted by Geroch \cite{Geroch_71}, the Lie group ${\rm SL}(2, \mathbb{R})$ acts on the resulting target $(\mathbb{N}, h)$ in the dimensional reduction procedure. The M\"obius transformations, which are the isometries of~$(\mathbb{N}, h)$ provide us a~Poisson algebra of nonlinear conserved quantities.

\begin{Corollary}
Suppose $U\colon \Sigma \times (t_1, t_2) \to (\mathbb{N}, h)$ is the wave map coupled to $2 + 1$ Einstein equations as above,
then there exist $($spacetime$)$ divergence-free vector fields $J_i$, $i =1, 2, 3$, such that if $C_i$ is the flux of $J_i$ at $\Sigma_t$, $t \in (t_1, t_2)$ hypersurface,
\begin{align*}
\{C_i, C_j \} = \sigma^{k}_{ij} C_k, \qquad i\neq j \neq k,
\end{align*}
where $\{ \cdot, \cdot \}$ is the Poisson bracket in the phase space $X_{{\rm EWM}}$ and \smash{$\sigma^k_{ij}$} are the structure constants of the $($M\"obius$)$ isometries $\{ K_i, K_2, K_3\}$ of $(\mathbb{N}, h)$.
\end{Corollary}
\begin{proof}
The M\"obius transformations on the target $(\mathbb{N}, h)$, the hyperbolic 2-pane, are isometries corresponding to translation, dilation and inversion $\{ K_1, K_2, K_3 \}$. It follows that
\begin{align*}
\leftexp{(h)}{\grad}_A (K_i)_B + \leftexp{(h)}{\grad}_B (K_i)_A =0, \qquad \forall \ i = 1, 2, 3,
\end{align*}
Consider the quantity
\begin{align*}
\ptl_t \bigl(K^A_i p_A\bigr) &{}= \ptl_C K_i^A p_A \bigl(N \bar{\mu}_q p^C + \mathcal{L}_N U^C\bigr) \notag\\
&\quad{}+ K_i^A \bigl(-N \bar{\mu}^{-1}_q \ptl_A h^{BC} p_B p_C
+ h_{AB} \ptl_a \bigl(N \bar{\mu}_q q^{ab} \ptl_b U^B\bigr) \notag\\
&\quad\hphantom{+ K_i^A \bigl(\ }{}+ N \bar{\mu}_q h_{AB} \leftexp{(h)}{\Gamma}^B_{CD} q^{ab} \ptl_a U^C \ptl_b U^D + \mathcal{L}_N p_A \bigr).
\end{align*}
Now consider,
\begin{align*}
\ptl_a \bigl(N \bar{\mu}_q q^{ab} \ptl_b U^B K_i^A h_{AB}\bigr) &{}= K_i^A h_{AB} \ptl_a \bigl(N \bar{\mu}_q q^{ab} \ptl_b U^B\bigr) \notag\\
&\quad{}+ N \bar{\mu}_q q ^{ab} \ptl_a U^C \ptl_c K_i^A h_{AB} \ptl_b U^B\!
+ N \bar{\mu}_q q^{ab} \ptl_b U^B K_i^A \ptl_C h_{AB} \ptl_a U^C
\end{align*}
and note that
\begin{gather*}
 -N \bar{\mu}_q q ^{ab} \ptl_a U^C \ptl_c K^A h_{AB} \ptl_b U^B - N \bar{\mu}_q q^{ab} \ptl_b U^B K^A \ptl_C h_{AB} \ptl_a U^C \\
\quad{} + N \bar{\mu}_q h_{AB} \leftexp{(h)}{\Gamma}^B_{CD} q^{ab} \ptl_a U^C \ptl_b U^D K_i^A \\
\qquad{}= -\halb N \bar{\mu}_q q^{ab} \ptl_a U^C \ptl_b U^D \ptl_A h_{CD}- N \bar{\mu}_q^{ab} \ptl_a U^C \ptl_c U^B \ptl_C K_i^A h_{AB} = 0
\end{gather*}
after relabeling of indices and on account of the fact that the deformation tensor of $K_i$ in the target $h$ is zero. Let us now define
\begin{align*}
&(J_i)^t = K^A_i p_A,\\
&(J_i)^a = \ptl_a \bigl(N \bar{\mu}_q q^{ab} \ptl_b U^B K_i^A h_{AB} + N^a K_i^A p_A\bigr).
\end{align*}
It follows from above that each $J_i$, $i = 1, 2, 3$, is a spacetime divergence-free vector density. Now~then
\begin{align*}
C_i := \int_{\Sigma_t}K_i^A p_A {\rm d}^2 x
\end{align*}
and consider the Poisson bracket
\begin{align*}
\{ C_i, C_j \} &{}= \biggl\{ \int_{\Sigma_t} K^A_i p_A \,,\, \int_{\Sigma_t} K^A_j p_A \biggr\}
= \int_{\Sigma_t} \bigl(\ptl_{U^A} K^B_i K^A_j - K^A_i \ptl_{U^A} K^B_j \bigr) p_B \\
&{}= \int_{\Sigma_t} \bigl[K_i, K_j \bigr]^A p_A = \int_{\Sigma_t} \sigma^k_{ij} K^A_k p_A
= \sigma^k_{ij} C_k, \qquad i \neq j \neq k.
\end{align*}
This result generalizes the equivalent result in \cite{kaluz1}, where each $\Sigma$ is $\mathbb{S}^2$, to our non-compact case and a general gauge on the target metric $h$. We would like to point out that these conservation laws are closely related to the `moment maps' associated to the M\"obius transformations in the phase space. It may be noted that our arguments readily extend to the $(n+1)$ higher-dimensional context, where the target is ${\rm SL}(n-2)/{\rm SO}(n-2)$.
\end{proof}

In the Weyl--Papapetrou coordinates, define a quantity $\mbo{\nu}$ such that
\begin{align} \label{conformal}
q= {\rm e}^{2\mbo{\nu}} q_0,
\end{align}
where $q_0$ is the flat metric and (mean curvature) scalar $\mbo{\tau} := \bar{\mu}^{-1}_q q_{ab} \mbo{\pi}^{ab}$. In our work, it will be convenient to split $2$-tensors into a trace part and the conformal Killing operator \cite{kaluz1}. In the following lemma, we shall streamline the related discussion and results obtained in \cite{kaluz1}.

\begin{Lemma}
Suppose the phase space variables $ \bigl\{ (q, \mbo{\pi}), \bigl(U^A, p_A\bigr) \bigr\} \in X$ are smooth in the interior of $\Sigma$, we have
\begin{enumerate}\itemsep=0pt
\item[$(1)$] Suppose a vector field $Y \in T (\Sigma)$, then conformal Killing operator defined as
\begin{align*}
({\rm CK}(Y, q))^{ab}:= \bar{\mu}_q \bigl(\leftexp{(q)}{\grad}^b Y^a + \leftexp{(q)}{\grad}^a Y^b - q^{ab} \, \leftexp{(q)}{\grad}_c Y^c\bigr),
\end{align*}
is invariant under a conformal transformation, i.e., ${\rm CK} (Y,q) = {\rm CK}(Y, q_0)$.
\item[$(2)$] There exists a vector field $Y \in T(\Sigma)$, which is determined uniquely up a conformal Killing vector, such that
\begin{align*}
\mbo{\pi}^{ab}= {\rm e}^{-2 \mbo{\nu}} ({\rm CK}(Y, q))^{ab} + \halb \mbo{\tau} \bar{\mu}_q q^{ab}.
\end{align*}
\item[$(3)$] If $ \bigl\{ (q, \mbo{\pi}), \bigl(U^A, p_A\bigr) \bigr\} \in X$ satisfy the constraint equations, the Hamilton and momentum constraint equations \eqref{constraints} can be represented as the elliptic equations
\begin{align*}
\bar{\mu}^{-1}_{q_0} \biggl({\rm e}^{-2 \mbo{\nu}} \Vert \varrho \Vert^2_{q_0} - \halb \tau^2 {\rm e}^{2 \mbo{\nu}} \bar{\mu}^2_{q_0} + \halb p_A p^A\biggr)
+ \bar{\mu}_{q_0} \bigl(2 \Delta_0 \mbo{\nu} + h_{AB} q_0^{ab} \ptl_a U^A \ptl_b U^B\bigr) =0
\end{align*}
and
\[
-\leftexp{(q_0)}{\grad}_b \varrho^b_a - \halb \ptl_a \mbo{\tau} {\rm e}^{2 \mbo{\nu}} \bar{\mu}_{q_0} + \halb p_A \ptl_a U^A=0, \qquad (\Sigma, q_0),
\]
respectively, where
\begin{align}\label{varrho-def}
\varrho^a_c= \bar{\mu}_{q_0} \bigl(\leftexp{(q_0)}{\grad}_c Y^a + \leftexp{(q_0)}{\grad}^a Y_c - \delta^a_c \, \leftexp{(q_0)}{\grad}_b Y^b\bigr).
\end{align}
\end{enumerate}
\end{Lemma}
\begin{proof}
Part (1) follows from the definitions and direct computations. Part (2) is based on the fact that the transverse-traceless tensors vanish for our form of the 2-metric. Consider the decomposition of $\mbo{\pi}^{ab}$ into a trace part and a traceless part:
\begin{align} \label{split-pi}
\mbo{\pi}^{ab} = \halb \tau \bar{\mu}_q q^{ab} + \cancel{\operatorname{Tr}} \, \mbo{\pi}^{ab}
= \halb \mbo{\tau} \bar{\mu}_q q^{ab} + (\mbo{\pi}_{{\rm TT}})^{ab} + {\rm e}^{-2\nu} {\rm CK} (Y, q),
\end{align}
where $(\mbo{\pi}_{{\rm TT}})^{ab}$ is such that
\begin{align} \label{tt-cond}
q_{ab} (\mbo{\pi}_{{\rm TT}})^{ab} =0 \qquad \text{and} \qquad \leftexp{(q)}{\grad}_a (\mbo{\pi}_{{\rm TT}})^{ab} =0.
\end{align}

The result (2) now follows from the fact that \eqref{tt-cond} is invariant under the conformal transformation \eqref{conformal} and the fact that transverse-traceless tensors vanish on the flat metric $q_0$, with suitable boundary conditions. The existence of $Y$ follows from the following elliptic equation
\begin{align*}
\leftexp{(q)}{\grad}_a \bigl({\rm e}^{-2 \mbo{\nu}} {\rm CK} (Y, q)\bigr) = \grad_a \biggl(\mbo{\pi}^{ab} -\halb \mbo{\tau} \bar{\mu}_q q^{ab}\biggr)
\end{align*}
 and Fredholm theory. It may be noted that the right-hand side is $L^2-$orthogonal to the kernel of the linear, self-adjoint elliptic operator on the left-hand side, which contains the conformal Killing vector fields of $q_0$,
\begin{align} 
&\leftexp{(q)}{\grad}_b \bigl({\rm e}^{-2\mbo{\nu}} \bar{\mu}_q \bigl(Y_a + \leftexp{(q)}{\grad}_a Y^b - \delta^b_a \leftexp{(q)}{\grad}_c Y^c\bigr) \bigr) \notag\\
&\qquad{}= \leftexp{(q_0)}{\grad}_b \bigl(\bar{\mu}_{q_0}\bigl(\leftexp{(q_0)}{\grad}_b Y^a\bigr) + \leftexp{(q_0)}{\grad}_a Y^b - \delta^b_a \leftexp{(q_0)}{\grad}_c Y^c\bigr) . \label{conf-identity}
\end{align}
It would now be convenient to define $\varrho$ as in \eqref{varrho-def}. Now then, using
\begin{align*}
\mbo{\pi}^a_b = \halb \mbo{\tau} {\rm e}^{2 \mbo{\nu}} \bar{\mu}_{q_0} \delta^a_c + {\rm e}^{-2 \mbo{\nu}} \varrho^a_c,
\end{align*}
and \eqref{conf-identity}, the momentum constraint can now be transformed into the following elliptic operator for $Y$ on $(\Sigma, q_0)$
\begin{align*}
H_a = -2\leftexp{(q_0)}{\grad}_b \varrho^b_a - \ptl_a \mbo{\tau} {\rm e}^{2 \mbo{\nu}} \bar{\mu}_{q_0} + p_A \ptl_a U^A, \qquad (\Sigma, q_0), \quad a= 1, 2.
\end{align*}
The scalar curvature $R_q$ of $(\Sigma, q)$ and $\Vert \mbo{\pi} \Vert^2_q$ can be expressed explicitly as
\begin{align*}
R_q = -2 {\rm e}^{-2 \mbo{\nu}} \Delta_0 \mbo{\nu}, \qquad \Vert \mbo{\pi} \Vert^2_q = \halb \mbo{\tau}^2 {\rm e}^{4 \nu} \bar{\mu}^2_{q_0} + \Vert \varrho \Vert^2_{q_0}.
\end{align*}
The Hamiltonian constraint can now be transformed to the elliptic operator
\begin{align*}
H &{}= \bar{\mu}^{-1}_{q_0} \biggl({\rm e}^{-2 \mbo{\nu}} \Vert \varrho \Vert^2_{q_0} - \halb \mbo{\tau}^2 {\rm e}^{2 \mbo{\nu}} \bar{\mu}^2_{q_0} + \halb p_A p^A\biggr) \notag\\
&\quad{}+ \bar{\mu}_{q_0} \biggl(2 \Delta_0 \mbo{\nu} + \halb h_{AB} q_0^{ab} \ptl_a U^A \ptl_b U^B\biggr), \qquad (\Sigma, q_0),
\end{align*}
where
\begin{align*}
\Delta_0 \mbo{\nu} := \frac{1}{\bar{\mu}_{q_0}} \ptl_b\bigl(q^{ab}_0 \bar{\mu}_{q_0} \ptl_b \mbo{\nu}\bigr).
\tag*{\qed}
\end{align*}
\renewcommand{\qed}{}
\end{proof}

The conditions for the dimensional reduction above are modeled along the Kerr metric \eqref{BL-Kerr}.
Let us now consider the corresponding field equations for the Kerr metric. It follows that for the Kerr wave map
\begin{align*}
 U \colon\ (M, g) \to (N, h)
 \end{align*}
 we have $p_1= p_2 \equiv 0$
 and $ \mbo{\pi}^{ab} \equiv 0$. As a consequence, the dimensionally reduced field equations for the Kerr metric \eqref{BL-Kerr} are
 \begin{align}
 \ptl_a \bigl(N \bar{\mu}_q q^{ab} U^A\bigr) + N \bar{\mu}_q \leftexp{(h)}{\Gamma}^{A}_{BC} q^{ab} \ptl_a U^B \ptl_b U^C&{}=0 \label{Kerr-p-dot}
 \end{align}
and
\begin{gather*}
 \bar{\mu}_q \bigl(\leftexp{(q)}{\grad}^b \, \leftexp{(q)}{\grad}^a N - q^{ab} \, \leftexp{(q)}{\grad}_c \leftexp{(q)}{\grad}^c N \bigr) \notag\\
\qquad{} + \halb N \bar{\mu}_q \biggl(q^{ac} q^{bd} - \halb q^{ab} q^{cd}\biggr) h_{AB} \ptl_a U^A \ptl_b U^B=0.
\end{gather*}
The Hamiltonian constraint
\begin{align*}
H = \bar{\mu}_q \biggl(-R_q + \halb h_{AB} q^{ab} \ptl_a U^A \ptl_b U^B\biggr) =0
\end{align*}
for $a$, $b$ and $A, B, C = 1, 2$.
The scalar $\mbo{\tau}$ is the mean curvature of the embedding $\Sigma \hookrightarrow M$, whose evolution is governed by the equation
\begin{align}\label{meancurv}
\ptl_t \mbo{\tau} = - \leftexp{(q)}{\grad}_a \leftexp{(q)}{\grad}^a N + N \bar{\mu}^{-1}_q \biggl(\Vert \mbo{\pi} \Vert_q^2 + \halb p_A p^A\biggr).
\end{align}
Following the notation introduced in \cite{kaluz1}, the evolution equation \eqref{meancurv} can be represented as
\begin{align*}
{\rm e}^{2\mbo{\nu}} \ptl_t \mbo{\tau} = - \Delta_0 N + N \mathfrak{q},
\end{align*}
where
\begin{align*}
\mathfrak{q} := \bar{\mu}^{-1}_q {\rm e}^{-2\mbo{\nu}}\biggl(\Vert \varrho \Vert^2_q + \halb \mbo{\tau}^2 {\rm e}^{4\mbo{\nu}} \bar{\mu}_q + \halb p_A p^A \biggr),
\end{align*}
where we again used the splitting expression \eqref{split-pi}. It follows that for the Kerr metric \eqref{BL-Kerr}
\begin{align}
\mbo{\tau} = \ptl_t \mbo{\tau} \equiv 0, \qquad \varrho \equiv 0 \qquad \text{and} \qquad \Delta_q N =0. \label{kerr-maximal}
\end{align}
The equation \eqref{qdot} can be decomposed as
\[
\ptl_t (\bar{\mu}_q) = N {\rm Tr}_q (\mbo{\pi}) - \halb \bar{\mu}_q q_{ab}\bigl(\leftexp{(q)}{\grad}^a N^b + \leftexp{(q)}{\grad}^b N^a \bigr)
\]
and the evolution of the densitized inverse metric
\[
\ptl_t \bigl(\bar{\mu}_q q^{ab}\bigr) = 2 N \biggl(\mbo{\pi}^{ab} - \halb q^{ab} {\rm Tr}_q(\mbo{\pi})\biggr) + \bar{\mu}_q \bigl(\leftexp{(q)}{\grad}^a N^b + \leftexp{(q)}{\grad}^b N^a - q^{ab} \leftexp{(q)}{\grad}_c N^c \bigr).
\]

\section[A Hamiltonian formalism for axially symmetric metric perturbations]{A Hamiltonian formalism for axially symmetric\\ metric perturbations}\label{section3}

In this section, we shall calculate the field equations and the Lagrangian and Hamiltonian variational principles for linear perturbation equations of the $2 + 1$ Einstein-wave map system. Consider a smooth curve
\begin{align*}
\mbo{\gamma}_s \colon\ [0, 1] \to C_{{\rm EWM}}
\end{align*}
parametrized by $s$ in the tangent bundle of configuration space $C_{{\rm EWM}}$ of the Einstein-wave map system. Like previously, we shall start with the wave map system. Let $U_s \colon (M, g) \to (N, h)$ be a 1-parameter family of maps generated by the flow along $\mbo{\gamma}_s$ such that
\begin{align*}
&U_0 \equiv U, \\
& U_s \equiv U, \qquad \textnormal{outside a compact set $\Omega \subset M$},
\end{align*}
and $ U' := D_{\mbo{\gamma}_s} \cdot U_s \big\vert_{s=0}$, where $U \colon (M, g) \to (N, h)$ is a given (e.g., Kerr) wave map. The~deformations along $\mbo{\gamma}_s$ can be manifested, for instance, by the exponential map $\operatorname{Exp}(s U)$. In the following, with a slight abuse of notation, we shall denote the manifestations of the deformations along $\mbo{\gamma}_s$ for the wave map $U \colon (M,g) \to (N, h)$, by $\mbo{\gamma}_s$ itself. Let us now denote the deformations along $\mbo{\gamma}_s$ of a point at $s=0$ in the tangent bundle of the wave map configuration space $C_{{\rm WM}}$ as~follows:
\begin{align*}
C'_{ {\rm WM}} := \bigl\{
U^{'A} = D_{\mbo{\gamma}_s} \cdot U_s^{A} \big\vert_{s=0}, \, \dot{U}^{'A} = D_{\mbo{\gamma}_s} \cdot \dot{U}_s^{A} (s) \big\vert_{s=0} \bigr\}.
\end{align*}
Now consider the Lagrangian action of wave map
\begin{align} \label{wm-action-again}
L_{{\rm WM}} (C_{{\rm WM}}) = -\halb \int \bigl(g^{\mu \nu}h_{AB} \ptl_\mu U^A \ptl_\nu U^B\bigr) \bar{\mu}_g.
\end{align}
For simplicity, we shall denote $L_{{\rm WM}}(\gamma(s))$ as $L_{{\rm WM}}(s)$. We have
\begin{align}\label{1-var-wm}
D_{\mbo{\gamma}_s} \cdot L_{{\rm WM}}(s) = \int h_{AB} \bigl(\square_g U^A + \leftexp{(h)}{\Gamma}^A_{BC} g^{\mu \nu} \ptl_\mu U^B \ptl_\nu U^C\bigr) U^{'B} \bar{\mu}_g,
\end{align}
where we have used the identity
\begin{align} \label{christof-intro}
& g^{\mu \nu} h_{AB}(U) \ptl_\mu U^A \ptl_\nu U^{'B} + \halb g^{\mu \nu}\ptl_C h_{AB} \ptl_\mu U^A \ptl_\nu U^B U^{'C} \notag\\
&\qquad{} = - h_{AB} U^{'B} \bigl(\square_g U^A + \leftexp{(h)}{\Gamma}^A_{BC} g^{\mu \nu} \ptl_\mu U^B \ptl_\nu U^C\bigr)
\end{align}
modulo boundary terms (see, e.g., \cite[pp.~19--20]{diss_13}). The following geometric construction shall be useful to represent our formulas compactly \cite{Misn_78}. Firstly, let us define the notions of induced tangent bundle and the associated `total' covariant derivative on the target $(\mathbb{N}, h)$, under the wave mapping $U \colon M \to N$. The induced tangent bundle $T_U N$ on $M$ consists of the 2-tuple~$(x, y)$, where $x \in M$ and $y \in T_{U(x)} N$, with the bundle projection
\begin{align*}
P \colon\ T_U N \to M, \qquad
(x, y) \to x.
\end{align*}
Consider the vector field $\dot{V}_s \in TM$, then the image of $\dot{V}_s$ under the wave map $U$ is a vector field \smash{$\dot{V}_s^A = \ptl_s U^A$} in a local coordinate system of $(N, h)$. As a consequence, we can define a covariant derivative on the induced bundle
\begin{align*}
\leftexp{(h)}{\grad}_\mu \dot{V}^A_s := \ptl_\mu \dot{V}_s^A + \leftexp{(h)}{\Gamma}^A_{BC} \dot{V}_s^B \ptl_\mu U^C.
\end{align*}
It may be verified explicitly that the induced connection is metric compatible $\leftexp{(h)}{\grad}_A h^{AB} \equiv 0$. Likewise, for a `mixed' tensor
\begin{align*}
&\Lambda := \Lambda^A_\mu \, \ptl_{x^A} \otimes {\rm d}x^\mu,
\\
&\leftexp{(h)}{\grad}_\nu \Lambda_\mu^A := \leftexp{(g)}{\grad}_\nu \Lambda_\mu^A + \leftexp{(h)}{\Gamma}^A_{BC} \Lambda^B_\mu \ptl_\nu U^C.
\end{align*}
In particular, for $\mbo{e}_B \in TN$, the second covariant derivative
\begin{align*}
\leftexp{(h)}{\grad}_ \mu \leftexp{(h)}{\grad}_\nu \mbo{e}_B = \ptl_\mu \bigl(\Gamma^A_{\nu B} \mbo{e}_A\bigr) - \leftexp{(g)}{\Gamma}^\a_{\mu \nu} \leftexp{(h)}{\Gamma}^A_{\a B} \mbo{e}_A + \leftexp{(h)}{\Gamma}^A_{\mu B} \leftexp{(h)}{\Gamma}^C_{\nu A} \mbo{e}_C
\end{align*}
provides the curvature for the induced connection
\begin{align*}
[ \grad_\mu, \grad_\nu ] \mbo{e}_B = R^A_{\,\,\,\, B \mu \nu} \, \mbo{e}_A.
\end{align*}
Now consider the `mixed' second covariant derivatives
\[
\leftexp{(h)}{\grad}_\mu \leftexp{(h)}{\grad}_A \mbo{e}_B \qquad \text{and} \qquad \leftexp{(h)}{\grad}_A \leftexp{(h)}{\grad}_\mu \mbo{e}_B.
\]
In view of the fact that $\mbo{e}_B$ and \smash{$\leftexp{(h)}{\grad}_A \mbo{e}_B $} do not have components in the tangent bundle of the domain $M$, the quantities
\[
\ptl_\mu U^C\leftexp{(h)}{\grad}_C \leftexp{(h)}{\grad}_A \mbo{e}_B \qquad \text{and} \qquad \leftexp{(h)}{\grad}_A \bigl(\ptl_\mu U^C \leftexp{(h)}{\grad}_C\bigr) \mbo{e}_B
\]
are equivalent to
\[
\leftexp{(h)}{\grad}_\mu \leftexp{(h)}{\grad}_A \mbo{e}_B \qquad \text{and} \qquad \leftexp{(h)}{\grad}_A \leftexp{(h)}{\grad}_\mu \mbo{e}_B, \qquad \text{respectively}.
\]
We have
\begin{subequations} \label{mixed-curv}
\begin{gather}
 U'^A \, \leftexp{(h)}{\grad}_A \leftexp{(h)}{\grad}_\mu \mbo{e}_B = U'^A \, \leftexp{(h)}{\grad}_A \bigl(\ptl_\mu U^C \, \leftexp{(h)}{\grad}_C \mbo{e}_B \bigr)\\
 \hphantom{U'^A \, \leftexp{(h)}{\grad}_A \leftexp{(h)}{\grad}_\mu \mbo{e}_B }{}
 = U'^A \ptl_\mu U^C \bigl(\ptl_A \leftexp{(h)}{\Gamma}^D_{CB} + \leftexp{(h)}{\Gamma}^D_{AE} \leftexp{(h)}{\Gamma}^E_{CB}\bigr) +
 U'^A \leftexp{(h)}{\grad}_A \ptl_\mu U^C \leftexp{(h)}{\grad}_C \mbo{e}_B,\notag
 \intertext{likewise}
 \ptl_\mu U^A \, \leftexp{(h)}{\grad}_A \bigl(U'^C\leftexp{(h)}{\grad}_C \mbo{e}_B\bigr) \notag\\
\qquad{} = \ptl_\mu U^A U'^C \leftexp{(h)}{\grad}_A \leftexp{(h)}{\grad}_C \mbo{e}_B + \ptl_\mu U^A \leftexp{(h)}{\grad}_A U'^C \leftexp{(h)}{\grad}_C \mbo{e}_B \notag\\
 \qquad{} = \ptl_\mu U^A U'^C \bigl(\ptl_C \leftexp{(h)}{\Gamma}^D_{AB} + \leftexp{(h)}{\Gamma}^D_{CE} \leftexp{(h)}{\Gamma}^E_{AB}\bigr)
 + \ptl_\mu U^A \leftexp{(h)}{\grad}_A U'^C \leftexp{(h)}{\grad}_C \mbo{e}_B,
\end{gather}
\end{subequations}
so that we have
\begin{align*}
U'^A \, \leftexp{(h)}{\grad}_A \leftexp{(h)}{\grad}_\mu \mbo{e}_B - \ptl_\mu U^A \, \leftexp{(h)}{\grad}_A \bigl(U'^C\leftexp{(h)}{\grad}_C \mbo{e}_B\bigr) = \leftexp{(h)}{R}^D_{\,\,\,\, B A \mu} \mbo{e}_D U'^A.
\end{align*}
This `mixed' derivative construction is relevant for our wave map deformations.
Let us assume that
\begin{subequations} \label{s-t-coordinates}
 \begin{align}
 & \bigl[ \ptl_\b U, U' \bigr] \equiv 0
 \intertext{from which, it follows that}
 &\ptl_\b U^A \, \leftexp{(h)}{\grad}_A U'^B - U'^A\,\leftexp{(h)}{\grad}_A \ptl_\b U^B \equiv 0.
 \end{align}
 \end{subequations}
Now consider another analogous curve $\mbo{\gamma}_\lambda$. The quantity \smash{$D^2_{\mbo{\gamma}_\lambda \mbo{\gamma}_s} \cdot L_{{\rm WM}} $}
involves the following terms:
\begin{align}\label{2-var-start}
\square_g U'^{A} + \ptl_{U^D} \leftexp{(h)}{\Gamma}^A_{B C} g^{\mu \nu} \ptl_\mu U^B \ptl_\nu U^C U'^D + 2 \leftexp{(h)}{\Gamma}^A_{BC} g^{\mu \nu} \ptl_\mu U'^B \ptl_\nu U^C.
\end{align}
Assuming that the Kerr wave map is a critical point of \eqref{1-var-wm} at $s=0$, the expression \eqref{wm-action-again} can consecutively be transformed as follows:
\begin{align*}
&{}= \square_g U'^{A} + \ptl_{U^D} \leftexp{(h)}{\Gamma}^A_{B C} g^{\mu \nu} \ptl_\mu U^B \ptl_\nu U^C U'^D + 2 \leftexp{(h)}{\Gamma}^A_{BC} g^{\mu \nu} \ptl_\mu U'^B \ptl_\nu U^C \notag\\
&\quad{} + \leftexp{(h)}{\Gamma}^{A}_{BC} U'^B \bigl(\square_g U^C + \leftexp{(h)}{\Gamma}^{C}_{DE} g^{\a\b} \ptl_\a U^D \ptl_\b U^E\bigr),
\end{align*}
which can be transformed to
\begin{align*} 
&g^{\mu \nu} U'^C \, \leftexp{(h)}{\grad}_C \leftexp{(h)}{\grad}_\mu \ptl_\nu U^A \notag\\
&\qquad= g^{\mu \nu} U'^C \bigl(\ptl_C \bigl(\leftexp{(h)}{\grad}_\mu \ptl_\nu U^A\bigr) - \leftexp{(h)}{\Gamma}^D_{C \mu} \leftexp{(h)}{\grad}_D \ptl_\nu U^A + \leftexp{(h)}{\Gamma}^{A}_{CD} \leftexp{(h)}{\grad}_\mu \ptl_\nu U^A \bigr).
\end{align*}
Now consider the operator
\begin{align*} 
g^{\mu\nu}\leftexp{(h)}{\grad}_\mu \bigl(\leftexp{(h)}{\grad}_C \ptl_\nu U^A\bigr) &{}= g^{\mu \nu} \leftexp{(g)}{\grad}_\mu \bigl(\leftexp{(h)}{\grad}_C \ptl_\nu U^A\bigr) \notag\\
&\quad{} - g^{\mu \nu} \bigl(\leftexp{(h)}{\Gamma}^D_{\mu C} \grad_D \ptl_\nu U^A + \leftexp{h}{\Gamma}^A_{\mu D} \leftexp{(h)}{\grad}_C \ptl_\nu U^D \bigr)
\end{align*}
and performing the computations analogous to \eqref{mixed-curv}, we get that \eqref{2-var-start} is equivalent to
\begin{align*}
\leftexp{(h)}{\square}\, U'^A + \leftexp{(h)}{R}^A_{\,\,\,\,BCD} g^{\mu \nu} \ptl_\mu U^B \ptl_\nu U^D U'^C,
\end{align*}
where
\begin{align*}
\leftexp{(h)}{\square} \, U'^A :={}& g^{\mu \nu} \leftexp{(h)}{\grad}_\mu \leftexp{(h)}{\grad}_\nu U'^A \notag\\
 ={}& \ptl_\mu \bigl(\leftexp{(h)}{\grad}_\nu U'^A\bigr) - \leftexp{(g)}{\Gamma}^{\gamma}_{\mu \nu} \leftexp{(h)}{\grad}_\gamma U'^A + \leftexp{(h)}{\Gamma}^A_{\mu C} \bigl(\leftexp{(h)}{\grad}_\nu U'^C\bigr),
 \end{align*}
which can be represented in terms of the covariant wave operator $\bigl(g^{\mu \nu}\leftexp{(g)}{\grad}_{\mu} \ptl_\nu U'^A\bigr)$ in the domain metric $g$ as
\begin{align*}
&{}= \square_g U'^A + g^{\mu \nu} \bigl(\ptl_\mu \bigl(\leftexp{(h)}{\Gamma}^A_{\nu C} U'^C\bigr) - \leftexp{(g)}{\Gamma}^\gamma_{\mu \nu} \leftexp{(h)}{\Gamma}^A_{\gamma C} U'^C + \leftexp{(h)}{\Gamma}^A_{\mu C} \ptl_\nu U'^C+ \leftexp{(h)}{\Gamma}^A_{\mu C} \leftexp{(h)}{\Gamma}^C_{\nu D} U'^D \bigr)
\end{align*}
and $\leftexp{(h)}{R}$ is the induced Riemannian curvature tensor
\[
 \leftexp{(h)}{R}^A_{\,\,\,\,BCD} = \ptl_C\leftexp{(h)}{\Gamma}^A_{DB} - \ptl_D\leftexp{(h)}{\Gamma}^A_{CB} + \leftexp{(h)}{\Gamma}^A_{CE} \leftexp{(h)}{\Gamma}^E_{DB} - \leftexp{(h)}{\Gamma}^A_{DE}\leftexp{(h)}{\Gamma}^E_{CB}.
\]
Now for the Kerr wave map critical point of $D_{\mbo{\gamma}_s} \cdot L_{{\rm WM}}$ at $s=0$, we then have
\begin{align} \label{div-var}
D^2_{\mbo{\gamma}_\lambda \mbo{\gamma}_s} \cdot L_{{\rm WM}} (s=0) = \int h_{AB} U'^B \bigl(\leftexp{(h)}{\square} U'^A + \leftexp{(h)}{R}^A_{\,\,\,\,BCD} g^{\mu \nu} \ptl_\mu U^B \ptl_\nu U^D U'^C\bigr) \bar{\mu}_g
\end{align}
as the Lagrangian variational principle for small linear deformations of the wave map $U_s \colon (M, g)\allowbreak \to (N, h)$.
In view of the divergence identity
\begin{align*}
\leftexp{(h)}{\grad}_\mu \bigl(h_{AB} U'^B \, \leftexp{(h)}{\grad}^\mu U'^A\bigr)= h_{AB} \leftexp{(h)}{\grad}^\mu U'^A \leftexp{(h)}{\grad}^{\mu} U'^B +
U'^B \leftexp{(h)}{\grad}_{\mu} \bigl(h_{AB} \leftexp{(h)}{\grad}^\mu U'^A\bigr)
\end{align*}
the variational principle \eqref{div-var} can equivalently be transformed into a self-adjoint variational~form
\begin{align}
& D^2_{\mbo{\gamma}_\lambda \mbo{\gamma}_s} \cdot L_{{\rm WM}} (s=0) \notag\\
&\qquad{}= - \halb \int \bigl(g^{\mu \nu}\, h_{AB}\leftexp{(h)}{\grad}_\mu U'^A \leftexp{(h)}{\grad}_\mu U'^B - h_{AB} U'^B \, \leftexp{(h)}{R}^A_{BCD} g^{\mu \nu} \ptl_\mu U^B \ptl_\nu U^C U'^D \bigr) \bar{\mu}_g.\!\!\label{quad-var}
\end{align}
Let us now calculate the Hamiltonian field equations for the linear perturbation theory, using the ADM decomposition of the background $(M, g)$
\begin{align*}
g = -N^2 {\rm d}t^2 + q_{ij} \bigl({\rm d}x^i + N^i {\rm d}t\bigr) \otimes \bigl({\rm d}x^j + N^j {\rm d}t\bigr).
\end{align*}
Let us denote the variational principle \eqref{div-var} and \eqref{quad-var} by $L_{{\rm WM}}(U')$. The Legendre transformation on \smash{$C'_{{\rm WM}}$} results in the phase space
\begin{align*}
X'_{{\rm WM}} := \bigl\{ \bigl(U'^{A}, p'_A\bigr) \bigr\}, \qquad \text{where $\bigl(U'^A, p'_A\bigr)$ are canonical pairs,}
\end{align*}
the conjugate momenta \smash{$p'_A = D_{\mbo{\gamma}} \cdot (p_{A} (s)) \vert_{s=0}$} are given by
\begin{align}
p'_A&{}= \frac{1}{N} \bar{\mu}_q h_{AB}(U) \bigl(\ptl_t U'^B + \leftexp{(h)}{\Gamma}^B_{t C} U'^C\bigr) - \frac{\bar{\mu}_q}{N} h_{AB}(U) \mathcal{L}_N U'^B \notag\\
&\quad{}- \frac{\bar{\mu}_q}{N} h_{AB}(U) N^a \leftexp{(h)}{\Gamma}^B_{a C}\label{u-prime-dot-2}
\end{align}
on account of the fact that the time derivative terms in the second term of \eqref{quad-var} only occur for background wave map $U$. Now then, using the quantity
\begin{align*}
h_{AB} \ptl_t U'^B = \bar{\mu}^{-1}_q N p'_A - h_{AB}(U)\leftexp{(h)}{\Gamma}^B_{tC} U'^C + h_{AB} \mathcal{L}_N U'^B + h_{AB} N^a \leftexp{(h)}{\Gamma}^B_{a C} U'^C,
\end{align*}
the Lagrangian and Hamiltonian densities, \smash{$\mathcal{L}'_{{\rm WM}}$} and \smash{$\mathcal{H}'_{{\rm WM}}$} can be expressed in terms of the phase space variables \smash{$X'_{{\rm WM}} = \bigl\{ \bigl(U'^A, p'_A\bigr) \bigr\}$} in a recognizable ADM form as follows:
\begin{align*}
\mathcal{L}'_{{\rm WM}} (U') &:=
\halb p'_A \ptl_t U'^A - \halb p'_A \mathcal{L}_N U'^A + \bigl(\leftexp{(h)}{\Gamma}^A_{tC} U'^C - \leftexp{(h)}{\Gamma}^A_{a C} N^a U'^C \bigr)\halb p'_A \notag\\
&\quad{}-\halb h_{AB}(U) N \bar{\mu}_q q^{ab} \bigl(\,\leftexp{(h)}{\grad}_a U^A \, \leftexp{(h)}{\grad}_b U^B \bigr) \notag\\
&\quad{}+ N \bar{\mu}_q h_{AE}(U) U'^A R^E_{\,\,\,\,BCD} q^{a b} \ptl_a U^B U'^C \ptl_b U^D \notag\\
&\quad{} -\frac{1}{N} \bar{\mu}_q h_{AE}(U) U'^A R^E_{\,\,\,\,BCD} \, \mathcal{L}_N U^B U'^C \mathcal{L}_N U^D \notag\\
&\quad{}- \frac{1}{N} \bar{\mu}_q h_{AE}(U) U'^A R^E_{\,\,\,\,BCD} \ptl_t U^B U'^C \ptl_t U^D \notag\\
&\q{} + \frac{2}{N}\bar{\mu}_q h_{AE}(U) U'^A R^E_{\,\,\,\,BCD} q^{a b} \ptl_t U^B U'^C \mathcal{L}_N U^D,
\end{align*}
likewise the Hamiltonian energy density can be expressed as
\begin{align*}
\mathcal{H}'_{{\rm WM}}
& := \halb p'_A \ptl_t U'^A + \halb p'_A \mathcal{L}_N U'^A -\bigl(\leftexp{(h)}{\Gamma}^A_{tC} U'^C - \leftexp{(h)}{\Gamma}^A_{a C} N^a U'^C \bigr)\halb p'_A \notag\\
&\quad{}+\halb h_{AB}(U) N \bar{\mu}_q q^{ab} \bigl(\,\leftexp{(h)}{\grad}_a U^A \, \leftexp{(h)}{\grad}_b U^B \bigr) \notag\\
&\quad{} - N \bar{\mu}_q h_{AE}(U) U'^A R^E_{\,\,\,\,BCD} q^{a b} \ptl_a U^B U'^C \ptl_b U^D \notag\\
&\quad{} +\frac{1}{N} \bar{\mu}_q h_{AE}(U) U'^A R^E_{\,\,\,\,BCD} \, \mathcal{L}_N U^B U'^C \mathcal{L}_N U^D \notag\\
&\quad{} + \frac{1}{N} \bar{\mu}_q h_{AE}(U) U'^A R^E_{\,\,\,\,BCD} \ptl_t U^B U'^C \ptl_t U^D \notag\\
&\quad{} - \frac{2}{N}\bar{\mu}_q h_{AE}(U) U'^A R^E_{\,\,\,\,BCD} q^{a b} \ptl_t U^B U'^C \mathcal{L}_N U^D,
\end{align*}
so that the critical point of $L'_{{\rm WM}}$
\begin{align*}
L'_{{\rm WM}} = \int^{t_2}_{t_1} \int_{\Sigma} \mathcal{L}'_{{\rm WM}} \,\, {\rm d}^2x {\rm d}t
\end{align*}
with respect to $U'^A$ gives the field equation
\begin{align}
\ptl_t p'_A &{}= \mathcal{L}_N p'_A + \bigl(\leftexp{(h)}{\Gamma} ^C_{tA} - \leftexp{(h)}{\Gamma}^C_{aA} N^a\bigr)p'_C + h_{AB} \leftexp{(h)}{\grad}_a \bigl(N \bar{\mu}_q q^{ab} \grad_b U'^B\bigr) \notag\\
&\quad{} + N \bar{\mu}_q h_{AE}(U) R^E_{\,\,\,\,BCD} q^{a b} \ptl_a U^B U'^C \ptl_b U^D \notag\\
&\quad{} -\frac{1}{N} \bar{\mu}_q h_{AE}(U) R^E_{\,\,\,\,BCD} \, \mathcal{L}_N U^B U'^C \mathcal{L}_N U^D
 - \frac{1}{N} \bar{\mu}_q h_{AE}(U) R^E_{\,\,\,\,BCD} \ptl_t U^B U'^C \ptl_t U^D \notag\\
&\quad{} + \frac{2}{N}\bar{\mu}_q h_{AE}(U) R^E_{\,\,\,\,BCD} q^{a b} \ptl_t U^B U'^C \mathcal{L}_N U^D.\label{p-prime-dot-2}
\end{align}
Analogously, it is straightforward to note that the field equations \eqref{u-prime-dot-2} and \eqref{p-prime-dot-2} are generated by the Hamiltonian $H'_{{\rm WM}} = \int \mathcal{H}'_{{\rm WM}} {\rm d}^2 x$, i.e.,
\begin{align*}
D_{p'^A} \cdot H'_{{\rm WM}} = \ptl_t U'^A, \qquad D_{U'^A} \cdot H'_{{\rm WM}} = - \ptl_t p'_A,
\end{align*}
respectively. Specializing to our stationary Kerr background metric, we have
\begin{gather*}
 h_{AB}(U) \ptl_t U'^B= \bar{\mu}^{-1}_q N p'_A, \\
\ptl_t p'_A = h_{AB} (U) \leftexp{(h)}{\grad}_a \bigl(N \bar{\mu}_q q^{ab} \grad_b U'^B\bigr) + N \bar{\mu}_q h_{AE}(U) R^E_{\,\,\,\,BCD} q^{a b} \ptl_a U^B U'^C \ptl_b U^D.
\end{gather*}

 Let us now construct the variational principle for the fully coupled Einstein-wave map perturbations.
Now suppose
\begin{align*}
q'_{ab} = D_{\mbo{\gamma}_s} \cdot (q_{ab} (s)) \vert_{s=0}, \qquad \mbo{\pi}'_{ab} = D_{\mbo{\gamma}_s} \cdot (\mbo{\pi}_{ab} (s)) \vert_{s=0},
\end{align*}
let us then denote the phase space corresponding to the perturbative theory of Kerr metric as~\smash{$X'_{{\rm EWM}}$}:
\begin{align*}
X' := \bigl\{ \bigl(U'^{A}, p'_A\bigr), \bigl(q'_{ab}, \mbo{\pi}'_{ab}\bigr) \bigr\}.
\end{align*}
Using the gauge-condition that the densitized metric $\bar{\mu}^{-1}_q q_{ab}$ is fixed, we can construct $D_{\mbo{\gamma}_s} \cdot H $ and $D_{\mbo{\gamma}_s} \cdot H_a $ at $s=0$
\begin{align*}
H':={}& D_{\mbo{\gamma}_s} \cdot H (s=0)=- \bar{\mu}^{-1}_q q_{ab} \mbo{\pi}'^{ab} \notag\\
 & - \bigl(\bar{\mu}_q R_q\bigr)'+ \halb \bar{\mu}_q q^{ab} \ptl_{U^C} h_{AB} (U) \ptl_a U^A \ptl_b U^B U'^C
\bar{\mu}_q q^{ab} h_{AB}(U) \ptl_a U'^A \ptl_b U^B
\end{align*}
and
\[
H'_a := D_{\mbo{\gamma}_s} \cdot H_a (s=0) = \leftexp{(q)}{\grad}_b \mbo{\pi}'^b_a + p'_A U^A,
\]
where
\[
\bigl(\bar{\mu}_q R_q\bigr)' = \bar{\mu}_q \bigl(- \Delta_q q' + \leftexp{(q)}{\grad}^a \leftexp{(q)}{\grad}^b q'_{ab} \bigr),\qquad q' := {\rm Tr}_q q'_{ab}.
\]
Again, after imposing that the Kerr metric is a critical point at $s=0$, we get
\begin{gather}
 D^2_{\mbo{\gamma}_\lambda \mbo{\gamma}_s} \cdot H(s=0) = \bar{\mu}^{-1}_q \bigl(2 \Vert \mbo{\pi}' \Vert_q^2 - 2\bigl(q_{ab} \mbo{\pi}'^{ab}\bigr)^2 + p'_A p'^A\bigr) \notag\\
 \hphantom{D^2_{\mbo{\gamma}_\lambda \mbo{\gamma}_s} \cdot H(s=0) =}{}
 - \bigl(\bar{\mu}_q R_q\bigr)'' + \halb \bar{\mu}_q q^{ab} \ptl^2_{U^D U^C} h_{AB}(U) \ptl_a U^A \ptl_b U^B U'^C U'^D \notag\\
 \hphantom{D^2_{\mbo{\gamma}_\lambda \mbo{\gamma}_s} \cdot H(s=0) =}{}
 + \bar{\mu}_q q^{ab} \ptl_{U^C} h_{AB}(U) \ptl_a U'^A \ptl_b U^B U'^C \notag\\
 \hphantom{D^2_{\mbo{\gamma}_\lambda \mbo{\gamma}_s} \cdot H(s=0) =}{}
 + \halb \bar{\mu}_q q^{ab} \ptl_{U^C} h_{AB}(U) \ptl_a U^A \ptl_b U^B U''^C \notag\\
 \hphantom{D^2_{\mbo{\gamma}_\lambda \mbo{\gamma}_s} \cdot H(s=0) =}{}
 + \bar{\mu}_q q^{ab} \ptl_{U^C} h_{AB} (U) \ptl_a U'^A \ptl_b U^B U'^C \notag\\
 \hphantom{D^2_{\mbo{\gamma}_\lambda \mbo{\gamma}_s} \cdot H(s=0) =}{}
 + \bar{\mu}_q q^{ab} h_{AB}(U) \ptl_a U''^A \ptl_b U^B + \bar{\mu}_q q^{ab} h_{AB} (U) \ptl_a U'^A \ptl_b U'^B ,\label{H''}\\
D^2_{\mbo{\gamma}_\lambda \mbo{\gamma}_s} \cdot H_a (s=0)
= -4 \leftexp{(q')}{\grad}_b \mbo{\pi}'^b_a -2 \leftexp{(q)}{\grad}_b \mbo{\pi}''^b_a + 2p'_A \ptl_a U'^A + \ptl_a U^A p''_A, \notag 
\end{gather}
where
\[
\leftexp{(q')}{\grad}_b V^a := \ptl_b V^a+ \halb q^{ad} \bigl(\leftexp{(q)}{\grad}_b q'_{dc} + \leftexp{(q)}{\grad}_c q'_{bd} -\leftexp{(q)}{\grad}_d q'_{bc} \bigr) V^c.
\]

We arrive at the following theorem.

\begin{Corollary}\label{second-var}
Suppose $X'$ is the first variation phase space, then the field equations for the dynamics in $X'$ are given by the variational principle
\begin{align*}
J_{{\rm EWM}} \bigl(X'_{{\rm EWM}}\bigr) := \int \biggl(\mbo{\pi}'^{ab} \ptl_t q'_{ab} + p'_A U^{'A} - \halb N H'' - N' H' - N'_a H'_a \biggr),
\end{align*}
where $H''$, $H'$, $H'_a$ are \eqref{H''}, $D_{\mbo{\gamma}_s} \cdot H$ and $D \cdot H_a$ at $s=0$, respectively, $N' := D_{\mbo{\gamma}_s} \cdot N \big\vert_{s=0}$ and~${N'_a := D_{\mbo{\gamma}_s} \cdot N_a \big\vert_{s=0}}$.
\end{Corollary}

The approach used above is the classical Jacobian method, as remarked by Moncrief \cite{Moncrief_74}. Separately, it may be noted that the construction of the wave map field equations is analogous to that of the geodesic deviation equations or the `Jacobi' fields \cite{Misn_78, synge_34}. In view of the fact that the Hamiltonian formulation of the geodesic deviation equations is relatively uncommon, our derivation may also be adapted for this purpose. Finally, we would like to emphasize that our assumption that \eqref{s-t-coordinates} holds, is not (effectively) a restriction in the class of perturbations. In case this condition is relaxed, we shall also pick up the Riemann curvature of the target, but with torsion. The fact that we pick only the curvature term of the target is crucial for our work. We would also like to remark that the deformations which correspond to the coordinate directional derivatives along the curves $\mbo{\gamma}_\lambda$ are equivalent to (induced) covariant deformations on the target, on account of the fact the Kerr wave map is a critical point of~\eqref{wm-action-again}.

The variational principle in Corollary \ref{second-var} and its field equations correspond to a general Weyl--Papapetrou gauge. If we consider further gauge-fixing \eqref{conformal}, where the densitized metric $ \bar{\mu}^{-1}_q q_{ab}$ or equivalently the densitized inverse metric $\bar{\mu}_q q^{ab}$ is fixed, we obtain
\begin{gather*}
H' = \bar{\mu}_{q_0} \bigl(2 \Delta_0 \mbo{\nu}'\bigr) + \halb \bar{\mu}_{q_0} \ptl_{U^C} h_{AB} q_0^{ab} \ptl_a U^A \ptl_b U^B U'^C
+ \bar{\mu}_{q_0} q_0^{ab} h_{AB} \ptl_a U'^A \ptl_b U^B, \notag\\
H'' =\bar{\mu}^{-1}_{q_0} \bigl(2 {\rm e}^{-2\mbo{\nu}} \Vert \varrho' \Vert^2_{q_0} - \tau'^2 {\rm e}^{2\nu} \bar{\mu}^2_{q_0} + p'_A p'^A\bigr)\notag \\
\hphantom{H'' =}{}
+ \bar{\mu}_{q_0} \bigl(2 \Delta_0 \mbo{\nu}'' + \ptl^2_{U^C U^D} h_{AB} q^{ab}_0 \ptl_a U^A \ptl_b U^B U'^C U'^D + \ptl_C h_{AB} q^{ab}_0 \ptl_a U'^A \ptl_b U^B U'^C \notag\\
\hphantom{H'' =+ \bar{\mu}_{q_0} \bigl(}{}
+ \halb \ptl_Ch_{AB} q_{0}^{ab} \ptl_a U^A \ptl_b U^B U''^C + 2 h_{AB} q^{ab}_{0} \ptl_a U''^A \ptl_b U^B \notag\\
\hphantom{H'' =+ \bar{\mu}_{q_0} \bigl(}{}
+ 2h_{AB} q^{ab}_0 \ptl_a U'^A \ptl_bU'^B + 2 \ptl_C h_{ab} q^{ab}_0 \ptl_a U'^A \ptl_b U^B U'^C \bigr).
\end{gather*}
The aim of our work is to construct an energy for the linear perturbative theory of Kerr black hole spacetimes, for which the Hamiltonian formulation is naturally suited. In contrast with the Lagrangian variational principles (e.g., \eqref{EH} and \eqref{EWM-Lag}), the Hamiltonian variation principles are not spacetime diffeomorphism invariant.
In this work, we shall work in the $2 + 1$ maximal gauge condition. We point out that this gauge condition was also used by Dain--de Austria for the extremal case \cite{DA_14}. We shall need the following statement.

\begin{claim} \label{2+1max}
Suppose $N' \in C^{\infty} (\Sigma)$,
\begin{subequations}
\begin{align}
&\Delta_0 N' = 0, \qquad \textnormal{in the interior of $(\Sigma, q_0 )$}, \label{Lap-N}\\
&N' \vert_{\ptl \Sigma}= 0, \label{Lap-bdry}
\end{align}
\end{subequations}
 then $ N' \equiv 0$ on $(\Sigma, q_0)$.
\end{claim}
\begin{proof}
If we multiply \eqref{Lap-N} with $N'$ and integrate by parts, we get $\int \vert \grad_0 N \vert^2 =0$ in the interior of $\Sigma$, after using \eqref{Lap-bdry}. It follows that $N'$ is a constant in $\Sigma$.
\end{proof}

The variational principle in Corollary \ref{second-var} now gives the following field equations (for smooth and compactly supported variations)
\begin{gather*}
 h_{AB}(U) \ptl_t U'^B= {\rm e}^{2 \mbo{\nu}}\bar{\mu}^{-1}_{q_0} N p'_A + h_{AB}(U) \mathcal{L}_{N'} U^B, \\ 
\ptl_t p'_A = h_{AB} (U) \leftexp{(h)}{\grad}_a \bigl(N \bar{\mu}_{q} q^{ab} \grad_b U'^B\bigr) + N \bar{\mu}_q h_{AE}(U) R^E_{\,\,\,\,BCD} q^{a b} \ptl_a U^B U'^C \ptl_b U^D
\notag\\ \hphantom{\ptl_t p'_A }{}
=h_{AB} (U) \leftexp{(h)}{\grad}_a \bigl(N \bar{\mu}_{q_0} q_0^{ab} \grad_b U'^B\bigr)
+ N \bar{\mu}_{q_0} h_{AE}(U) R^E_{\,\,\,\,BCD} q_0^{a b} \ptl_a U^B U'^C \ptl_b U^D, \\ 
\ptl_t q'_{ab}= 2N \bar{\mu}^{-1}_{q} {\rm CK}_{ab}\bigl(Y', q\bigr)+ \leftexp{(q)}{\grad} _a N'_b + \leftexp{(q)}{\grad}_b N'_a \notag\\
\hphantom{\ptl_t q'_{ab}}{}
= 2N {\rm e}^{-2 \nu} \bar{\mu}^{-1}_{q_0} {\rm CK}_{ab}\bigl(Y', q_0\bigr)+ \mathcal{L}_{N'} \bigl({\rm e}^{ 2 \mbo{\nu}}(q_0)_{ab}\bigr), \\
\ptl_t \mbo{\pi}'^{ab} = \bigl(\bar{\mu}_q q^{bc} q^{ad}\bigr)' \bigl(\ptl^2_{dc} N - \leftexp{(q)}{\Gamma} ^f_{cd} \ptl_f N\bigr) \notag\\
\hphantom{\ptl_t \mbo{\pi}'^{ab} =}
+ \bar{\mu}_q q^{bc} q^{ad} \bigl(q^{fl} \bigl(\leftexp{(q)}{\grad} _d q'_{l c} + \leftexp{(q)}{\grad} _c q'_{l d}+\leftexp{(q)}{\grad} _l q'_{cd}\bigr) \ptl_f N\bigr) \notag\\
\hphantom{\ptl_t \mbo{\pi}'^{ab} =}
+ \biggl(\halb N \bar{\mu}_q \biggl(q^{ac} q^{bd} - \halb q^{ab} q^{cd}\biggr)\biggr)' h_{AB}(U) \ptl_a U^a \ptl_b U^B \notag\\
\hphantom{\ptl_t \mbo{\pi}'^{ab} =}
+ \halb N \bar{\mu}_q \biggl(q^{ac} q^{bd} - \halb q^{ab} q^{cd}\biggr) \bigl(2 h_{AB} (U) \ptl_a U'^A \ptl_b U^B + \ptl_{U^C} h_{AB} (U) \ptl_a U^A \ptl_b U^B U'^C\bigr)
\end{gather*}
together with the constraints
\begin{align*}
H'= 0 \qquad
\text{and}\qquad
H'_a= \leftexp{(q)}{\grad}_b \mbo{\pi}'^b_a + p'_A \ptl_a U^A= 0, 
\end{align*}
in the $2+1$ maximal gauge.
We have
\[
q'= \operatorname{Tr} q'_{ab}, \qquad \tau' = \bar{\mu}^{-1}_q q_{ab} \mbo{\pi}'^{ab},
\]
and
\[
\ptl_t q' = -2N \tau' + 2 \, \leftexp{(q)}{\grad}^c N'_c, \qquad \ptl_t \tau' = - \Delta_0 N' + N \mathfrak{q}'.
\]
In the $2+1$ maximal gauge (cf.\ Claim \ref{2+1max}),
\[
\ptl_t q' = 2\leftexp{(q)}{\grad}^c N'_c, \qquad \Delta_0 N' =0.
\]
Let us now formally discuss the structures associated to our Hamiltonian framework.
 The phase space \smash{$X'_{{\rm EWM}}$} is such that \smash{$\bigl(q'_{ab}, U'^A\bigr)$} are $C^{\infty}(\Sigma)$ symmetric covariant 2-tensor and smooth vector field respectively and \smash{$\bigl(\mbo{\pi}'^{ab}, p'_A\bigr)$} are $C^{\infty}(\Sigma)$ symmetric 2-tensor densities and scalar density (for each $A$) respectively, which together form the cotangent bundle $T^*\mathcal{M}$, which we had represented as \smash{$X'_{{\rm EWM}}$}. The Hamiltonian and momentum constraint spaces $\mathscr{C}_{H'}$, $\mathscr{C}_{H'_a}$ are defined as follows:
 \begin{subequations} \label{constraint-set}
\begin{align}
&\mathscr{C}_{H'} = \bigl\{ \bigl(q'_{ab}, \mbo{\pi}'^{ab}\bigr)\bigl(U'^A, p'_A\bigr) \in T^*\mathcal{M} \mid H' =0 \bigr\}, \\
&\mathscr{C}_{H'_a} = \bigl\{ \bigl(q'_{ab}, \mbo{\pi}'^{ab}\bigr)\bigl(U'^A, p'_A\bigr) \in T^*\mathcal{M} \mid H'_a =0,\, a =1,2 \bigr\}.
\end{align}
\end{subequations}

Furthermore, we consider our time coordinate gauge condition to be `$2+1$ maximal'
\begin{align*}
\mathscr{C}_{\tau'} = \bigl\{ \bigl(q'_{ab}, \mbo{\pi}'^{ab}\bigr)\bigl(U'^A, p'_A\bigr) \in T^*\mathcal{M} \mid
\tau' =0 \bigr\}.
\end{align*}
In our work, we shall be interested in the space
\begin{align*}
\mathscr{C}_{H'} \cap \mathscr{C}_{H'_a} \cap \mathscr{C}_{\tau'}
\end{align*}
for our initial value framework.
In general, proving local existence of Einstein equations using the Hamiltonian initial value problem is a complex problem. 
We note the following statement from the Lagrangian framework of Einstein's equations from the classical result of Choquet-Bruhat and Geroch \cite{Bruhat_Geroch_classic} in $3+1$ dimensions.

Suppose \smash{$\bigl\{ \bigl(\bar{q}'_{ab}, \bar{\mbo{\pi}}'^{ab}\bigr)\bigr\}_0 \in \mathscr{C}_{\bar{H}'} \cap \mathscr{C}_{\bar{H}'_i}$}, then it follows from the classic results of Choquet-Bruhat and Geroch, adapted to our linear perturbation problem, that there exists a unique, regular, maximal development of \smash{$\bigl\{ \bigl(\bar{q}'_{ab}, \bar{\mbo{\pi}}'^{ab}\bigr) \bigr\}_0$}, $ \iota \colon \olin{\Sigma} \to \olin{\Sigma} \times \mathbb{R}$, such that \smash{$\bigl\{ \bigl(\bar{q}'_{ab}, \bar{\mbo{\pi}}'^{ab}\bigr) \bigr\}_t \in \mathscr{C}_{\bar{H}'} \cap \mathscr{C}_{\bar{H}'_i} $} is causally determined from the initial data \smash{$\bigl\{ \bigl(q'_{ab}, \mbo{\pi}'^{ab}\bigr), \bigl(U'^A, p'_A\bigr) \bigr\}_0 $} in a suitable gauge; where~\smash{$\mathscr{C}_{\bar{H}'}$} and \smash{$\mathscr{C}_{\bar{H}'_i}$} are defined analogous to \eqref{constraint-set}.

Let us introduce the following notions from the machinery of linearization stability.
Let~us~define the constraint map $\olin{\Psi}$ of $\bigl(\overline{\Sigma}, \bar{q}\bigr)$ as a map from the cotangent bundle to a $4-$tuple of scalar densities, $ \olin{\Psi} \colon T^* \olin{\mathcal{M}} \to \mathcal{C}^{\infty} (\olin{\Sigma}) \times \mathcal{T} \olin{\Sigma}$, such that
\begin{align*}
\olin{\Psi} (\bar{q}, \bar{\mbo{\pi}}) = \bigl(\bar{H}, \bar{H}_i\bigr), \qquad i = 1, 2, 3.
\end{align*}
Let us denote the deformation of the constraint map as $D \cdot \olin{\Psi} (\bar{q}, \bar{\mbo{\pi}})$.
Then the $L^2$-adjoint,
$D^\dagger \cdot \olin{\Psi} (\bar{q}, \bar{\mbo{\pi}}) \bigl(\bar{C}, \bar{Z}\bigr), $ of the deformations $D \cdot \Psi$ of the constraint map is a
2-tuple (an element of a Banach space) consisting of a covariant symmetric 2-tensor and a contravariant symmetric 2-tensor density and is given by
\begin{align}
D^\dagger \cdot \Psi (\bar{q}, \bar{\mbo{\pi}}) \bigl(\bar{C}, \bar{Z}\bigr) &:= \biggl(\bar{\mu}^{-1}_{\bar{q}} \biggl( \halb \biggl(\Vert \mbo{\pi} \Vert_{\bar{q}}^2 - {\rm Tr}_{\bar{q}}(\bar{\mbo{\pi}})^2 \biggr) \bar{q}^{ij} \bar{C} - 2 \biggl(\bar{\mbo{\pi}}^{ik} \bar{\mbo{\pi}}_k^{j} - \halb \mbo{\pi}^{ij} {\rm Tr}_{\bar{q}} (\bar{\mbo{\pi}}) \biggr) \bar{C} \biggr) \notag\\
&\hphantom{:= \biggl(}{} -\bar{\mu}_q \biggl(\bar{q}^{ij} \Delta_{\bar{q}} C - \leftexp{(\bar{q})}{\grad}^{i} \leftexp{(\bar{q})}{\grad}^{j} \bar{C} + R^{ij} \bar{C} - \halb \bar{q}^{ij} R_{\bar{q}} \bar{C} \biggr) \notag\\
&\hphantom{:= \biggl(}{} + \leftexp{(\bar{q})}{\grad}_k \bigl(\bar{Z}^k \bar{\mbo{\pi}}^{ij}\bigr) - \leftexp{(\bar{q})}{\grad}_k \bar{Z}^i \bar{\mbo{\pi}}^{kj} - \leftexp{(\bar{q})}{\grad}_k \bar{Z}^i \bar{\mbo{\pi}}^{jk}, \notag\\
&\hphantom{:= \biggl(}{} -2 \bar{\mu}_{\bar{q}}^{-1} \bar{C} \biggl(\bar{\mbo{\pi}}_{ij} - \halb {\rm Tr} (\bar{\mbo{\pi}}) \bar{q}_{ij} \biggr) - \leftexp{(q)}{\grad}_i \bar{Z}_j - \leftexp{(q)}{\grad}_j \bar{Z}_i \biggr).\label{3+1adjoint}
\end{align}
The expression \eqref{3+1adjoint} is closely related to the $L^2-$adjoint of the Lichnerowicz operator.
Moncrief had characterized the splitting theorem, established by Fischer--Marsden \cite{FM_75}, of the (Banach) spaces acted on by the constant map
\begin{align*}
 \operatorname{ker}D^\dagger \cdot \Psi (\bar{q}, \bar{\mbo{\pi}}) \bigl(\bar{C}, \bar{Z}\bigr) \,\oplus\, \text{range} \, D \cdot \Psi (\bar{q}, \bar{\mbo{\pi}})\bigl(\bar{q}', \bar{\mbo{\pi}}'\bigr),
\end{align*}
by associating the kernel of the adjoint operator \big($\operatorname{ker} D^\dagger \cdot \olin{\Psi} (\bar{q}, \bar{\mbo{\pi}}) \bigl(\bar{C}, \bar{Z}\bigr)$\big) to the existence of spacetime Killing isometries. In particular, Moncrief proved that $\operatorname{ker} D^\dagger \cdot \olin{\Psi} \bigl(\bar{C}, \bar{Z}\bigr)$ is non-empty if and only if there exists a spacetime Killing vector.
 This result is crucial for our work, but in the dimensionally reduced framework. In the following, we shall establish equivalent results in our dimensionally reduced perturbation problem.
\begin{Lemma}
Suppose $(M, g)$ is the $2+1$ spacetime obtained from the dimensional reduction of the axially symmetric, Ricci-flat $3+1$ spacetime $\bigl(\bar{M}, \bar{g}\bigr)$ and $D \cdot \Psi$ is the deformation around the Kerr metric of the constraint map $\Psi$ of the dimensionally reduced $2+1$ Einstein-wave map system on $(M, g)$, then
\begin{enumerate}\itemsep=0pt
\item[$(1)$] The adjoint $D^\dagger \cdot \Psi (q', \mbo{\pi}') (C, Z) $ of the constraint map $\Psi$ is given by
\begin{gather}
D^{\dagger} \cdot \Psi = \biggl(\bar{\mu}_q \bigl(\leftexp{(q)}{\grad}^b \,\leftexp{(q)}{\grad}^a C -q^{ab}\, \leftexp{(q)}{\grad}_c \leftexp{(q)}{\grad}^c C\bigr)\nonumber\\
\hphantom{D^{\dagger} \cdot \Psi =}{}
+ \halb C \bar{\mu}_q h_{AB} \biggl(q^{ac}q^{bd} - \halb q^{ab} q^{cd}\biggr)\ptl_a U^A \ptl_b U^B
 - \leftexp{(q)}{\grad}_a Z_b - \leftexp{(q)}{\grad}_b Z_a \biggr).\label{2+1adjoint}
\end{gather}

\item[$(2)$] The kernel $\bigl(\operatorname{ker} \bigl(D^\dagger \Psi\bigr)\bigr)$ of the adjoint of the constraint map $\Psi$ is one dimensional and is equal to
$(N, 0)^{\mathsf{T}}$.
\end{enumerate}
\end{Lemma}
\begin{proof}
Consider the constraint map $\Psi$
\begin{align*}
\Psi (g, \mbo{\pi}) = (H, H_i),
\end{align*}
then from the deformation of $\Psi, D \cdot \Psi (q, \mbo{\pi}) (q', \mbo{\pi}') := (H', H'_a)$, around a general metric,
it~follows that its $2+1$ $L^2-$adjoint is given by
\begin{align*} 
D^{\dagger} \cdot \Psi= \biggl(&\halb C \bar{\mu}^{-1}_q q^{ab} \bigl(\Vert \mbo{\pi} \Vert^2_q - {\rm Tr}(\mbo{\pi})^2\bigr)- 2C \bar{\mu}^{-1}_q \bigl( \mbo{\pi}^{ac} \mbo{\pi}^{b}_c - \mbo{\pi}^{ab} {\rm Tr}_q(\mbo{\pi}) \bigr) \notag\\
&{}+\bar{\mu}_q \bigl(\leftexp{(q)}{\grad}^b \,\leftexp{(q)}{\grad}^a C -q^{ab}\, \leftexp{(q)}{\grad}_c \leftexp{(q)}{\grad}^c C\bigr) \notag\\
&{}+ \leftexp{(q)}{\grad}_c\bigl(\mbo{\pi}^{ab} Z^c\bigr) - \leftexp{(q)}{\grad}_c Z^a \mbo{\pi}^{cb} - \leftexp{(q)}{\grad}_c Z^b \mbo{\pi}^{ca}\notag\\
&{}+ \frac{1}{4}\bar{\mu}_q^{-1} C q^{ab} p_A p^A + \halb C \bar{\mu}_q h_{AB} \biggl(q^{ac}q^{bd} - \halb q^{ab} q^{cd}\biggr)\ptl_a U^A \ptl_b U^B, \notag\\
&{} -2C \bar{\mu}^{-1}_q (\mbo{\pi}_{ab} - q_{ab} {\rm Tr}_q \mbo{\pi}) - \leftexp{(q)}{\grad}_a Z_b - \leftexp{(q)}{\grad}_b Z_a\biggr),
\end{align*}
analogous to \eqref{3+1adjoint}, while noting that the (dimensionally reduced) wave map variables are not constrained due to the introduction of the twist potential, after using the Poincar\'e Lemma (see, e.g., \eqref{2+1constraints} and then \eqref{2+1-no-constraints}). The expression \eqref{2+1adjoint} follows for the case of dimensionally reduced Kerr metric.
Now assume that $(C, Z) \in \operatorname{ker} D^\dagger \cdot \Psi$. It follows from \eqref{2+1adjoint} that a vector $K = K_{\perp} \mbo{n} + K_{\parallel} $
satisfies
\begin{align}\label{Killing}
\leftexp{(g)}{\grad}_\a K_\b + \leftexp{(g)}{\grad}_\b K_\a =0
\end{align}
with $K_{\perp} = C$ and $K_{\parallel} = Z$
which implies that $K = (C, Z)$ is a (spacetime) Killing vector in $(M, g)$. Conversely, assuming that \eqref{Killing} holds
it follows that the left-hand side of \eqref{2+1adjoint} vanishes, which implies $K \in \operatorname{ker} D^\dagger \cdot \Psi$. In particular, for the dimensionally reduced Kerr metric $(M,g)$ the only remaining linearly independent Killing vector is $\ptl_t$, so $(C, Z)^{\mathsf{T}} \equiv (N, 0)^{\mathsf{T}}$, which, as will be shown later, resolves (P2).
\end{proof}

In the following, we shall establish that the Hamilton vector field $(H', H'_a)$ is tangential to the flow of the phase space variables $\bigl\{ (q', \pi'), \bigl(U'^A, p'_A\bigr) \bigr\}$ in $\mathscr{C}_{H'} \cap \mathscr{C}_{H'_a} \cap \mathscr{C}_{\tau'}$.

\begin{Lemma}
Suppose $H'$ and $H'_a$ are the linearized Hamiltonian and momentum constraints of the $2+1$ Einstein-wave map system, then their propagation equations are
\begin{subequations}
\begin{align}
&\frac{\ptl}{\ptl t} H' = q^{ab} \ptl_a N H'_b + \ptl_b \bigl(N q^{ab} H'_a\bigr), \label{h-dot}\\
&\frac{\ptl}{\ptl t} H'_a= \ptl_a N H' \label{ha-dot}
\intertext{and}
&N H' = \ptl_b \bigl(N \bar{\mu}_q q^{ab} h_{AB} U'^{A} \ptl_b U^B - 2 \bar{\mu}_q q^{ab} \ptl_a N \mbo{\nu}'\bigr). \label{LS}
\end{align}
\end{subequations}
\end{Lemma}
\begin{proof}
The statements \eqref{h-dot} and \eqref{ha-dot} follow from the linearized and background (exact) field equations of our $2+1$ Einstein-wave map system. For simplicity in computations, we shall perform our computations with $q'_0$
held fixed. Recall
\begin{gather*}
\ptl_t \varrho'^a_b = N \bar{\mu}_{q_0} \biggl(q^{ac}_0 \delta^d_b - \halb q_0^{cd} \delta^a_b\biggr) \biggl(h_{AB} \ptl_c U'^A \ptl_d U^B + \halb \ptl_{U^C} h_{AB} (U) \ptl_c U^A \ptl_d U^B U'^C\biggr) \notag\\ \hphantom{\ptl_t \varrho'^a_b =}{}
+ \bar{\mu}_{q_0} q_0^{cd} \delta^a_b \ptl_c N \ptl_d \mbo{\nu}' - \bar{\mu}_{q_0} q^{ac}_0 \bigl(\ptl_b N \ptl_c \mbo{\nu}'\bigr), \\
\ptl_t \mbo{\nu}'= \frac{1}{2 \bar{\mu}_{q_0}} \ptl_c \bigl(\bar{\mu}_{q_0} N'^c\bigr) + 2 \mathcal{L}_{N'} \mbo{\nu}.
\end{gather*}
Consider the quantities
\begin{gather*}
\ptl_t \bigl(2 \bar{\mu}^{-1}_{q_0} \ptl_b \bigl(\bar{\mu}_{q_0} q^{ab}_0 \ptl_a \mbo{\nu}'\bigr) \bigr) = \bar{\mu}^{-1}_{q_0} \ptl_b \bigl(\bar{\mu}_{q_0} q^{ab}_0 \ptl_a \bigl( \bar{\mu}^{-1}_{q_0} \ptl_c \bigl(\bar{\mu}_{q_0} N'^c\bigr) + 2 N'^c\ptl_c \mbo{\nu}\bigr)\bigr),
\\
\halb \bar{\mu}_q q^{ab} \ptl_{U^C} h_{AB} \ptl_a U^A \ptl_b U^B \ptl_t U'^C = \halb N q^{ab} \ptl_{U^C} h_{AB} \ptl_a U^A \ptl_b U^B p'^C \notag\\
\hphantom{\halb \bar{\mu}_q q^{ab} \ptl_{U^C} h_{AB} \ptl_a U^A \ptl_b U^B \ptl_t U'^C =}{}
+ \halb \bar{\mu}_q q^{ab} \ptl_{U^C} h_{AB} \ptl_a U^A \ptl_b U^B \mathcal{L}_{N'} U^C, \\
 \bar{\mu}_q q^{ab} h_{AB}(U) \ptl_a \bigl(\ptl_t U'^A\bigr) \ptl_b U^B= \bar{\mu}_q q ^{ab} h_{AB}(U) \ptl_a \bigl(\bar{\mu}^{-1}_q N p'^A\bigr)
 \ptl_b U^B \notag\\ \hphantom{\bar{\mu}_q q^{ab} h_{AB}(U) \ptl_a \bigl(\ptl_t U'^A\bigr) \ptl_b U^B=}{}
 + \bar{\mu}_q q ^{ab} h_{AB}(U) \ptl_a \bigl(\mathcal{L}_{N'} U^A\bigr) \ptl_b U^B
 \notag\\ \hphantom{\bar{\mu}_q q^{ab} h_{AB}(U) \ptl_a \bigl(\ptl_t U'^A\bigr) \ptl_b U^B}{}
 = \bar{\mu}_q q ^{ab} h_{AB}(U) \ptl_a \bigl(\bar{\mu}^{-1}_q N p'^A\bigr)\ptl_b U^B
 \notag\\ \hphantom{\bar{\mu}_q q^{ab} h_{AB}(U) \ptl_a \bigl(\ptl_t U'^A\bigr) \ptl_b U^B=}{}
 + \bar{\mu}_q q ^{ab} h_{AB}(U) \mathcal{L}_{N'} \bigl(\ptl_a U^A\bigr) \ptl_b U^B.
\end{gather*}
Combining the results above and noting that for our gauge
\begin{align*}
{\rm CK}^{ab} (N', q_0) = \bar{\mu}_{q_0} \bigl(\leftexp{(q_0)}{\grad}^a N'^b + \leftexp{(q_0)}{\grad}^b N'^a - q^{ab}_0 \leftexp{(q_0)}{\grad}_c N'^c\bigr) = - 2 N {\rm e}^{-2 \nu} q^{bc}_0 \varrho'^a_c,
\end{align*}
we get
\begin{align*}
\ptl_t H' &{}= q^{ab} \ptl_b N \bigl(-2 \leftexp{(q_0)}{\grad}_c \varrho'^c_a + p'_A \ptl_a U^A\bigr) -2 \ptl_b \bigl(N q^{ab} \leftexp{(q_0)}{\grad}_c \varrho'^c_a\bigr) + \ptl_b \bigl(N q^{ab} p'_A \ptl_a U^A\bigr) \\
&{}= q^{ab} \ptl_a N H'_b + \ptl_b \bigl(N q^{ab} H'_a\bigr).
\end{align*}
Likewise, for \eqref{ha-dot} consider
\begin{gather*}
 \leftexp{(q_0)}{\grad}_a (\ptl_t \varrho'^a_b) = \leftexp{(q_0)}{\grad}_a \biggl(N \bar{\mu}_{q_0} \biggl(q^{ac}_0 \delta^d_b - \halb q_0^{cd} \delta^a_b\biggr)\notag\\ \hphantom{\leftexp{(q_0)}{\grad}_a (\ptl_t \varrho'^a_b) =\leftexp{(q_0)}{\grad}_a \biggl(}{}
 \times\biggl(h_{AB} \ptl_c U'^A \ptl_d U^B+ \halb \ptl_{U^C} h_{AB} (U) \ptl_c U^A \ptl_d U^B U'^C\biggr) \notag\\
 \hphantom{\leftexp{(q_0)}{\grad}_a (\ptl_t \varrho'^a_b) =\leftexp{(q_0)}{\grad}_a \biggl(}{}
 +\bar{\mu}_{q_0} q_0^{cd} \delta^a_b \ptl_c N \ptl_d \mbo{\nu}' - \bar{\mu}_{q_0} q^{ac}_0 \bigl(\ptl_b N \ptl_c \mbo{\nu}'\bigr) \biggr), \\
\ptl_a U^A \ptl_t p'_A = h_{AB} \ptl_a U^A \leftexp{(h)}{\grad}_c \bigl(N \bar{\mu}_q q^{cb} \leftexp{(h)}{\grad}_b U'^B\bigr) \notag\\ \hphantom{\ptl_a U^A \ptl_t p'_A =}{}
+ N \bar{\mu}_q h_{AB} (U) \ptl_a U'^B R^{E}_{\,\,\,\,BCD} q^{ab} \ptl_a U^B \ptl_b U^D U'^C.
\end{gather*}

Now combining all the above, we have
\begin{align*}
\ptl_t H'_c = \ptl_c N H'
\end{align*}
in view of the background field equations \eqref{Kerr-p-dot} and \eqref{kerr-maximal}. 
For \eqref{LS}, first note that
\begin{align}
 N \bar{\mu}_{q_0} \ptl_{U^C} h_{AB} (U) q^{ab}_0 \ptl_a U^A \ptl_b U^C U'^B &{}= N \bar{\mu}_{q_0} h_{AB}\leftexp{(h)}{\Gamma}^A_{CD} (U) q^{ab}_0 \ptl_a U^C \ptl_b U^D U'^B \notag\\
 &\quad- \halb N \bar{\mu}_{q_0} \ptl_{U^C} h_{AB}(U) q^{ab}_0 \ptl_a U^A \ptl_b U^B U'^C\label{var-from-chris}
\end{align}
after a suitable relabelling of the indices. Now consider
\begin{align} \label{LS-computation}
N H' &{}= 2N \bar{\mu}_{q_0} \Delta_0 \mbo{\nu}' + N \bar{\mu}_q h_{AB}\leftexp{(h)}{\Gamma}^A_{CD} (U) q^{ab}_0 \ptl_a U^C \ptl_b U^D U'^B \notag\\
&\quad{}- N \bar{\mu}_{q_0} \ptl_{U^C} h_{AB} (U) q^{ab}_0 \ptl_a U^A \ptl_b U^C U'^B + \bar{\mu}_{q_0} q^{ab}_0 h_{AB} \ptl_a U'^A \ptl_b U^B \notag\\
&{}= 2N \bar{\mu}_{q_0} \Delta_0 \mbo{\nu}'+ \ptl_b \bigl(N \bar{\mu}_{q_0} q_0^{ab}h_{AB} \ptl_a U^A U'^B\bigr) \notag\\
&{}= \ptl_b \bigl(-2 \bar{\mu}_{q_0} q_0^{ab} \ptl_a N \mbo{\nu}' + 2 \ptl_a \mbo{\nu}' + N \bar{\mu}_q q^{ab} h_{AB} \ptl_a U^A U'^B\bigr),
\end{align}
where we have used \eqref{var-from-chris} and the background field equations \eqref{Kerr-p-dot} and \eqref{kerr-maximal}.
Fundamentally, underlying the statement \eqref{LS-computation} is the fact that $(N, 0)^{\mathsf{T}}$ is the kernel of the adjoint of the constraint map of our linear perturbation theory.
\end{proof}

\section[A positive-definite Hamiltonian energy from negative curvature of the target\\ and~the~Hamiltonian dynamics]{A positive-definite Hamiltonian energy from negative\\ curvature of the target and the Hamiltonian dynamics}\label{section4}

In arriving at the variational principles and their corresponding field equations, we have used smooth compactly supported deformations. In the construction of a Hamiltonian energy function the underlying computations are bit more subtle, in connection with the boundary terms and the initial value problem. We impose the regularity conditions on the axis of initial hypersurface~$\Sigma$ by fiat, so that the fields smoothly lift up to the original $\olin{\Sigma}$. We shall assume the following conditions on the two disjoint segments of the axes $\Gamma = \Gamma_1 \cup \Gamma_2$. In the wave map $U \colon M \to \mathbb{N}$, one of the components corresponds to the norm of the Killing vector $\vert \Phi \vert$ and the other the~`twist'. For the twist component, we assume
\begin{align*}
U'^A\big\vert_{\Gamma_1} = U'^A \big\vert_{\Gamma_2} \qquad \text{for the corresponding $A$}
\end{align*}
on account of our assumption that the perturbation of the angular momentum is zero. Without (effective) loss of generality we assume,
\begin{align*}
U'^A = 0 \qquad \text{on} \ \Gamma \ \text{which implies} \ \ptl_{\vec{t}}\, U'^A=0,
\end{align*}
where \smash{$\ptl_{\vec{t}}$} is the derivative tangent to the axis. To prevent conical singularity on the axis, which, as we remarked, allows us to smoothly lift our fields up to the original manifold $\olin{\Sigma}$, we assume
\begin{align*}
\vert \Phi \vert' =0, \qquad \ptl_{\vec{t}} \, \vert \Phi \vert' =0
\end{align*}
for the `norm' component of $U$, and
\begin{align*}
\ptl_{\vec{n}} U'^A =0, \qquad p'_A =\ptl_{\vec{t}} p'_A =0, \qquad \ptl_{\vec{n}} p'_A =0,
\end{align*}
where $\ptl_{\vec{n}}$ is the derivative normal to the axes $\Gamma$.
In this work, for $\mbo{\nu}'$ we shall assume
\begin{align*}
\ptl_{\vec{n}} \mbo{\nu}' =0,
\end{align*}
which corresponds to the preservation of the condition that inner (horizon) boundary is a minimal surface. Now define an `alternative' Hamiltonian constraint \smash{$H'^{{\rm Alt}}$}
\begin{align*}
H'^{{\rm Alt}} := \bar{\mu}_{q_0} \bigl(2 \Delta_0 \mbo{\nu}' + h_{AB}(U) U'^B \bigl(\Delta_0 U^A + \leftexp{(h)}{\Gamma}^A_{BC} q^{ab}_0 \ptl_a U^B \ptl_b U^C\bigr)\bigr),
\end{align*}
where we have now used the following identity to transform from $H'$:
\begin{align*}
&\halb \ptl_{U^C} h_{AB} q^{ab} \ptl_a U^A \ptl_b U^B U'^C + h_{AB}(U) q^{ab} \ptl_a U'^A \ptl_b U^B \notag\\
&\qquad{}= h_{AB} U'^B \bigl(\Delta_q U'^A + \leftexp{(h)}{\Gamma}^A _{CD} q^{ab} \ptl_a U^C \ptl_b U^D\bigr),
\end{align*}
which is analogous to \eqref{christof-intro}, but now for the $q$ metric. 
Analogously define
\begin{align*}
H''^{{\rm Alt}} &:= \bar{\mu}^{-1}_{q_0} \bigl(2 {\rm e}^{-2 \mbo{\nu}} \Vert \varrho' \Vert^2_{q_0} - \tau'^2 {\rm e}^{2 \mbo{\nu}} \bar{\mu}^2_{q_0} + p'_A p'^A\bigr) \notag\\
&\ \quad{}- \bar{\mu}_{q_0} h_{AB} U'^B\bigl(\leftexp{(h)}{\Delta} U'^A +R^A_{BCD} q_0^{ab} \ptl_a U^B \ptl_b U^C U'^D\bigr).
\end{align*}
Using a further divergence identity
\begin{align*}
& \leftexp{(h)}{\grad}_a \bigl(h_{AB} U'^B \leftexp{(h)}{\grad} ^a U'^A\bigr) - h_{AB} U'^B \leftexp{(h)}{\grad}_a \leftexp{(h)}{\grad}^a U'^A \notag\\
&\qquad{}= h_{AB} q^{ab} \leftexp{(h)}{\grad}_a U'^A \leftexp{(h)}{\grad}_b U'^B, \qquad (\Sigma, q),
\end{align*}
let us now define our `regularized' Hamiltonian energy density as
\begin{gather}
\mathbf{e}^{{\rm Reg}} := N \bar{\mu}^{-1}_{q_0} {\rm e}^{-2 \mbo{\nu}} \biggl(\Vert \varrho' \Vert^2_{q_0} + \halb p'_A p'^A \biggr) - \halb N {\rm e}^{2 \mbo{\nu}} \bar{\mu}_{q_0} \tau'^2 \notag\\ \hphantom{\mathbf{e}^{{\rm Reg}} :=}{}
+ \halb N \bar{\mu}_{q_0} q_0^{ab} h_{AB}(U) \leftexp{(h)}{\grad}_a U'^A \leftexp{(h)}{\grad}_b U'^B \notag\\ \hphantom{\mathbf{e}^{{\rm Reg}} :=}{}
- \halb N \bar{\mu}_{q_0} q^{ab}_0 h_{AE} (U) U'^A \leftexp{(h)}R^E _{BCD} \ptl_a U^B \ptl_b U^C U'^D \label{e-reg-den}
\end{gather}
and the `regularized' Hamiltonian $H^{{\rm Reg}}$
\begin{align} \label{e-reg}
H^{{\rm Reg}} := \int_{\Sigma} \mathbf{e}^{{\rm Reg}} \, {\rm d}^2 x.
\end{align}
It is evident that $H^{{\rm Reg}}$ is manifestly positive-definite in the maximal gauge $\tau' \equiv 0$, in view of the fact that the target is the (negatively curved) hyperbolic 2-plane. Indeed, we obtain a~similar energy expressions \eqref{e-reg-den} and \eqref{e-reg} in the higher-dimensional case where the target for wave maps is ${\rm SL}(n-2)/{\rm SO}(n-2)$.
As we already alluded to, the purpose of distinguishing the quantities~$H^{{\rm Reg}}$ is that they are transformed, using divergence identities, from $H$, and thus differ by boundary terms. In case the perturbations are compactly supported in $(\Sigma)$ it is immediate that they have the same value. In general, dealing with all the boundary terms and their evolution in our problem is considerably subtle (see Section~\ref{section5}).

 The aim of this work is the construction of the positive-definite energy functional $H^{{\rm Reg}}$. In~the~following, we shall establish the validity and consistency of our approach to construct the energy using two separate methods. Firstly, we shall show that $H^{{\rm Reg}}$ serves as a Hamiltonian that drives the dynamics of the unconstrained or `independent' phase-space variables. Secondly, we shall establish that there exists a spacetime divergence-free vector density, whose flux through~$\Sigma$~is~$H^{{\rm Reg}}$.

We would like to point out that, in our problem, the $2+1$ geometric phase space variables \big(e.g., $\mbo{\nu}', \varrho'^a_b$\big) are completely determined by the constraints and gauge-conditions. Therefore, their Hamiltonian dynamics are governed by the `independent' or `unconstrained' dynamical variables $\bigl(U'^A, p'_A\bigr)$. In the following, we shall prove that $H^{{\rm Reg}}$ drives the \emph{coupled} Einstein-wave map dynamics of $\bigl(U'^A, p'_A\bigr)$.

\begin{Theorem}
Suppose the globally regular, maximal development of the smooth, compactly supported perturbation initial data in the domain of outer communications of the Kerr metric is such that $\bigl\{ \bigl(q'_{ab}, \mbo{\pi}'^{ab}\bigr), \bigl(U'^A, p'_A\bigr) \bigr\}_{t} \in \mathscr{C}_{H'} \cap \mathscr{C}_{H'_a} \cap \mathscr{C}_{\tau'}$, then functional $H^{{\rm Reg}}$ is a Hamiltonian for the coupled dynamics of $\bigl(U'^A, p'_A\bigr)$,
\begin{align*}
&D_{p'_A} \cdot H^{{\rm Reg}} = \ptl_t U'^A, \qquad
D_{U'^{A}} \cdot H^{{\rm Reg}} =- \ptl_t p'_A.
\end{align*}
\end{Theorem}
\begin{proof}
The first variation $D_{p'_A} \cdot H^{{\rm Reg}}$ contains the terms
\begin{align*}
\varrho'^a_b \varrho''^a_b &{}= \halb N^{-1} \bar{\mu}_q \bigl(\leftexp{(q)}{\grad}^b N'_a + \leftexp{(q)}{\grad}_a N'^b - \delta_a^b \leftexp{(q)}{\grad}_c N'^c\bigr) \notag\\
&\quad{}{} \times \bigl(\leftexp{(q)}{\grad}^a N''_b + \leftexp{(q)}{\grad}_b N''^a - \delta_b^a \leftexp{(q)}{\grad}_c N''^c\bigr) \notag\\
&{}= \halb N^{-1} \bar{\mu}_q \leftexp{(q)}{\grad}_a N'^b \bigl(\leftexp{(q)}{\grad}^a N''_b + \leftexp{(q)}{\grad}_b N''^a - \delta_b^a \leftexp{(q)}{\grad}_c N''^c\bigr).
\end{align*}
We have the divergence identity
\begin{gather*}
 \leftexp{(q)}{\grad}_a \bigl(N^{-1} N'^b \bar{\mu}_q \bigl(\leftexp{(q)}{\grad}^a N''_b + \leftexp{(q)}{\grad}_b N''^a - \delta_b^a \leftexp{(q)}{\grad}_c N''^c\bigr)\bigr) \notag\\
\qquad\quad{}\times N^{-1} N'^b \bigl(\leftexp{(q)}{\grad}_a \bigl(\leftexp{(q)}{\grad}^a N''_b + \leftexp{(q)}{\grad}_b N''^a - \delta_b^a \leftexp{(q)}{\grad}_c N''^c\bigr)\bigr) \notag\\
\qquad\quad{}+ N^{-1} \bar{\mu}_q \leftexp{(q)}{\grad}_a N'^b \bigl(\leftexp{(q)}{\grad}^a N''_b + \leftexp{(q)}{\grad}_b N''^a - \delta_b^a \leftexp{(q)}{\grad}_c N''^c\bigr) \notag\\
\qquad{}= -N'^b \bigl(\ptl_b U^A p''_A\bigr) + N^{-1} \bar{\mu}_q \leftexp{(q)}{\grad}_a N'^b \bigl(\leftexp{(q)}{\grad}^a N''_b + \leftexp{(q)}{\grad}_b N''^a - \delta_b^a \leftexp{(q)}{\grad}_c N''^c\bigr)
\end{gather*}
after using the momentum constraint, and
\begin{align*}
D_{p'_A} \cdot \halb p'_A p'^A = p'^A.
\end{align*}
Collecting the terms above gives the Hamilton equation
\begin{align*}
D_{p'_A} \cdot H^{{\rm Reg}} = N \bar{\mu}^{-1}_q p'^A + N'^b \ptl_b U^A
= \ptl_t U'^A.
\end{align*}
The quantity $D_{U'^A} \cdot H^{{\rm Reg}}$ contains the terms
\[
D_{U'^A} \cdot \halb h_{AB}(U) \leftexp{(h)}{\grad}_a U'^A \leftexp{(h)}{\grad}_b U'^B = N \bar{\mu}_q q^{ab} h_{AB} \leftexp{(q)}{\grad}_a U''^A \leftexp{(q)}{\grad}_b U'^B \]
note that
\begin{align*}
\leftexp{(q)}{\grad}_a \bigl(N \bar{\mu}_q q^{ab} h_{AB} U''^A \ptl_b U'^B\bigr)&{}= U''^A \leftexp{(q)}{\grad}_a \bigl(N \bar{\mu}_q q^{ab} h_{AB} \ptl_a U'^B\bigr) \\
&\quad{}+ N \bar{\mu}_q q^{ab} h_{AB} \leftexp{(q)}{\grad}_a U''^A \leftexp{(q)}{\grad}_b U'^B
\end{align*}
and
\begin{align*}
&D_{U'^A} \cdot \biggl(- \halb N \bar{\mu}_{q} q^{ab} h_{AE} (U) U'^A \leftexp{(h)}R^E _{BCD} \ptl_a U^B \ptl_b U^C U'^D\biggr) \notag\\
&\qquad{}= - N \bar{\mu}_q q^{ab} h_{AE} (U) \leftexp{(h)}R^E _{BCD} \ptl_a U^B \ptl_b U^C U'^D,
\end{align*}
which combine to give
\begin{align*}
D_{U'^A} \cdot H^{{\rm Reg}} &{}= - \leftexp{(q)}{\grad}_a \bigl(N \bar{\mu}_q q^{ab} h_{AB} \ptl_a U'^B\bigr) \notag\\
&\quad{}- N \bar{\mu}_q q^{ab} h_{AE} (U) \leftexp{(h)}R^E _{BCD} \ptl_a U^B \ptl_b U^C U'^D
\intertext{which is the Hamilton equation}
&{}= -\ptl_t p'_A.
\tag*{\qed}
\end{align*}
\renewcommand{\qed}{}
\end{proof}

\begin{Theorem}
Suppose the variables $\bigl\{ \bigl(q'_{ab}, \mbo{\pi}'^{ab}\bigr), \bigl(U'^A, p'_A\bigr) \bigr\} \in \mathscr{C}_{H'} \cap \mathscr{C}_{H'_a} \cap \mathscr{C}_{\tau'}$, then there exists a $($spacetime$)$ divergence-free vector field density such that its flux through $t-$constant hypersurfaces is $H^{{\rm Reg}}$ $($positive-definite$)$.
\end{Theorem}
\begin{proof}
In the proof we shall use the perturbation evolution equations and the background (Kerr metric) field equations. Consider \smash{$\ptl_t \mathbf{e}^{{\rm Reg}}$} and it contains the following terms:
\begin{gather}
1. \ \ N \bar{\mu}^{-1}_q p'^A \ptl_t p'_A= N \bar{\mu}^{-1}_q p'^A \bigl(h_{AB} \leftexp{(h)}{\grad}_a \bigl(N \bar{\mu}_q q^{ab} \leftexp{(h)}{\grad}_b U'^B\bigr) \notag\\
\hphantom{1. \ N \bar{\mu}^{-1}_q p'^A \ptl_t p'_A=}{} + N \bar{\mu}_q h_{AB} R^E_{\,\,\,\, BCD} q^{ab} \ptl_a U^B \ptl_b U^D U'^C \bigr),\label{term1}
\\
2. \ \ N \bar{\mu}_q q^{ab} \leftexp{(h)}{\grad}_a \bigl(\ptl_t U'^A\bigr) \leftexp{(h)}{\grad}_b U'^B \notag\\
 \qquad{}= N \bar{\mu}_q q^{ab} h_{AB}(U) \leftexp{(h)}{\grad}_b U'^B \bigl(\leftexp{(h)}{\grad}_a \bigl(\bar{\mu}^{-1}_q Np'^A + \mathcal{L}_{N'} U^A\bigr)\bigr).\label{term2}
 \end{gather}
 Note the divergence relation involving the terms from \eqref{term1} and \eqref{term2},
 \begin{align*}
 \leftexp{(h)}{\grad}_a \bigl(N^2 q^{ab} h_{AB} p'^A \leftexp{(h)}{\grad}_b U'^B\bigr) &{}= \bar{\mu}^{-1}_q N p'^A \leftexp{(h)}{\grad}_a \bigl(N \bar{\mu}_q q^{ab} \leftexp{(h)}{\grad}_b U'^B\bigr) \notag\\
 & \quad{}+ N \bar{\mu}_q q^{ab} h_{AB} \leftexp{(h)}{\grad}_b U'^B \leftexp{(h)}{\grad}_a \bigl(N \bar{\mu}^{-1}_q p'^A\bigr)
 \end{align*}
\begin{gather*}
3. \ \ N \bar{\mu}_q h_{AE}(U) \ptl_t U'^A R^A_{\,\,\,\,BCD} q^{ab} \ptl_a U^B \ptl_b U^D U'^C \notag\\
 \qquad{}=N \bar{\mu}_q h_{AE} \bigl(\bar{\mu}^{-1}_q N p'_A + \mathcal{L}_{N'} U^A\bigr) R^E_{\,\,\,BCD} q^{ab} \ptl_a U^B \ptl_b U^D U'^C
\\
 4. \ \ {\rm e}^{-2 \mbo{\nu}} N \bar{\mu}^{-1}_{q_0} \varrho'^c_a \ptl_t \varrho'^a_c \\
 \qquad{} = N \bigl(\leftexp{(q)}{\grad}_a N'^b + \leftexp{(q)}{\grad}_a N'^b - \delta^b_a \leftexp{(q)}{\grad}_c N'^c\bigr) \\
\qquad\quad{}\times\biggl( N \bar{\mu}_q \biggl(q^{ac} \delta^d_b - \halb \delta^a_b q^{cd}\biggr) \biggl(h_{AB} \ptl_b U'^A \ptl_d U^B + \halb \ptl_{U}h_{AB} \ptl_b U'^A \ptl_d U^B \biggr) \\
\qquad\quad\hphantom{\times\biggl(\,}{} + 2 \bar{\mu}_q q^{ac}\ptl_c N \ptl_b \mbo{\nu}'
 - \bar{\mu}_q \delta^a_b q^{cd} \ptl_c N \ptl_d \mbo{\nu}'\biggr) \\
\qquad{} = \mathcal{L}_{N'} \bigl(\bar{\mu}_{q_0} q^{ab}_0\bigr) \biggl(h_{AB} \ptl_a U'^A \ptl_b U^B + \halb \ptl_{U^C} h_{AB} \ptl_a U^A \ptl_b U^B U'^C -2 \ptl_a N \ptl_b \mbo{\nu}'\biggr).
\end{gather*}

Consider the following divergence identities:
\begin{subequations} \label{div-ident-mom}
\begin{gather}
\leftexp{(q_0)}{\grad}^a \bigl(N'^b \ptl_a N \ptl_b \mbo{\nu}' \bar{\mu}_{q_0}\bigr)= \bar{\mu}_{q_0} \leftexp{(q_0)}{\grad}^a N'^b \ptl_a N \ptl_b \mbo{\nu}' + N'^b \leftexp{(q_0)}{\grad}^a \bigl(\bar{\mu}_{q_0} \ptl_a N \ptl_b \mbo{\nu}'\bigr) \\
\leftexp{(q_0)}{\grad}^b \bigl(N'^a \ptl_a N \ptl_b \mbo{\nu}' \bar{\mu}_{q_0}\bigr)= \bar{\mu}_{q_0} \leftexp{(q_0)}{\grad}^b N'^a \ptl_a N \ptl_b \mbo{\nu}' + N'^a \leftexp{(q_0)}{\grad}^b \bigl(\bar{\mu}_{q_0} \ptl_a N \ptl_b \mbo{\nu}'\bigr) \\
\leftexp{(q_0)}{\grad}_c \bigl(N'^c q^{ab}_0 \ptl_a N \ptl_b \mbo{\nu}' \bar{\mu}_{q_0}\bigr)= \leftexp{(q_0)}{\grad}_c N'^c\bigl(q^{ab}_0 \ptl_a N \ptl_b \mbo{\nu}' \bar{\mu}_{q_0}\bigr)\nonumber\\ \hphantom{\leftexp{(q_0)}{\grad}_c \bigl(N'^c q^{ab}_0 \ptl_a N \ptl_b \mbo{\nu}' \bar{\mu}_{q_0}\bigr)=}{}
+ N'^c \leftexp{(q_0)}{\grad}_c \bigl(q^{ab}_0 \ptl_a N \ptl_b \mbo{\nu}' \bar{\mu}_{q_0}\bigr).
\end{gather}
\end{subequations}
Based on the right-hand sides of the divergence identities in \eqref{div-ident-mom}, we get after using the background field equation \eqref{kerr-maximal}
\begin{gather*}
-2 \mathcal{L}_{N'} \bigl(\bar{\mu}_q q^{ab}\bigr) \ptl_a N \ptl_b \mbo{\nu}' \notag\\
\qquad{}= -2 N'^b \ptl_b N \ptl_a \bigl(\bar{\mu}_q q^{ab} \ptl_a \mbo{\nu}'\bigr) + 2 \leftexp{(q_0)}{\grad}_a \bigl(N'^c \ptl_c \mbo{\nu}' \bar{\mu}_{q_0} q^{ab}_0 \ptl_b N\bigr) \notag\\
\qquad\quad{} + 2 \leftexp{(q_0)}{\grad}_a \bigl(N'^c \ptl_c N \bar{\mu}_{q_0} q^{ab}_0 \ptl_b \mbo{\nu}'\bigr) - 2 \leftexp{(q_0)}{\grad}_c \bigl(N'^c \ptl_a \mbo{\nu}' \bar{\mu}_{q_0} q^{ab}_0 \ptl_b N\bigr) \notag\\
\qquad{}= \mathcal{L}_{N'} N \biggl(- H' + h_{AB} \ptl_a U'^A \ptl_b U^B + \halb \ptl_{U^C} h_{AB} \ptl_a U^A \ptl_b U^B U'^C\biggr) \notag\\
\qquad\quad{} + 2 \leftexp{(q_0)}{\grad}_b \bigl(\mathcal{L}_{N'} \mbo{\nu}' \bar{\mu}_{q_0} \ptl^b N+ \mathcal{L}_{N'} N \bar{\mu}_{q_0} \ptl^b \mbo{\nu'} - \bar{\mu}_{q_0} N'^b \ptl_a N \ptl^a \mbo{\nu}'\bigr).
\end{gather*}
Now let us focus on the remaining `shift' terms:
\begin{align}
&N \bar{\mu}_q q^{ab} h_{AB} \leftexp{(h)}{\grad}_b U'^B \leftexp{(h)}{\grad}_a \bigl(\mathcal{L}_{N'} U^A\bigr),\nonumber\\
&- N \bar{\mu}_q h_{AE} \bigl(\mathcal{L}_{N'} U^A\bigr) R^E_{\,\,\,BCD} q^{ab} \ptl_a U^B \ptl_b U^D U'^C, \nonumber\\
& \mathcal{L}_{N'} \bigl(\bar{\mu}_{q_0} q^{ab}_0\bigr) \biggl(h_{AB} \ptl_a U'^A \ptl_b U^B + \halb \ptl_{U}h_{AB} \ptl_a U'^A \ptl_b U^B\biggr).\label{grad-square}
\end{align}
Consider the quantity $N \bar{\mu}_q q^{ab} h_{AB} \ptl_a U'^A \bigl(\mathcal{L}_{N'}\bigl(\ptl_b U^B\bigr)\bigr)$ that occurs in \eqref{grad-square}, we have
\begin{gather*}
N \bar{\mu}_q q^{ab} h_{AB} \ptl_a U'^A \bigl(\mathcal{L}_{N'}\bigl(\ptl_b U^B\bigr)\bigr) + \ptl_{U^C} h_{AB} \mathcal{L}_{N'} U^C N \bar{\mu}_qq^{ab}\ptl_a U'^A \ptl_b U^B \notag\\
\qquad{}= \ptl_a U'^A \mathcal{L}_{N'} \bigl(N \bar{\mu}_q q^{ab} h_{AB}(U) \ptl_b U^B\bigr) \notag\\
\quad\qquad{} -\mathcal{L}_{N'} N \bigl(h_{AB} \bar{\mu}_q q^{ab} \ptl_b U'^B\bigr) - \mathcal{L}_{N'} \bigl(\bar{\mu}_q q^{ab}\bigr) N h_{AB} \ptl_a U'^A\ptl_b U^B,
\end{gather*}
likewise
\begin{gather*}
 U'^A \mathcal{L}_{N'} \bigl(\ptl_b \bigl(N \bar{\mu}_q q^{ab} h_{AB} \ptl_a U^B\bigr) \bigr) \notag\\
\qquad{}= U'^A \mathcal{L}_{N'} h_{AB}\ptl_a \bigl(N \bar{\mu}_q q^{ab} \ptl_b U^B\bigr) + U'^A \mathcal{L}_{N'} \bigl(N \bar{\mu}_q q^{ab} \ptl_b U^B \ptl_{U^C} h_{AB} \ptl_a U^C\bigr).
\end{gather*}

Collecting the terms above, while using the background field equations \eqref{Kerr-p-dot} and computations analogous to the ones in Section~\ref{section3} and
\begin{align*}
& \ptl_a \bigl(U'^A \mathcal{L}_{N'} \bigl(N \bar{\mu}_q q^{ab} h_{AB} \ptl_a U^A\bigr)\bigr) \notag\\
&\qquad{} = \ptl_a U'^A \mathcal{L}_{N'} \bigl(N \bar{\mu}_q q^{ab} h_{AB} \ptl_a U^A\bigr) + \mathcal{L}_{N'} \bigl(\ptl_a \bigl(N \bar{\mu}_q q^{ab} h_{AB} \ptl_a U^A\bigr) \bigr) \notag \\
&\qquad{} = \ptl_a U'^A \mathcal{L}_{N'} \bigl(N \bar{\mu}_q q^{ab} h_{AB} \ptl_a U^A\bigr) + \ptl_a \bigl(\mathcal{L}_{N'} \bigl(N \bar{\mu}_q q^{ab} h_{AB} \ptl_a U^A\bigr)\bigr),
\end{align*}
 we have
\begin{align*}
\ptl_t \mathbf{e}^{{\rm Reg}} &{}= \ptl_b \bigl(N^2 \bar{\mu}^{-1}_q \bigl(\bar{\mu}_{q_0} q_{0}^{ab} p'_A \ptl_a U'^A\bigr) + U'^A \mathcal{L}_{N'} \bigl(\bar{\mu}_{q_0} q^{ab}_0 h_{AB} \ptl_b U^B\bigr) \bigr) \notag\\
 &\quad{}\times \mathcal{L}_{N'} (N) \bigl(2 \bar{\mu}_{q_0} q_{0}^{ab} \ptl_a \mbo{\nu}' + 2 \mathcal{L}_N \mbo{\nu}' \bar{\mu}_q q^{ab} \ptl_a N - 2 N'^b \bar{\mu}_q q^{bc} \ptl_a \mbo{\nu}' \ptl_c N\bigr) - H' \mathcal{L}_{N'}(N)
 \intertext{for $ \bigl\{ \bigl(q'_{ab}, \mbo{\pi}'^{ab}\bigr), \bigl(U'^A, p'_A\bigr) \bigr\} \in \mathscr{C}_{H'}$ this reduces to }
 &{}= \ptl_b \bigl(N^2 \bar{\mu}^{-1}_q \bigl(\bar{\mu}_{q_0} q_{0}^{ab} p'_A \ptl_a U'^A\bigr) + U'^A \mathcal{L}_{N'} \bigl(\bar{\mu}_{q_0} q^{ab}_0 h_{AB} \ptl_b U^B\bigr) \bigr) \notag\\
 &\quad{}\times \mathcal{L}_{N'} (N) \bigl(2 \bar{\mu}_{q_0} q_{0}^{ab} \ptl_a \mbo{\nu}'\bigr) + 2 \mathcal{L}_N \mbo{\nu}' \bar{\mu}_q q^{ab} \ptl_a N - 2 N'^b \bar{\mu}_q q^{bc} \ptl_a \mbo{\nu}' \ptl_c N .
\end{align*}
Thus, if we define
\begin{align*}
&(J^ t)^ {{\rm Reg}} := \mathbf{e}^{{\rm Reg}}, \notag\\
&(J^b)^{{\rm Reg}} := N^2 {\rm e}^{-2 \mbo{\nu}}\bigl(q_{0}^{ab} p'_A \ptl_a U'^A\bigr)+ U'^A \mathcal{L}_{N'} \bigl(N \bar{\mu}_{q_0} q^{ab}_0 h_{AB} \ptl_b U^B\bigr) \notag\\
 &\hphantom{(J^b)^{{\rm Reg}} :=}{}\times \mathcal{L}_{N'} (N) \bigl(2 \bar{\mu}_{q_0} q_{0}^{ab} \ptl_a \mbo{\nu}'\bigr) + 2 \mathcal{L}_{N'} \mbo{\nu}' \bar{\mu}_q q^{ab} \ptl_a N - 2 N'^b \bar{\mu}_q q^{bc} \ptl_a \mbo{\nu}' \ptl_c N,
\end{align*}
it follows that $J^{{\rm Reg}}$ is a divergence-free vector field density for
\[
\big \{ (q'_{ab}, \mbo{\pi}'^{ab}), \bigl(U'^A, p'_A\bigr) \big \} \in \mathscr{C}_{H'} \cap \mathscr{C}_{H'_a} \cap \mathscr{C}_{\tau'}.
\tag*{\qed}
\]
\renewcommand{\qed}{}
\end{proof}

\section{Boundary behaviour of the dynamics in the orbit space}\label{section5}

In any dynamical physical theory, the concept of a conserved energy plays a fundamental role in the analysis of stability of its solutions. For instance, in the Lyapunov theory of stability, the notion of energy,
its positive-definiteness and its dynamical behaviour act as important basis for various notions of stability. Likewise, PDE techniques are typically based on a conserved and positive energy.

In the previous sections, we constructed a positive-definite `bulk' Hamiltonian energy functional for the axially symmetric linear perturbative theory of Kerr black hole spacetimes for the entire subextremal range $(\vert a \vert<M)$; and the associated variational principles.

We used a special `Weyl--Papapetrou' gauge that provides additional structure in Einstein's equations for general relativity:
\begin{align*}
R_{\mu \nu} =0, \qquad \mu, \nu=0, 1,2, 3, \qquad \bigl(\bar{M}, \bar{g}\bigr),
\end{align*}
when the $3+1$ Lorentzian spacetime $\bigl(\bar{M}, \bar{g}\bigr)$ admits a rotational isometry
\begin{align} \label{WP-gauge}
\bar{g} = \vert \Phi \vert^{-1} g + \vert \Phi \vert ({\rm d}\phi + A_\nu {\rm d}x^\nu)^2, \qquad \nu = 0,1, 2, \qquad \text{(Weyl--Papapetrou gauge)},
\end{align}
where $g$ is the Lorentzian metric on the orbit space $M := \bar{M}/{\rm SO}(2)$. $A$, $g$ are independent of the parameter $\phi$ corresponding to the rotationally symmetric Killing vector $\Phi :=\ptl_\phi$, whose spacetime norm squared is represented as $\vert \Phi \vert$.

As we already discussed, the effects of the geometry of the Kerr black hole spacetime manifest themselves in the Kerr black hole stability problem in a fundamental way. In particular, as~a~result of the ergo-region the perturbations (scalar wave, Maxwell and linearized gravity) do not necessarily have a positive-definite and conserved energy. This significantly limits the immediate use of the standard PDE techniques in establishing the boundedness and asymptotic decay of the perturbations.

A standard technique in the literature to construct a positive-definite energy functional is to consider an energy current from the linear combination of the time-translational $\ptl_t$ and rotational~$\ptl_\phi$ vector fields
\begin{align*}
\ptl_ t + \chi \ptl_\phi.
\end{align*}
 In view of the fact that the corresponding energy is not necessarily conserved, a separate, intricate Morawetz spacetime integral estimate is needed to control this energy in time. Moreover, these methods are suitable for small angular-momentum $\vert a \vert \ll M$.

Thus, for our Hamiltonian energy to be efficacious, it is important to establish that it is strictly conserved in time. In the process of construction of the positive-definite Hamiltonian energy in our work, we pick up several boundary terms, at various stages. From both conceptual and technical perspectives these boundary terms play an important role in our dynamical stability problem. These issues do not arise in the corresponding black hole uniqueness theorems of stationary black holes. For the convenience of the reader, let us elaborate on why it is fundamental to rigorously understand the behaviour of these boundary terms and provide a~motivation for this article.

 As we already remarked, there is a possibility that the energy density in the original construction can be made to be locally negative. The positive-definite energy is constructed using the transformations that involve boundary terms from a Hamiltonian energy with an indefinite sign, while preserving the `bulk' Hamiltonian structure of the equations. Therefore, in our argument it is fundamental to rule out the `negativity' or the ambiguity of sign does not `secretly get hidden' in a plethora of boundary terms that occur in both the Carter--Robinson type identities and the linearization stability methods.

 In the usual PDE theory, we perform a variation with respect to smooth compactly supported~$C^\infty_0$ `test functions' to evaluate the field equations and understand the critical points. However, in the Einstein equations one cannot ignore the boundary terms, because they can have a physical and geometric interpretation.

 From the perspective of calculus of variations, a positive-definite second variation mass-energy corresponds to the mass of the Kerr black hole being a minimizer, for fixed angular momentum, in the space of admissible metrics. This interpretation combines well with the mass-angular momentum inequalities for axisymmetric spacetimes \cite{D09}. However, as we already discussed, the boundary terms cannot be ignored in this interpretation.

 Most PDE techniques in establishing the uniform boundedness and decay, work at their natural best if there exists a positive-definite and (strictly) conserved energy, that goes together with the evolution of the PDEs. In case there exists a residual, dynamically non-vanishing surface integral (especially with an indefinite sign), together with a positive-definite `bulk' Hamiltonian~\eqref{e-reg} and~\eqref{e-reg-den}, then this could significantly impede the efficacy of our positive-definite `bulk' Hamiltonian \eqref{e-reg} and \eqref{e-reg-den} in PDE methods that establish uniform-boundedness and decay.

The Weyl--Papapetrou gauge \eqref{WP-gauge}, although it plays a fundamental role in the construction of our bulk Hamiltonian energy \eqref{e-reg} and \eqref{e-reg-den}, presents significant regularity issues at the axes and at the infinity. This is due to the behaviour of the rotationally symmetric Killing vector field, i.e., \smash{$\vert \Phi \vert^{-1}$} blows up at the axes $\Gamma$ and $\vert \Phi \vert$ blows up at the infinity $\iota^0$. These expressions routinely occur in our formulas due to the form of the Weyl--Papapetrou gauge \eqref{WP-gauge}. It may be noted that these regularity issues manifest themselves in the form of boundary terms that occur precisely at these boundaries $\Gamma$, $\iota^0$ and $\mathcal{H}^+$.

The Einstein equations in the Weyl--Papapetrou gauge are not purely hyperbolic in nature. In particular, the Weyl--Papapetrou gauge has coupled elliptic-hyperbolic PDE structure. This causes gauge-related causality issues. Even if we start with a compactly supported initial data, away from the boundaries $\Gamma$, $\iota^0$, $\mathcal{H}^+$, thereby ensuring the regularity and the vanishing of boundary flux integrals initially, it is not necessary that the fields stay compactly supported at later times. In other words, `pure gauge' perturbations and the associated boundary terms, can `kick-in' in the asymptotic regions, away from the causal future of the support of initial data, at later times, thus affecting the boundary behaviour of the dynamics.

On account of the boundary related issues and complications mentioned above, the question whether the Weyl--Papapetrou gauge is even compatible with the axially symmetric evolution problem of the Einstein equations can legitimately arise. In this work, we shall provide a~favourable answer to this question. We explain how we overcome these complications and discuss our methods in the following.

We formulate the initial value problem of the linearized Einstein equations in the harmonic gauge and transform the solutions to the Weyl--Papapetrou gauge. In the harmonic gauge, the constraints and the gauge condition propagate in time automatically, as long they are satisfied on the initial data. Indeed, in the harmonic gauge the linearized Einstein equations system is purely a hyperbolic system of equations, which in turn implies propagation of regularity and causality, from standard theory of hyperbolic PDE. We take advantage of the global regularity in harmonic gauge including at the axes and infinity in harmonic gauge. We would like to point out that, when we make such gauge transformations, we can lose regularity but in our construction we show in fact that there exists a $(C^\infty$-)diffeomorphism.
 In our work, we also construct a~variational formulation of the linearized Einstein equations in harmonic gauge, which may be of independent interest and fits well with the theme of our approach.

In view of such conceptual and technical subtleties, the `safest' way to prove that the positive-energy \smash{$H^{{\rm Reg}}$} we constructed is strictly conserved is to use the argument that the time-derivative of the energy \smash{$H^{{\rm Reg}}$} vanishes dynamically in time. In other words, as we have already shown that the time derivative of the energy density $\mathbf{e}^{{\rm Reg}}$ is a pure spatial divergence (a fact that is associated to $(N, 0)^{\mathsf{T}}$ being the kernel of the adjoint of the dimensionally reduced constraint~map)
\begin{align*}
&\bigl(J^ t\bigr)^ {{\rm Reg}}:= \mathbf{e}^{{\rm Reg}}, \notag\\
&\bigl(J^b\bigr)^{{\rm Reg}}:= N^2 {\rm e}^{-2 \mbo{\nu}}\bigl(q_{0}^{ab} p'_A \ptl_a U'^A\bigr) + U'^A \mathcal{L}_{N'} \bigl(N\bar{\mu}_{q_0} q^{ab}_0 h_{AB} \ptl_b U^B\bigr) \notag\\
&\hphantom{\bigl(J^b\bigr)^{{\rm Reg}}:=}{}\times \mathcal{L}_{N'} (N) \bigl(2 \bar{\mu}_{q_0} q_{0}^{ab} \ptl_a \mbo{\nu}'\bigr) + 2 \mathcal{L}_{N'} \mbo{\nu}' \bar{\mu}_q q^{ab} \ptl_a N - 2 N'^b \bar{\mu}_q q^{bc} \ptl_a \mbo{\nu}' \ptl_c N,
\end{align*}
where $J^{{\rm Reg}}$ is a divergence-free vector field density, we need to prove that the fluxes of $J^b$ at the boundaries vanish dynamically in time.

Using Fredholm theory and that transverse-traceless $2$-tensors
vanish\footnote{This is in contrast with the $(2+1)$-dimensional relativistic Teichm\"uller theory (for, e.g., $\Sigma$ is compact and of genus $>1$), where transverse-traceless tensors play an important role (see~\cite{Moncrief_2007}).} for our geometry and topology, we reduce the elliptic operators into conformal Killing operators. Furthermore, benefiting from conformal invariance of these operators, we reduce the elliptic operators into (tensorial) Poisson equations. We are then able to obtain the desired decay rates of the fields using both the fundamental solution and Fourier methods.

In the fundamental solution approach, we construct regularity and decay rate in our orbit space geometry, using the method of images in such a way that the total `charge' of the source is zero. In the Fourier methods, the regularity conditions imply that the frequency corresponding to the logarithmic blow up of the solution does not occur. In both methods, we recover a faster decay rate than for the usual Poisson equation in two dimensions. This faster decay rate, in contrast to the translational symmetric dimensional reduction, plays a fundamental role in our work.

 The diffeomorphism invariance of the Einstein equations allows us the gauge-freedom. In~the ${3+1}$ (vis-\'a-vis $2+1$) picture, this gauge freedom is reflected in terms of the `lapse' and the `shift' vector field. After fixing the gauge, we estimate the behaviour of gauge dependent quantities in the asymptotic regions where the perturbations are assumed to be pure gauge, starting from the independent wave map phase space variables and then moving on to dependent phase space variables.

 Following the estimates on the fundamental phase space variables, we estimate the fluxes of each of the quantities in $J^b$ term by term and we prove that regularity holds and that they dynamically vanish at all the three boundaries $\Gamma$, $\iota^0$ and $\H^+$, including at the corner $\Gamma \cap \H^+$.

 In establishing the boundary behaviour of these terms,
 we pay special attention to the quantity $\mbo{\nu}$ that is related to a conformal factor. We note that $\mbo{\nu}'$ \emph{does not transform like a scalar}.
 This result may be of independent interest in conformal geometry. Secondly, we also establish that an integral quantity $Y_0(\mathcal{H})$ that vanishes for all times. These two results are crucial in resolving the regularity issues on the axes and at the corners.

In the current work, in keeping with the spirit of part I, we pay special attention to the covariance
of our analysis with respect to the target metric and not rely on a specific gauge on the target. Apart from the aesthetics, this can be used as a basis for studying the stability of higher $(n+1)$ dimensional black holes with toroidal symmetry $\mathbb{T}^{n-2}$.

The problem of stability of Kerr black hole spacetimes is currently an area of very active research among several groups worldwide. In this connection, there are several important and remarkable works (see, e.g., \cite{ABBM_19,LB_15_1, LB_15_2,HDR_17,DR_11, HHV_19, SMa_17_1, SMa_17_2, Tato_11}). Furthermore, there have been recent advances in the nonlinear black hole stability problem. In \cite{HDRT_21}, nonlinear stability of Schwarzschild black holes was announced. In \cite{GKS_22}, the full proof of slowly rotating black holes is announced (see also \cite{GKS_20, KS_21}). In \cite{HDR_16}, Dafermos--Holzegel--Rodnianski have established linear stability of Schwarzschild black hole spacetimes for gravitational perturbations, where a positive-definite energy played a fundamental role (see~\cite{GH_16}).
A positive-definite energy functional for both even and odd gravitational perturbations of Schwarzschild black holes was constructed by Moncrief \cite{Moncrief_74} using mode decomposition. We should point out that even in the case of Schwarzschild black holes, which do not contain the ergo-region, establishing the existence of a~positive-definite energy functional for gravitational is non-trivial.

In the case of axial symmetry, wave map behaviour in Kerr spaces has been studied in~\cite{IK_15}. As~we already pointed out, a positive-definite energy was first constructed by Dain--de Austria~\cite{DA_14} for extremal Kerr black holes $(\vert a \vert =M)$ using Brill mass formula \cite{B59, D09}.

 These works (except \cite{DA_14}) are dedicated to the stability of Kerr black hole spacetimes with `small' or `very small' angular momentum $(\vert a \vert \ll M)$. Relatively less is understood about the stability of Kerr black hole spacetimes for large, but subextremal angular momentum $ \vert a \vert<M$. This is mainly attributed due to the effects of the ergo-region that always surrounds a Kerr black hole spacetime with a non-vanishing angular momentum.

In the case of the large and sub-extremal angular-momentum of Kerr black hole spacetimes $ (\vert a \vert <M)$, the effects of the ergo-region become even more subtle and counterbalancing its effects to obtain uniform-boundedness and decay of propagating fields is even more difficult from a PDE perspective.
The decay of a linear wave equation on Kerr black hole spacetimes with $\vert a \vert$ is studied in the remarkable works \cite{DRS_16, FKSY_05,FKSY_06,FKSY_08}.

 We expect that our results will be useful to fill this gap for Maxwell and gravitational (i.e., Einsteinian) perturbations. Even among the methods for large $\vert a \vert<M$, our approach is different in the sense that the (Hamiltonian) flow of our phase space variables is restricted to positive-definite energy surfaces. Thus, we have a relation (equality) between the energy at different time levels without the need for a spacetime `bulk' integral or Morawetz estimate.

The more general class of $3+1$ Lorentzian spacetimes is the Kerr--Newman family of spacetimes which is a solution of coupled $3+1$ Einstein--Maxwell equations for general relativity. In~a~series of classical works \cite{Moncrief_74_1,Moncrief_74_2,Moncrief_74_3}, the stability of Reissner--Nordstr\"om spacetimes is studied for the entire sub-extremal range $\vert Q \vert <M$. Indeed, the fact that the Hamiltonian stability results of Reissner--Nordstr\"om spacetimes hold for the full sub-extremal range was an early indication that the current results and the Kerr--Newman results \cite{GM17_gentitle} are feasible. In view of the rigidity of the Reissner--Nordstr\"om spacetimes, `non-trivial' perturbations of the RN spacetimes belong to the Kerr--Newman family of black hole spacetimes.

We expect that the results in \cite{GM17_gentitle} shall allow us to venture in this direction. Equivalent results for the stability of Kerr--Newman--de Sitter spacetimes $(\vert a \vert, \vert Q \vert<M)$, which has different asymptotics and gauge issues compared to the Kerr--Newman problem in \cite{GM17_gentitle}, is a work in progress (cf.\ \cite{NG_17_2, NG_17_1} for now).

Hollands and Wald \cite{WH_13} have developed a notion of canonical energy, which was later extended by Prabhu--Wald \cite{WP_13}, where they showed that if the energy is not positive-definite for axisymmetric perturbations of Kerr black hole spacetimes, it would blow up at later times.
This work, based on the $2+1$ Einstein-wave map formalism and dimensional reduction in Weyl--Papapetrou gauge, confirms their criterion for axisymmetric stability.

\section{Global existence and propagation of regularity}\label{section6}

Suppose $\bigl(\bar{M}, \bar{g}\bigr)$ is a Lorentzian spacetime, then consider the Einstein--Hilbert action on $\bigl(\bar{M}, \bar{g}\bigr)$,
\begin{align*}
S_{\text{H}}[\bar{\gg}] := \int \bar{R}_{\bar{\gg}} \bar{\mu}_{\bar{\gg}}.
\end{align*}
Consider a curve $\mbo{\gamma}_s \colon [0,1] \to C(\bgg)$ where $C(\bgg)$ is a space of smooth Lorentzian metrics $\bgg$. We shall use the following notation $ \bar{\gg}' := D_{\mbo{\gamma}_s} \cdot \bar{\gg}$ and define
\begin{align*}
D_{\mbo{\gamma}_s} \cdot S_{\text{H}} \bigl[\bgg', \bgg\bigr] := \int \biggl( -\text{Ric}(\gg)^{\mu \nu} \bar{\gg}' + \halb R_\gg \gg^{\mu \nu} \gg'_{\mu \nu} + \leftexp{(\bar{\gg})}{\grad}^\mu \leftexp{(\bar{\gg})}{\grad}^\nu \bar{\gg}'_{\mu \nu} - \leftexp{(\bar{\gg})}{\grad}^2 {\rm tr}\bigl(\bgg'\bigr) \biggr) \bar{\mu}_\bgg.
\end{align*}
Now, consider another analogous curve $\mbo{\gamma}_\la \colon [0, 1] \to C (\bgg)$ and let us define $ \bgg'' := D^2_{\gamma_s \gamma_\lambda} \cdot \bgg $, we then have for the Kerr background metric. The functional \eqref{EH-second-variation} can be simplified as follows:
\begin{gather}
D^2_{\mbo{\gamma}_\la \mbo{\gamma}_s} \cdot S_{\text{H}} \bigl[\bgg'', \bgg', \bgg\bigr] := \int
\bigl( \bgg^{\mu \nu} \leftexp{(\bar{\gg})}{\grad}^\gm \leftexp{(\bar{\gg})}{\grad}^\d \bgg'_{\gm \d} - \bgg^{\mu \nu} \leftexp{(\bar{\gg})}{\grad}^\gm \leftexp{(\bar{\gg})}{\grad}_\gm {\rm tr} \bigl(\bgg'\bigr)- \leftexp{(\bar{\gg})}{\grad}_\gm \leftexp{(\bar{\gg})}{\grad}^\mu \bgg'^{\gm \nu}
\notag\\
 \hphantom{D^2_{\mbo{\gamma}_\la \mbo{\gamma}_s} \cdot S_{\text{H}} \bigl[\bgg'', \bgg', \bgg\bigr] := \int\bigl(}{}
- \leftexp{(\bar{\gg})}{\grad}_\gm \leftexp{(\bar{\gg})}{\grad}^\nu \bgg'^{\gm \mu}
+ \leftexp{(\bar{\gg})}{\grad}^\mu \leftexp{(\bar{\gg})}{\grad}^\nu {\rm tr}\bigl(\bgg '\bigr) \notag\\
 \hphantom{D^2_{\mbo{\gamma}_\la \mbo{\gamma}_s} \cdot S_{\text{H}} \bigl[\bgg'', \bgg', \bgg\bigr] := \int\bigl(}{}
+ \leftexp{(\bar{\gg})}{\grad}^\a \leftexp{(\bar{\gg})}{\grad}_\a \bgg'^{\mu \nu} \bigr) \halb \bgg'_{\mu \nu} \bar{\mu}_\bgg \notag\\
 \hphantom{D^2_{\mbo{\gamma}_\la \mbo{\gamma}_s} \cdot S_{\text{H}} \bigl[\bgg'', \bgg', \bgg\bigr] :=\int}{}
 + \leftexp{(\bar{\gg'})}{\grad}^\mu\bigl( \leftexp{(\bar{\gg})}{\grad}^\nu \bgg'_{\mu \nu} - \bgg^{\gm \d} \leftexp{(\bar{\gg})}{\grad}_\mu \bgg'_{\gm \d}\bigr) \bar{\mu}_{\bgg}
\notag\\
 \hphantom{D^2_{\mbo{\gamma}_\la \mbo{\gamma}_s} \cdot S_{\text{H}} \bigl[\bgg'', \bgg', \bgg\bigr] :=\int}{}
+ \leftexp{(\bar{\gg})}{\grad}^\mu \bigl(\leftexp{(\bar{\gg})}{\grad}^\nu \bgg'_{ \mu \nu} - \bgg^{\gm \d} \leftexp{(\bar{\gg})}{\grad}_\mu \bgg'_{\gm \d}\bigr) \bar{\mu}_\bgg \notag\\
 \hphantom{D^2_{\mbo{\gamma}_\la \mbo{\gamma}_s} \cdot S_{\text{H}} \bigl[\bgg'', \bgg', \bgg\bigr] :=\int}{}
 + \leftexp{(\bar{\gg})}{\grad}^\mu \bigl(\leftexp{(\bar{\gg})}{\grad}^\nu \bgg'_{\mu \nu} - \bgg^{\gm \d} \leftexp{(\bar{\gg})}{\grad}_\mu \bgg'_{ \gm \d}\bigr) \biggl(\halb \bar{\mu}_\bgg \bgg^{\mu \nu} \bgg'_{\mu \nu}\biggr).\label{EH-second-variation}
\end{gather}
We are interested in a variational principle, so that we get the Euler--Lagrangian field equations for the linearized Einstein equations around the Kerr black hole spacetimes. The functional~\eqref{EH-second-variation} can be simplified as follows:
\begin{align}\label{LG-Lagrangian}
S_{\text{LEE}}\bigl[\bar{\gg}', \bar{\gg}\bigr] &{}= \int \halb \bigl(\leftexp{(\bar{\gg})}{\grad}_\a \bar{\gg}'_{\mu \nu} \leftexp{(\bar{\gg})}{\grad}^\mu \gg'^{\a \nu} + \leftexp{(\bar{\gg})}{\grad}_\a \bar{\gg}'_{\mu \nu} \leftexp{(\bar{\gg})}{\grad}^\nu \bar{\gg}'^{\a \mu} - \leftexp{(\bar{\gg})}{\grad}^\mu \gg'_{\mu \nu} \leftexp{(\bar{\gg})}{\grad}^\nu {\rm tr} \bigl(\gg'\bigr)\\
&\quad{} - \leftexp{(\bar{\gg})}{\grad}_\a \gg'_{\mu \nu} \leftexp{(\bar{\gg})}{\grad}^\a \gg'^{\mu \nu}
- \leftexp{(\bar{\gg})}{\grad}^\a {\rm tr} \bigl(\gg'\bigr) \leftexp{(\bar{\gg})}{\grad}^\b \gg'_{\a \b} +
\leftexp{(\bar{\gg})}{\grad} ^\a {\rm tr} \bigl(\gg'\bigr) \leftexp{(\bar{\gg})}{\grad}_\a {\rm tr}\bigl(\gg'\bigr) \bigr) \bar{\mu}_{\bar{\gg}}. \notag
\end{align}
Suppose $\mathcal{D}$ is the local continuous group of diffeomorphisms generated by the vector field $\bar{Y}$. It~follows that the 2-tensor
\begin{align*}
\bigl(\bar{\gg}'\bigr)_{\a \b} = (\mathcal{L}_{\bar{Y}} \bar{\gg})_{\a \b} = \leftexp{(\gg)}{\grad}_\a \bar{Y}_\b + \leftexp{(\gg)}{\grad}_\b \bar{Y}_\a, \qquad \a, \b = 0, 1, 2, 3,
\end{align*}
is a critical point of the variational principle \eqref{LG-Lagrangian}, corresponding to the pure-gauge perturbations of the Kerr metric.
 It may be noted that the gauge transformations are abelian. In general, the computation of the Euler--Lagrange field equations corresponding to the variation principle~\eqref{LG-Lagrangian} involves the terms listed below.
\begin{itemize}\itemsep=0pt
\item The terms $\leftexp{(\bar{\gg})}{\grad} \bgg''_{\mu \nu} \leftexp{(\bar{\gg})}{\grad}^\mu \bgg'^{\a \nu} +
\leftexp{(\bar{\gg})}{\grad}_\a \bgg'_{\mu \nu} \leftexp{(\bar{\gg})}{\grad}^\mu \bgg''^{\a \nu} $
can be transformed as
\begin{align*}
&\leftexp{(\bar{\gg})}{\grad}_\a \bgg''_{\mu \nu} \leftexp{(\bar{\gg})}{\grad}^\mu \bgg'^{\a \nu} =
\leftexp{(\bar{\gg})}{\grad}_\a \bigl(\bgg'' \leftexp{(\bar{\gg})}{\grad}^\mu \bgg'^{\a \nu}\bigr) - \bgg''_{\mu \nu} \leftexp{(\bar{\gg})}{\grad}_\a \leftexp{(\bar{\gg})}{\grad}^\mu \bgg'^{\a \nu}, \\
&\leftexp{(\bar{\gg})}{\grad}_\a \bgg'_{\mu \nu} \leftexp{(\bar{\gg})}{\grad}^\mu \bgg''^{\a \nu}=
\leftexp{(\bar{\gg})}{\grad}^\mu \bigl(\bgg''^{\a \nu} \leftexp{(\bar{\gg})}{\grad}_\a \bgg'_{\nu \nu}\bigr)-
\leftexp{(\bar{\gg})}{\grad}^\mu \leftexp{(\bar{\gg})}{\grad}_a \bgg'_{\mu \nu},
\end{align*}

\item likewise, the terms $ \leftexp{(\bar{\gg})}{\grad}_\a \bgg''_{\mu \nu} \leftexp{(\bar{\gg})}{\grad}^\nu \bgg'^{\a \mu} + \leftexp{(\bar{\gg})}{\grad}_\a \bgg'_{ \mu \nu} \leftexp8{(\bar{\gg})}{\grad}^\nu \bgg''^{\a \mu}$ can be rewritten as
\begin{align*}
&\leftexp{(\bar{\gg})}{\grad}_\a \bgg''_{\a \nu} \leftexp{(\bar{\gg})}{\grad}^\nu \bgg'^{ \a \nu}=
\leftexp{(\bar{\gg})}{\grad}_a \bigl(\leftexp{(\bar{\gg})}{\grad} \bgg'_{\mu \nu} \leftexp{(\bar{\gg})}{\grad} ^\nu \bgg'^{ \a \mu }\bigr) - \bgg''_{\mu \nu} \leftexp{(\bar{\gg})}{\grad}_\a \leftexp{(\bar{\gg})}{\grad}^\nu \bgg'^{\a \mu} ,\\
&\leftexp{(\bar{\gg})}{\grad}_a \bgg'_{\mu \nu} \leftexp{(\bar{\gg})}{\grad}^\nu \bgg''^{\a \nu}=
\leftexp{(\bar{\gg})}{\grad}^\nu \bigl(\bgg''^{\a \nu} \leftexp{(\bar{\gg})}{\grad}_\a \bgg'_{\mu \nu}\bigr) -
\bgg'' \leftexp{(\bar{\gg})}{\grad}^\nu \leftexp{(\bar{\gg})}{\grad}_\a \bgg'_{\mu \nu},
\end{align*}

\item the terms $ - \leftexp{(\bar{\gg})}{\grad}^\mu \bgg''_{\mu \nu} \leftexp{(\bar{\gg})}{\grad}^\nu {\rm tr} \bgg' - \leftexp{(\bar{\gg})}{\grad}^\mu \bgg'_{\mu \nu} \leftexp{(\bar{\gg})}{\grad}^\nu {\rm tr} \bgg''$, can be transformed conveniently as
\begin{align*}
&- \leftexp{(\bar{\gg})}{\grad}^\mu \bgg''_{\mu \nu} \leftexp{(\bar{\gg})}{\grad}^\nu {\rm tr}\bgg'=
- \leftexp{(\bar{\gg})}{\grad}^\mu \bigl(\bgg''_{\mu \nu} \leftexp{(\bar{\gg})}{\grad}^\nu {\rm tr} \bgg'\bigr)
+ \leftexp{(\bar{\gg})}{\grad}''_{\mu \nu} \leftexp{(\bar{\gg})}{\grad}^\mu \bigl( \leftexp{(\bar{\gg})}{\grad}^\nu {\rm tr} \bgg'\bigr), \\
&- \leftexp{(\bar{\gg})}{\grad}^\mu \bgg'_{\mu \nu} \leftexp{(\bar{\gg})}{\grad}^\nu {\rm tr} \bgg''=
- \leftexp{(\bar{\gg})}{\grad}^\nu \bigl({\rm tr}\bgg'' \leftexp{(\bar{\gg})}{\grad}^\mu \bgg'_{\mu \nu}\bigr) +
{\rm tr}\bgg'' \leftexp{(\bar{\gg})}{\grad}^\nu \leftexp{(\bar{\gg})}{\grad}^\mu \bgg'_{ \mu \nu}.
\end{align*}

\item Finally, the terms $ - \leftexp{(\bar{\gg})}{\grad}_a \bgg''{\mu \nu} \leftexp{(\bar{\gg})}{\grad}^\a \bgg'^{ \mu \nu}$ and $ \leftexp{(\bar{\gg})}{\grad}^\a {\rm tr} \bgg'' \leftexp{(\bar{\gg})}{\grad}_\a {\rm tr} \bgg'$ can be transformed as
\begin{align*}
&-\leftexp{(\bar{\gg})}{\grad}_\a \bgg''_{\mu \nu} \leftexp{(\bar{\gg})}{\grad}^\a \bgg'^{\mu \nu}=
-\leftexp{(\bar{\gg})}{\grad}_a \bigl(\bgg'' \leftexp{(\bar{\gg})}{\grad}^\a \bgg'^{\mu \nu}\bigr) + \bgg''_{\mu \nu} \leftexp{(\bar{\gg})}{\grad}_\a \leftexp{(\bar{\gg})}{\grad}^\a \bgg'^{\mu \nu}
\end{align*}
and
\[
\leftexp{(\bar{\gg})}{\grad}^\a {\rm tr}\bgg'' \leftexp{(\bar{\gg})}{\grad}_\a {\rm tr} \bgg' = \leftexp{(\bar{\gg})}{\grad}^\a \bigl({\rm tr}\bgg'' \leftexp{(\bar{\gg})}{\grad}_\a {\rm tr}\bgg'\bigr) -
{\rm tr} \bgg'' \leftexp{(\bar{\gg})}{\grad}^\a \leftexp{(\bar{\gg})}{\grad}_\a {\rm tr} \bgg'
\]
respectively.
\end{itemize}
Assembling the terms from above, we obtain the Euler--Lagrange equations for the variational principle \eqref{LG-Lagrangian} as
\begin{align*}
&\leftexp{(\bar{\gg})}{\grad}^\gm \leftexp{(\bar{\gg})}{\grad}_\mu \bgg'_{\gm \nu} + \leftexp{(\bar{\gg})}{\grad}^\gm \leftexp{(\bar{\gg})}{\grad}_\nu \bgg_{\gm \mu} - \leftexp{(\bar{\gg})}{\grad}_\mu \leftexp{(\bar{\gg})}{\grad}_\nu {\rm tr}\bigl(\bgg'\bigr) - \leftexp{(\bar{\gg})}{\grad}^\gm \leftexp{(\bar{\gg})}{\grad}_\gm \bgg'_{ \mu \nu} \notag\\
&\qquad{}- \bgg_{\mu \nu}
\leftexp{(\bar{\gg})}{\grad}^\gm \leftexp{(\bar{\gg})}{\grad}^\d \bgg'_{ \gm \d} + \bgg_{\mu \nu}
\leftexp{(\bar{\gg})}{\grad}^\gm \leftexp{(\bar{\gg})}{\grad}_\gm {\rm tr} \bgg' =0, \qquad \bigl(\bar{M}, \bgg\bigr),
\end{align*}
which are the field equations for the linearized Einstein equations around the Kerr background spacetime.
\subsection*{Harmonic coordinates}
Consider a coordinate system $x^\a$
\begin{align*}
\bar{x}^\a \colon\ \bigl(\bar{M}, \bar{\gg}\bigr) \to (M, \bar{\gg}), \qquad \a = 0, 1, 2, 3.
\end{align*}
It follows that the coordinate functions $x^\a$ satisfy the equations
\begin{align} \label{coordinate-wm}
\square_{\bar{\gg}} x^\a + \leftexp{(\bar{\gg})}{\Gamma}_{\b \gamma}^\a \bar{\gg}^{\b \gamma} =0, \qquad \bigl(\bar{M}, \bar{\gg}\bigr), \qquad \a, \b, \gamma = 0,1, 2, 3.
\end{align}
We would like to point out that this equation \eqref{coordinate-wm} is reminiscent of the wave map equations. With the harmonic gauge condition $\square_\gg x^\a =0$, it follows from \eqref{coordinate-wm} that $ \leftexp{(\bar{\gg})}{\Gamma}_{\b \gamma}^\a \bar{\gg}^{\b \gamma} =0 $. We~define an equivalent gauge condition in the perturbative theory using \smash{$D_{\mbo{\gamma}_s} \cdot \leftexp{(\bar{\gg})}{\Gamma}_{\b \gamma}^\a \bar{\gg}^{\b \gamma}$}, i.e.,
\begin{align} \label{harmonic-gauge-condition}
\leftexp{(\bar{\gg})}{\grad}^ \mu \bgg'_{\mu \nu} = \halb \ptl_ \nu {\rm tr} \bgg'
\end{align}
and if we define \smash{$\tilde{\gg}' := \bar{\gg}'- \halb \bar{\gg}^{\a \b} \bar{\gg}'_{\a \b}$}, the `trace-reversed' metric perturbation $\bar{\gg}'$, then this condition can be compactly represented as
\begin{align*} 
\leftexp{(\bgg)}{\grad}_\a \tgg^{\a \b} =0.
\end{align*}
 If we consider the variational principle for the linearized gravity in the Einstein equations for general relativity 
and transform using harmonic coordinates, we get
\begin{align*} 
S_{{\rm HG}}\bigl[\bar{\gg}', \bar{\gg}\bigr] &{}= \int \halb \bigl(\leftexp{(\bar{\gg})}{\grad}_\a \bar{\gg}'_{\mu \nu} \leftexp{(\bar{\gg})}{\grad}^\mu \gg'^{\a \nu} + \leftexp{(\bar{\gg})}{\grad}_\a \bar{\gg}'_{\mu \nu} \leftexp{(\bar{\gg})}{\grad}^\nu \bar{\gg}'^{\a \mu}
 - \leftexp{(\bar{\gg})}{\grad}_\a \gg'_{\mu \nu} \leftexp{(\bar{\gg})}{\grad}^\a \gg'^{\mu \nu}
\bigr) \bar{\mu}_{\bar{\gg}}.
\end{align*}

It may be noted that the structure of the variational functional, is closely related to the deformed wave map action that we considered previously, e.g., in \cite{NG_19_2}. The quantities in the variational principle can be transformed, using the identities
\begin{align*}
&\leftexp{(\bar{\gg})}{\grad}_\gamma \leftexp{(\bar{\gg})}{\grad}_\mu \bgg'_{\a \nu} - \leftexp{(\bar{\gg})}{\grad}_\mu \leftexp{(\bar{\gg})}{\grad}_\gm \bgg'_{\a \nu} = -{\rm Riem}(\bgg)^\sigma_{\,\,\,\,\a \gm \mu} \bgg'_{\sigma \nu} - {\rm Riem}^\sigma_{\,\,\,\, \nu \gm \mu} \bgg'_{\a \sigma}, \notag\\
&\leftexp{(\bar{\gg})}{\grad}_\gamma \leftexp{(\bar{\gg})}{\grad}_\nu \bgg'_{\a \mu} - \leftexp{(\bar{\gg})}{\grad}_\nu \leftexp{(\bar{\gg})}{\grad}_\gm \bgg'_{\a \mu} = -{\rm Riem}(\bgg)^\sigma_{\,\,\,\,\a \gm \nu} \bgg'_{\sigma \mu} - {\rm Riem}^\sigma_{\,\,\,\, \mu \gm \nu} \bgg'_{\a \sigma}
\end{align*}
then the Euler--Lagrange linearized Einstein equations in harmonic coordinates are
\begin{align} \label{harmonic-linearized-einstein}
\leftexp{(\bgg)}{\grad}^\gm \leftexp{(\bar{\gg})}{\grad}_\gm \bgg'_{\mu \nu} + {\rm Riem}(\bgg)^{\sigma \,\,\,\a}_{\,\,\,\mu\,\,\,\nu} \bgg'_{\a \sigma} + {\rm Riem}(\bgg)^{\sigma \,\,\,\a}_{\,\,\,\nu\,\,\,\mu} \bgg'_{\a \sigma} =0, \qquad \bigl(\bar{M}, \bgg\bigr).
\end{align}
A similar equation is satisfied by the trace-reversed metric $\tgg$,
\begin{align} \label{harmonic-linearized-einstein-trace-reversed}
\leftexp{(\bgg)}{\grad}^\gm \leftexp{(\bar{\gg})}{\grad}_\gm \tgg'_{\mu \nu} + {\rm Riem}(\bgg)^{\sigma \,\,\,\a}_{\,\,\,\mu\,\,\,\nu} \tgg'_{\a \sigma} + {\rm Riem}(\tgg)^{\sigma \,\,\,\a}_{\,\,\,\nu\,\,\,\mu} \tgg'_{\a \sigma} =0, \qquad \bigl(\bar{M}, \bgg\bigr).
\end{align}

In the initial value framework, we need to solve the field equations \eqref{harmonic-linearized-einstein} or \eqref{harmonic-linearized-einstein-trace-reversed} together with the gauge determining equations \eqref{harmonic-gauge-condition}.
Now consider the divergence of the linearized Einstein tensor,
\begin{align*}
&\halb \leftexp{(\bgg)}{\grad}^ \mu \bigl( \leftexp{(\bgg)}{\grad}^\gm \leftexp{(\bgg)}{\grad}_\mu \bgg'_{ \gm \nu} + \leftexp{(\bgg)}{\grad}^\gm \leftexp{(\bgg)}{\grad}_\nu \bgg'_{\gm \mu} - \leftexp{(\bgg)}{\grad}_\mu \leftexp{(\bgg)}{\grad}_\nu {\rm tr}\bgg' - \leftexp{(\bgg)}{\grad}^\gm \leftexp{(\bgg)}{\grad}_\gm \bgg'_{ \mu \nu} \notag\\
&\hphantom{\halb \leftexp{(\bgg)}{\grad}^ \mu \bigl(}{}
- \bgg_{\mu \nu} \leftexp{(\bgg)}{\grad}^\gm \leftexp{(\bgg)}{\grad}^\d \bgg'_{ \gm \d} + \bgg_{\mu \nu} \leftexp{(\bgg)}{\grad}^\gm \leftexp{(\bgg)}{\grad}_\gm {\rm tr} \bgg' \bigr).
\end{align*}

If we define the gauge-fixing quantity, \smash{$F_\nu = \leftexp{(\bgg)}{\grad}^\mu \tgg'_{ \mu \nu}$}, we can construct a propagation equation for $F_\nu$ as
\begin{align*}
\leftexp{(\bgg)}{\grad}^\gm \leftexp{(\bgg)}{\grad}_\gm F_\nu &{}\equiv \leftexp{(\bgg)}{\grad}^ \mu
\bigl(\leftexp{(\bgg)}{\grad}^\gm \leftexp{(\bar{\gg})}{\grad}_\gm \bgg'_{\mu \nu} + {\rm Riem}(\bgg)^{\sigma \,\,\,\a}_{\,\,\,\mu\,\,\,\nu} \bgg'_{\a \sigma} + {\rm Riem}(\bgg)^{\sigma \,\,\,\a}_{\,\,\,\nu\,\,\,\mu} \bgg'_{\a \sigma}\bigr) \\
&{}= 0, \qquad \bigl(\bar{M}, \bgg\bigr).
\end{align*}

If we consider the initial value problem
\begin{align*}
&\leftexp{(\bgg)}{\grad}^\gm \leftexp{(\bgg)}{\grad}_\gm F_\nu= 0, \qquad \bigl(\bar{M}, \bgg\bigr), \\
&F_\nu \vert_{\olin{\Sigma}_0} =0, \quad \ptl_{\vec{t}} F_\nu \vert_{\olin{\Sigma}_0} =0, \qquad \bigl(\olin{\Sigma}_0, \bar{q}_0\bigr),
\end{align*}
it is straightforward to show that $F_\nu$ and $\ptl_{\vec{t}} F_\nu \equiv 0$ for all times in the domain of outer communications of the Kerr black hole spacetime. In other words, the gauge condition \eqref{harmonic-gauge-condition} is propagated for all times if it holds on the initial data and the linearized Einstein equations hold in the harmonic gauge. Analogously, consider the propagation of the constraint equations, defined as
\begin{align*}
 H':={}& G' (n, n) \notag\\
={}& n^\mu n^\nu \bigl(\leftexp{(\bgg)}{\grad}^\gm \leftexp{(\bgg)}{\grad}_\mu \bgg'_{ \gm \nu} + \leftexp{(\bgg)}{\grad}^\gm \leftexp{(\bgg)}{\grad}_\nu \bgg'_{\gm \mu} - \leftexp{(\bgg)}{\grad}_\mu \leftexp{(\bgg)}{\grad}_\nu {\rm tr}\bgg' - \leftexp{(\bgg)}{\grad}^\gm \leftexp{(\bgg)}{\grad}_\gm \bgg'_{ \mu \nu} \notag\\
 & - \bgg_{\mu \nu} \leftexp{(\bgg)}{\grad}^\gm \leftexp{(\bgg)}{\grad}^\d \bgg'_{ \gm \d} + \bgg_{\mu \nu} \leftexp{(\bgg)}{\grad}^\gm \leftexp{(\bgg)}{\grad}_\gm {\rm tr} \bgg' \bigr)
 \intertext{the expression for the propagation of the Hamiltonian constraint $H'$ follows from the transformations analogous to the above}
={}& n^\mu n^\nu \biggl(\leftexp{(\bgg)}{\grad}_\mu \biggl(\leftexp{(\bgg)}{\grad}^\gm \bgg_{\gm \nu} - \halb \leftexp{(\bgg)}{\grad}_\nu {\rm tr}\bgg' \biggr) + \leftexp{(\bgg)}{\grad}_\nu \biggl(\leftexp{(\bgg)}{\grad}^\gm \bgg_{\gm \mu} - \halb \leftexp{(\bgg)}{\grad}_\mu {\rm tr}\bgg' \biggr) \biggr) \notag\\
 & - n^\mu n^\nu \bigl(\leftexp{(\bgg)}{\grad}^\gm \leftexp{(\bar{\gg})}{\grad}_\gm \bgg'_{\mu \nu} + {\rm Riem}(\bgg)^{\sigma \,\,\,\a}_{\,\,\,\mu\,\,\,\nu} \bgg'_{\a \sigma} + {\rm Riem}(\bgg)^{\sigma \,\,\,\a}_{\,\,\,\nu\,\,\,\mu} \bgg'_{\a \sigma} \bigr) \notag\\
 & + n^ \nu n_\nu \leftexp{(\bgg)}{\grad}^\gm \bigl(\leftexp{(\bgg)}{\grad}^\gm \leftexp{(\bgg)}{\grad}^\d - \leftexp{(\bgg)}{\grad}^\gm \leftexp{(\bgg)}{\grad}_\gm {\rm tr}\bgg' \bigr)
 \intertext{after imposing the linearized Einstien field equations in the harmonic gauge }
=&{} n^ \mu n^\nu \bigl(\leftexp{(\bgg)}{\grad}_\mu F_\nu + \leftexp{(\bgg)}{\grad}_\nu F_\mu \bigr) + n^\nu n_\nu \leftexp{(\bgg)}{\grad}^\gm F_\gm =0
\end{align*}
for all times, due to the Harmonic gauge propagation. Likewise, for the momentum constraint, after plugging in the relevant formulas
\begin{align*}
 H'_i :={}& G' (n, X) \notag\\
 ={}& n^\mu X^j \bigl(\leftexp{(\bgg)}{\grad}^\gm \leftexp{(\bgg)}{\grad}_\mu \bgg'_{ \gm j} + \leftexp{(\bgg)}{\grad}^\gm \leftexp{(\bgg)}{\grad}_j \bgg'_{\gm \mu} - \leftexp{(\bgg)}{\grad}_\mu \leftexp{(\bgg)}{\grad}_j {\rm tr}\bgg' - \leftexp{(\bgg)}{\grad}^\gm \leftexp{(\bgg)}{\grad}_\gm \bgg'_{ \mu j} \notag\\
 & \hphantom{n^\mu X^j \bigl(}{} - \bgg_{\mu j} \leftexp{(\bgg)}{\grad}^\gm \leftexp{(\bgg)}{\grad}^\d \bgg'_{ \gm \d} + \bgg_{\mu j} \leftexp{(\bgg)}{\grad}^\gm \leftexp{(\bgg)}{\grad}_\gm {\rm tr} \bgg' \bigr) \notag\\
 ={}& n^\mu X^j \biggl(\leftexp{(\bgg)}{\grad}_\mu \biggl(\leftexp{(\bgg)}{\grad}^\gm \bgg_{\gm j} - \halb \leftexp{(\bgg)}{\grad}_\nu {\rm tr}\bgg' \biggr) + \leftexp{(\bgg)}{\grad}_j \biggl(\leftexp{(\bgg)}{\grad}^\gm \bgg_{\gm \mu} - \halb \leftexp{(\bgg)}{\grad}_\mu {\rm tr}\bgg' \biggr) \biggr) \notag\\
 & - n^\mu X^j \bigl(\leftexp{(\bgg)}{\grad}^\gm \leftexp{(\bar{\gg})}{\grad}_\gm \bgg'_{\mu j} + {\rm Riem}(\bgg)^{\sigma \,\,\,\a}_{\,\,\,\mu\,\,\,j} \bgg'_{\a \sigma} + {\rm Riem}(\bgg)^{\sigma \,\,\,\a}_{\,\,\,j\,\,\,\mu} \bgg'_{\a \sigma} \bigr) \notag\\
 & + n^ \mu X^j \bgg_{ \mu j} \leftexp{(\bgg)}{\grad}^\gm \bigl(\leftexp{(\bgg)}{\grad}^\gm \leftexp{(\bgg)}{\grad}^\d - \leftexp{(\bgg)}{\grad}^\gm \leftexp{(\bgg)}{\grad}_\gm {\rm tr}\bgg' \bigr)
 \notag\\
 ={}& n^ \mu X^j \bigl(\leftexp{(\bgg)}{\grad}_\mu F_j + \leftexp{(\bgg)}{\grad}_j F_\mu \bigr) + n^\nu X^j \bgg_{ \nu j} \leftexp{(\bgg)}{\grad}^\gm F_\gm =0, \qquad X^i \in T(\olin{\Sigma}),\quad i= 1, 2, 3.
\end{align*}

It follows that, in the harmonic gauge, the constraints are automatically propagated as long as they are satisfied on the initial data, analogous to the propagation of the harmonic gauge condition. This is an analogous statement for the nonlinear theory developed in the classic work of Choquet-Bruhat.

Let us now discuss the degrees of freedom of linearized gravity. The number of independent
degrees of freedom, modulo the gauge degrees of freedom, is 6. We have shown that the constraint propagation in time follows from the harmonic gauge condition. As a consequence, it follows that these remaining degrees of freedom are all independent and unconstrained.

It may be noted that the system of equations for the linearized Einstein equations is purely a hyperbolic partial differential equation system, from which it follows that the evolution of the data is causal.
\begin{Proposition}
Suppose $\bigl(\olin{\Sigma}', \bar{\mbo{q}'}\bigr)$ is the linearized initial data for the linearized Einstein equations, satisfying the constraint equations, then
\begin{enumerate}\itemsep=0pt
\item[$(1)$] The linearized Einstein equations in harmonic gauge is purely a hyperbolic differential equation system, with the independent degrees of freedom being dynamically unconstrained.
\item[$(2)$] The future $($and past$)$ development of the linearized initial data in harmonic gauge is globally regular and globally hyperbolic.
\end{enumerate}
\end{Proposition}
Let us make a couple of comments about the aforementioned global existence theorem. As we already discussed, the Weyl--Papapetrou gauge offers significant benefits in terms of geometry and topology, that allows us to construct a positive-definite energy in the first place, but presents (gauge-related) causality and regularity issues at the boundaries.

On the other hand, as we already discussed above, the global development of the linearized Einstein equations in the harmonic gauge is untroubled by the (gauge-related) causality and regularity issues at the axes $(\Gamma)$, infinity $\bigl(\bar{\iota}^0\bigr)$ and at the corners $\Gamma \cap \mathcal{H}^+$. In other words, the global development of the linearized Einstein equations from regular initial data is regular $(C^\infty)$ for all times, including at the pathologies like axes, infinity and at the corner of axes and horizon. For example, we have
\begin{align*}
\ptl_{\vec{n}} A =0,
\end{align*}
for a scalar, and
\[
A^{\perp} =0, \qquad \ptl_n A^{\parallel} =0
\]
for a vector $A$, near the axes, globally in time. In our work, our approach is to combine these two gauges and take advantage of benefits offered in each.

Recall that the metric $q$ in the orbit space can be expressed in harmonic gauge as $\mbo{q} = {\rm e}^{2 \mbo{\nu}} \mbo{q}_0, $ where $q_0$ is the flat metric.

The preservation of the flatness condition is
\begin{align} \label{flatness-condition}
D \cdot R'_{\mbo{q}_0} = \bar{\mu}_{\mbo{q}_0} \bigl(\leftexp{(\mbo{q}_0)}{\grad}^a \leftexp{(\mbo{q}_0)}{\grad}^b \bigl(\mbo{q}'_0\bigr)_{ab} - \leftexp{(\mbo{q}_0)}{\grad}^a \leftexp{(\mbo{q}_0)}{\grad}_a \mbo{q}^{cd}_0 \bigl(\mbo{q}'_0\bigr)_{cd} \bigr) =0, \qquad \text{on each $\Sigma$}.
\end{align}
The tensor $\mbo{q}'_0$ can be decomposed as \smash{$(\mbo{q}'_0)_{ab} = \bigl(\mbo{q}'^{{\rm TT}}_0\bigr)_{ab} +
\leftexp{(q_0)}{\grad}_a Y_b + \leftexp{(q_0)}{\grad}_a Y_a + \halb (\mbo{q}_0) {\rm tr} \, \mbo{q}'_0$}. In particular, it may be noted that the pure gauge perturbations \smash{$(\mbo{q}'_0)_{ab} = \leftexp{(q_0)}{\grad}_a Y_b + \leftexp{(q_0)}{\grad}_b Y_a$} satisfy the condition \eqref{flatness-condition}.
 The perturbed metric $(\mbo{q}'_0)$ has the following regularity conditions on the axes $\Gamma$, expressed in $(R, \theta)$ coordinates \cite{O_Rinne_J_Stewart_2005}
\begin{align*}
\ptl_\theta \mbo{q}'_0 (\ptl_R, \ptl_R) =0, \qquad \ptl_\theta \mbo{q}'_0(\ptl_\theta, \ptl_\theta) =0, \qquad \mbo{q}'_0 (\ptl_\theta, \ptl_R) =0,\qquad \text{at the axes $\Gamma$}.
\end{align*}

Our construction can also be used to study stability problem of Kerr black hole spacetimes in harmonic gauge. However, in principle, we can use our energy to study the stability problem in any preferred gauge. If we are given a solution $(M', \bgg')$ in any gauge, we can transform the solution to a harmonic gauge by solving the linearized wave map equations $x^\a \colon \bar{M} \to \bar{M}$, for which it can be proven that they admit smooth solutions for all times. Subsequently, we can transform our the solution to the Weyl--Papapetrou gauge using a $(C^\infty$-)diffeomorphism, by making use of the abelian nature of gauge-transformations.

\section{Canonical phase space variables and Lagrange multipliers\\ in the Weyl--Papapetrou gauge}\label{section7}

Suppose, the group ${\rm SO}(2)$ acts on the $3+1$ Lorentzian spacetime $\bigl(\bar{M}, \bar{g}\bigr)$ such that the orbits of the group ${\rm SO}(2)$ are closed and the group action has a nonempty fixed point set, denoted by~$\Gamma$. These conditions are satisfied by the Kerr metric $\bigl(\bar{M}, \bar{g}\bigr)$
\begin{align*}
&\bar{g}:= - \Sigma^{-1} \bigl(\Delta -a^2 \sin^2 \theta\bigr) {\rm d}t^2 - 4a\Sigma^{-1} \sin^{2} \theta mr {\rm d}t {\rm d}\phi \notag\\
&\hphantom{\bar{g} :=}{}+ \Sigma^{-1} \bigl(\bigl(r^2 +a^2\bigr)^2-\Delta a^2 \sin^2 \theta\bigr) \sin^2 \theta {\rm d}\phi^2 + \Delta^{-1} \Sigma {\rm d}r^2 + \Sigma {\rm d}\theta^2,
\end{align*}
where
\[
\Sigma := r^2 + a^2 \cos^2 \theta \qquad \text{and} \qquad \Delta = r^2-2Mr + a^2,
\]
which can be represented as
\begin{align*}
\bar{g} &{}= \vert \Phi \vert \bigl(\Delta \sin^{-2} \theta {\rm d}t^2 + R^{-2} \sin^2 \theta \bigl(\bigl(r^2 + a^2\bigr)^2- a^2 \Delta \sin^2 \theta\bigr) \bigl({\rm d} \rho^2 + {\rm d}z^2\bigr)\bigr) \\
&\quad{}+ \vert \Phi \vert \bigl({\rm d}\phi^2 - 2 mar \bigl(\bigl(r^2+a^2\bigr)^2 - a^2 \Delta \sin^2 \theta\bigr)^{-1}\bigr)^2,
\end{align*}
where
\begin{gather*}
\rho = R \sin \theta, \qquad z = R \cos \theta, \qquad R := 2 \bigl(r-m + \sqrt{\Delta}\bigr), \ \theta \in [0, \pi],
\\
	\vert \Phi \vert = \frac{\sin^2 \theta }{r^2 +a^2 \cos^2 \theta} \bigl( \bigl(r^2+a^2\bigr)^2 -a^2 \Delta \sin^2 \theta \bigr), \\
	q_{ab} = \sin^2 \theta \frac{ \bigl(r^2+a^2\bigr)^2 -a^2 \Delta \sin^2 \theta }{R^2}.
\end{gather*}

Now consider the conjugate harmonic functions $(\bar{\rho}, \bar{z})$ such that
\begin{align*}
\bar{\rho} := \rho \biggl(1- \frac{\bigl(m^2-a^2\bigr)}{4 \bigl(\rho^2 + z^2\bigr)}\biggr), \qquad \bar{z} := z \biggl(1+ \frac{m^2-a^2}{4 \bigl(\rho^2 + z^2\bigr)}\biggr),
\end{align*}
so that the Jacobian is
\begin{align*}
J := \begin{pmatrix} 1 + \frac{m^2-a^2}{4 (\r^2 + z^2)} \Bigl( \frac{2 \r^2}{ \r^2 +z^2} -1\Bigr) & \r (1 + \frac{(m^2-a^2)z}{2 (\r^2 +z^2)^2} \\
z \Bigl(1- \frac{ (m^2-a^2) \r}{ 2 (\r^2 + z^2)^2}\Bigr) & 1+ \frac{m^2-a^2}{4 (\r^2 +z^2)} \Bigl(1- \frac{2 z^2}{\r^2 +z^2} \Bigr)
\end{pmatrix}.
\end{align*}
In these coordinates, the Kerr black hole horizon $\mathcal{H}^+$ corresponds to a `cut' on the $\{ \bar{\rho} =0\}$ curve and its complement on the $\{\bar{\rho} =0\}$ curve corresponds to the union of two axes, $\Gamma$. This coordinate system and the $(\rho, z)$ coordinate system in the extremal case are the ones originally used by Carter \cite{Car_71}.

Let us briefly recall our construction. In the Weyl--Papapetrou gauge for the Einstein equations, we can reduce the Einstein--Hilbert action
into the reduced Einstein-wave map system:
\begin{align*}
\int \bigl(R_g - h_{AB}(U) g^{\a \b} \ptl_\a U^A \ptl_\b U^B\bigr) \bar{\mu}_g.
\end{align*}

It is straightforward to verify that the Kerr metric is a critical point of the variational functional. In the Hamiltonian version of the dimensional reduction, we also encounter the intermediate phase space $X^{{\rm Max}}$:
\begin{align*}
	X^{{\rm Max}} := \bigl\{\mathcal{A}_i, \mathcal{E}^i \bigr\}
	\end{align*}
	so that the Hamiltonian and momentum constraints for the combined phase space
 \[
 \bigl\{ (q, \mbo{\pi}), (\mathcal{A}, \mathcal{E}), \bigl(\vert \Phi \vert^{1/2}, p\bigr) \bigr\}
 \]
are 	
	\begin{gather*}
	H := \bar{\mu}^{-1}_q \bigl(\Vert \mbo{\pi} \Vert^2_q - {\rm tr}(\mbo{\pi})^2\bigr) + \frac{1}{8} p^2 + \frac{1}{2} \vert \Phi \vert^{-2} \mathcal{E}^a \mathcal{E}_a \\
\hphantom{H :=}{}
	 + \bar{\mu}_q \biggl(-R_q + \halb q^{ab} \ptl_a \log \vert \Phi \vert \ptl_b \log \vert \Phi \vert
	 + \frac{1}{4} q^{ab} q^{bd} \ptl_{[b} \mathcal{A}_{a]} \ptl_{[d} \mathcal{A}_{c]} \biggr), \\
	H_a = -2 \leftexp{(q)}{\grad}_b \mbo{\pi}^b_a + \halb p \ptl_a \log \vert \Phi \vert + \mathcal{E}^b \bigl(\ptl_{[a} \mathcal{A}_{b]}\bigr)
	\end{gather*}
	with the Lagrange multipliers
	\begin{align*}
	\{ N, N^a, A_0\}.
	\end{align*}
	Subsequently, the phase space $X$ was introduced
	\begin{align*}
	X := \{ (q, \mbo{\pi}), (\vert \Phi \vert, p), (\omega, \mbo{r}) \},
	\end{align*}
	which resulted in the Hamiltonian and momentum constraint equations:
	\begin{align*}
	&H= \bar{\mu}^{-1}_q \biggl(\vert \mbo{\pi} \vert^2_q - {\rm tr} (\mbo{\pi})^2 + \halb p_A p^A\biggr) + \bar{\mu}_q \biggl(- R_q + \halb h_{AB} q^{ab} \ptl_a U^A \ptl_b U^B\biggr) ,\\
	&H_a= -2 \leftexp{(q)}{\grad}_b \mbo{\pi}^b_a + p_A \ptl_a U^A,
	\end{align*}
	where the Lagrange multipliers are now	
	\begin{align} \label{lapse-shift-3D}
	\{ N, N^a \}.
	\end{align}
	The fact that the Lagrange multiplier set \eqref{lapse-shift-3D} is now simplified is due to the special topological
	structure of the orbit space $M$ of the Kerr metric. This plays a convenient role in the properties of the adjoint of the dimensionally reduced constraint map.
	The main result of Section~\ref{section4} was to obtain a positive-definite energy functional for the linear perturbative theory of Kerr black hole spacetimes within the assumption of axial symmetry, which allows the aforementioned dimensional reduction. The regularized Hamiltonian energy functional is
\begin{align*}
	H^{{\rm Reg}} := \int_{\Sigma} \mathbf{e}^{{\rm Reg}} {\rm d}^2x,
	\end{align*}
where
	\begin{gather*}
	\mathbf{e}^{{\rm Reg}} := N \bar{\mu}^{-1}_{q_0} {\rm e}^{-2 \nu} \biggl(\Vert \varrho' \Vert_{q_0}^2 + \halb p'_A p'^A\biggr) - \halb N {\rm e}^{2 \nu} \bar{\mu}_{q_0} \mbo{\tau}'^2 \notag\\
	\hphantom{\mathbf{e}^{{\rm Reg}} :=}{}+ \halb N \bar{\mu}_{q_0} q^{ab}_0 h_{AB}(U) \leftexp{(h)}{\grad}_a U'^A \leftexp{(h)}{\grad}_b U'^B \notag\\
	\hphantom{\mathbf{e}^{{\rm Reg}} :=}{}- \halb N \bar{\mu}_{q_0} q_0^{ab} h_{AE} U'^A \leftexp{(h)}{R}^E_{\,\,\,\, BCD} \ptl_a U^B \ptl_b U^C U'^D.
	\end{gather*}
	Let us now formally define the Weyl--Papapatrou gauge.
\begin{Definition}
Suppose $\bigl(\bar{M}, \bar{g}\bigr)$ is a Lorentzian spacetime such that $\bar{M}$ admits the ADM decomposition $\bar{M} = \olin{\Sigma} \times \mathbb{R}$ and the group ${\rm SO}(2)$ acts on $\olin{\Sigma}$ through isometries such that the fixed point set is nonempty and the orbits of its action on $\olin{\Sigma}$ are closed.
Then we define $\bar{g}$ to be in Weyl--Papapetrou form if
\begin{enumerate}\itemsep=0pt
\item $\bar{g}$ admits the decomposition
\begin{align} \label{WP-def}
&\bar{g} = \vert \Phi \vert^{-1} g + \vert \Phi \vert \mathcal{A}^2,
\intertext{where the 1-form $\mathcal{A}$ in $\bar{M}$ is defined as}
&\mathcal{A} = {\rm d}\phi + A_\nu {\rm d}x^\nu, \qquad \text{and} \notag
\end{align}
$\ptl_\phi$ is the (Killing) vector field corresponding to the ${\rm SO}(2)$ symmetry of $\bar{M}$ the scalar \smash{${\vert \Phi \vert = \bar{g}_{\a \b} (\ptl_\phi)^\a (\ptl_\phi)^\b}$} is the spacetime norm of the Killing vector $\ptl_\phi$.
$g$, $A$, $\vert \Phi \vert$ are independent of $\phi$, i.e., $\mathcal{L}_{\ptl_\phi} g = \mathcal{\ptl_\phi} A =0$. It may be noted that $g$ is a metric of Lorentzian signature in the orbit space $M := \bar{M}/{\rm SO}(2)$ such that it further admits the ADM decomposition $M = \Sigma \times \mathbb{R}$, where now $\Sigma = \olin{\Sigma}/{\rm SO}(2)$:
\begin{align*}
g = -N^2 {\rm d}t^2 + q_{ab} \bigl({\rm d}x^a + N^a {\rm d}t\bigr) \otimes \bigl({\rm d}x^b + N^b {\rm d}t\bigr)
\end{align*}
and $q$ is the induced metric of $\Sigma$.
\item There exists a coordinate basis $e^1$ and $e^2$ in $(\Sigma, q)$ such that
\begin{align*} 
q\bigl(e^1, e^1\bigr) - q\bigl(e^2, e^2\bigr) =0,\qquad \text{and} \qquad q\bigl(e^1, e^2\bigr) =0, \qquad \text{on each $\Sigma$}
\end{align*}
and
\begin{align*}
q = {\rm e}^{2 \mbo{\nu}} q_0, \qquad \text{where $q_0$ is the flat 2-metric}.
\end{align*}
\end{enumerate}
\end{Definition}
The metric $\bar{g}$ represented in the above coordinate conditions is referred to as in `Weyl--Papapetrou' form.
We would like to remark that the representation of the metric $\bar{g}$ in terms of a Weyl--Papapetrou form is \emph{not unique}. We would also like to emphasize that on the fixed point set of the ${\rm SO}(2)$ action on $\bar{M}$ we have $\vert \Phi \vert \to 0$ \big($\vert \Phi \vert^{-1}\to\infty$\big) and at the outer asymptotic end of $\bar{M}$ we have $\vert \Phi \vert \to \infty$ \big(and $\vert \Phi \vert^{-1} \to 0$\big). Counterbalancing these effects to obtain well-defined, convergent, gauge-independent quantities, in the context of the initial value problem, is one of the main aspects in our work. We shall apply this construction for the perturbative theory.

We shall pay particular attention to the implications of the perturbed Weyl--Papapetrou gauge on $\Sigma$ and the basis $\{ e_1, e_2\}$. Now then, we have the following conditions on the gauge-transformed perturbed metric $\bar{g}'$ as follows:
\begin{align*}
&\mathcal{L}_Y \bar{g} (e_1, e_1) -\halb q_{0} (e_1, e_1) q^{ab}_0 (\mathcal{L}_{Y} \bar{g} ) _{ab}= - \biggl(\bar{g}' (e_1, e_1) - \halb q_0 (e_1, e_1) q^{ab}_0 \bar{g}'_{ab} \biggr) , \\
&\mathcal{L}_Y \bar{g} (e_1, e_2)= - \bar{g}' (e_1, e_2) , \\
&\mathcal{L}_Y \bar{g} (e_2, e_2) -\halb q_{0} (e_1, e_2) q^{ab}_0 (\mathcal{L}_{Y} \bar{g} ) _{ab}= - \biggl(\bar{g}' (e_1, e_2) - \halb q_0 (e_1, e_1) q^{ab}_0 \bar{g}'_{ab} \biggr),
\end{align*}
which can further be expressed as, using the form of \eqref{WP-def},
\begin{align*}
&\mathcal{L}_Y q_0 (e_1, e_1) -\halb q_{0} (e_1, e_1) q^{ab}_0 (\mathcal{L}_{Y} q_0 ) _{ab}= - {\rm e}^{-2 \mbo{\nu}}\vert \Phi \vert \biggl(\bar{g}' (e_1, e_1) - \halb q_0 (e_1, e_1) q^{ab}_0 \bar{g}'_{ab} \biggr) , \\
&\mathcal{L}_Y q_0 (e_1, e_2)= -{\rm e}^{-2\mbo{\nu}} \vert \Phi \vert \bar{g}' (e_1, e_2) , \\
&\mathcal{L}_Y q_0 (e_2, e_2) -\halb q_{0} (e_1, e_2) q^{ab}_0 (\mathcal{L}_{Y} q_0 ) _{ab}= - {\rm e}^{-2\mbo{\nu}}\vert \Phi \vert \biggl(\bar{g}' (e_1, e_2) - \halb q_0 (e_1, e_1) q^{ab}_0 \bar{g}'_{ab} \biggr).
\end{align*}

 The above system can be expressed compactly in a covariant form as follows
\begin{align*}
(\mathcal{L}_{Y} q_0)_{ab} - \halb (q_0)_{ab}\, (q_0)^{cd} (\mathcal{L}_{Y} q_0)_{cd}=-{\rm e}^{-2 \mbo{\nu}} \vert \Phi \vert \biggl(\bar{g}_{ab} - \halb (q_0)_{ab} (q_0)^{cd} \bar{g}'_{cd} \biggr).
\end{align*}
We have established the following.
\begin{Lemma} \label{Y-conformal-killing}
	Suppose the perturbations of Einstein's equations in axial symmetry are compactly supported in the harmonic gauge then
	\begin{enumerate}\itemsep=0pt
		\item[$(1)$] 	
		\begin{align} \label{2d-gauge-transform}
		\mathcal{L}_Y (\bar{\mu}_q q) ^{ab} = \mathcal{L}_Y (\bar{\mu}_{q_0} q_0) ^{ab}= \bar{\mu}_q q^{ac} q^{bd} \vert \Phi \vert \biggl(\bar{g}'_{cd} - \halb q_{cd} q^{ef} \bar{g}'_{ef}\biggr)
		\end{align}
		the gauge transformation vector field $Y \in T(\Sigma)$ is the projection of the spacetime gauge transformation vector field $\bar{Y}$.
	
	\item[$(2)$] Suppose $\bar{g}'$ is such that it is compactly supported away from the horizon $\mathcal{H^+}$ and the spatial infinity $\iota^0$, then $Y$	is a conformal Killing vector field in $(\Sigma, q)$, i.e., ${\rm CK}(Y, q) = 0$ in the asymptotic regions $($i.e., in the complement of the support of $\bar{g}'$, $\Sigma \setminus \operatorname{Supp}(\bar{g}) \vert_{\Sigma})$.
\end{enumerate}
\end{Lemma}

Now, for estimates, let us choose a gauge for the dimensionally reduced Cauchy hypersurface~$(\Sigma, q)$. If we choose the polar coordinates $(R, \theta)$, it follows from the condition \eqref{2d-gauge-transform} that:
\begin{subequations}\label{Y-theta-transport}
\begin{align}
&\ptl_ \theta Y^\theta= R \ptl_R \frac{Y^R}{R} + \frac{1}{R} \biggl(\bar{\mu}_q q^{aR} q^{bR} \vert \Phi \vert \biggl(\bar{g}'_{ab} - \halb q_{ab} q^{cd} \bar{g}'_{cd}\biggr) \biggr), \\
&\ptl_R Y^\theta= - \frac{1}{R} \ptl_\theta \frac{Y^R}{R} - \frac{1}{R} \biggl(\bar{\mu}_q q^{aR} q^{b\theta} \vert \Phi \vert \biggl(\bar{g}'_{ab} - \halb q_{ab} q^{cd} \bar{g}'_{cd}\biggr)\biggr).
\end{align}
\end{subequations}

It may be noted that above system of differential equations is an overdetermined system. It~follows from the Picard theorem and the Frobenius theorem that the necessary and sufficient conditions for the existence of the solutions is the compatibility condition:
\begin{gather*}
\frac{1}{R} \ptl_R \biggl(R \ptl_R \frac{Y^R}{R} \biggr) + \frac{1}{R^2} \ptl^2 _{\theta} \frac{Y^R}{R}=
- \frac{1}{R} \ptl_R \biggl(\frac{1}{R} \bar{\mu}_q q^{aR} q^{bR} \vert \Phi \vert \biggl(\bar{g}'_{ab} - \halb q_{ab} q^{cd} \bar{g}'_{cd}\biggr)\biggr) \notag\\
\hphantom{\frac{1}{R} \ptl_R \biggl(R \ptl_R \frac{Y^R}{R} \biggr) + \frac{1}{R^2} \ptl^2 _{\theta} \frac{Y^R}{R}=}{}
-\frac{1}{R^2} \ptl_ \theta \biggl(\bar{\mu}_q q^{aR} q^{b \theta} \vert \Phi \vert \biggl(\bar{g}'_{ab} - \halb q_{ab} q^{cd} \bar{g}'_{cd}\biggr)\biggr).
\end{gather*}
In formal terms, this corresponds to vanishing of the commutator of the vector fields corresponding to the differential equations \eqref{Y-theta-transport}.
It may be noted that the compatibility condition fortuitously turns out to be a Poisson equation for \smash{$\frac{Y^R}{R}$}, for the Laplacian $\Delta_0$
\begin{align*}
\Delta_0 = \frac{1}{R} \frac{\ptl}{ \ptl R} \biggl(R \frac{\ptl }{\ptl R} \biggr) + \frac{1}{R^2} \frac{\ptl^2}{ \ptl \theta^2}.
\end{align*}
in the $R$, $\theta$ gauge. Likewise, equations \eqref{Y-theta-transport}, which are equivalent to equation \eqref{2d-gauge-transform} in the Lemma~\ref{Y-conformal-killing}, can be transformed into an overdetermined system of equations for \smash{$\frac{Y^R}{R}$}
\begin{align*}
&\ptl_\theta \frac{Y^R}{R}= - R \ptl_R Y^\theta - \biggl(\bar{\mu}_q q^{aR} q^{b\theta} \vert \Phi \vert \biggl(\bar{g}'_{ab} - \halb q_{ab} q^{cd} \bar{g}'_{cd}\biggr) \biggr), \\
&\ptl_R \frac{Y^R}{R}= \frac{1}{R} \ptl_\theta Y^\theta - \frac{1}{R^2} \biggl(\bar{\mu}_q q^{aR} q^{bR} \vert \Phi \vert \biggl(\bar{g}'_{ab} - \halb q_{ab} q^{cd} \bar{g}'_{cd}\biggr) \biggr)
\end{align*}
for which the compatibility condition is
\begin{align*}
\frac{1}{R} \ptl_R \bigl(R \ptl_R Y^\theta\bigr) + \frac{1}{R^2} \ptl^2_\theta Y^\theta &{}=
\frac{1}{R^3} \ptl_\theta \biggl(\frac{1}{R} \bar{\mu}_q q^{aR} q^{bR} \vert \Phi \vert \biggl(\bar{g}'_{ab} - \halb q_{ab} q^{cd} \bar{g}'_{cd}\biggr)\biggr) \notag\\
&\quad{}- \frac{1}{R} \ptl_R \biggl(\frac{1}{R} \bar{\mu}_q q^{aR} q^{b\theta} \vert \Phi \vert \biggl(\bar{g}'_{ab} - \halb q_{ab} q^{cd} \bar{g}'_{cd}\biggr)\biggr),
 \end{align*}
which is again a Poisson equation for $Y^\theta$. For the reasons of regularity on the axes, we impose the conditions \cite{O_Rinne_J_Stewart_2005}{\samepage
\begin{align*}
Y^ \theta =0, \qquad \ptl_\theta Y^R =0, \qquad \text{on the axes} \ \Gamma.
\end{align*}
In particular, we assume that the behaviour of $Y^\theta \sim \sin \theta$ and $\ptl_\theta Y^R \sim \sin \theta$ close to the axes $\Gamma$.}

Now consider the boundary value problem:
\begin{equation*}
\left.\begin{aligned}
&\Delta_0 \frac{Y^R}{R} =- \frac{1}{R} \ptl_R \biggl(\frac{1}{R} \bar{\mu}_q q^{aR} q^{bR} \vert \Phi \vert \biggl(\bar{g}'_{ab} - \halb q_{ab} q^{cd} \bar{g}'_{cd}\biggr)\biggr) \notag \\
&\hphantom{\Delta_0 \frac{Y^R}{R} =}{} -\frac{1}{R^2} \ptl_ \theta \biggl(\bar{\mu}_q q^{aR} q^{b\theta} \vert \Phi \vert \biggl(\bar{g}'_{ab} - \halb q_{ab} q^{ef} \bar{g}'_{ef}\biggr)\biggr) \qquad \text{on} \ (\Sigma)\\
&\frac{Y^R}{R}= Y^R_{\mathcal{H}^+} \qquad \text{on} \ \bigl(\mathcal{H}^+\bigr)
\end{aligned}
 \right\}
\qquad {\rm (D-BVP)}_{Y^R}
\end{equation*}
We solve the Dirichlet problem ${\rm (D-BVP)}_{Y^R}$ above with the method of images. Let us first consider the Poisson equation
\begin{align*}
\Delta_0 u= F, \qquad \text{in any Lipschitz domain $(\Sigma)$}.
\end{align*}
From elliptic theory, it follows that if $F \in C^{\infty} (\Sigma)$ for regular Dirichlet or Neumann boundary data, then $u \in C^{\infty} (\Sigma)$.
Suppose $K_u$ is the fundamental solution such that
\begin{align*}
\Delta_0 K_u = \delta \bigl(x-x'\bigr)
\end{align*}
where $\delta$ is a Dirac-delta function with a Euclidean metric on $\Sigma$.
Then consider the quantity,
\begin{align*}
\ptl_a (u \ptl^a K_u - K_u \ptl^a u) &{}= (\ptl_a u \ptl^a K_u+ u \Delta K_u) - (\ptl_a K_u \ptl^a u + K_u \Delta u) \notag\\
&{}= u \Delta K_u - K_u \Delta u
\end{align*}
and upon integration over the domain $\Sigma$, we get
\begin{align*}
\int_{\Sigma} \ptl_a (u \ptl^a K_{u} - K_{u} \ptl^a u) &{}= \int_{\ptl \Sigma} n \cdot (u\ptl^a K_{u} - K_{u} \ptl^a u) \notag \\
&{}= \int_{\ptl \Sigma} u \ptl_n K_u - K_u \ptl_n u
= \int_{\Sigma} u\delta \bigl(x-x'\bigr) - K_u F = u - \int_{\Sigma} K_u F.
\end{align*}
Therefore, the general representation formula for $u$ is
\begin{align} \label{gen-Poisson-representation-formula}
u = \int_{\Sigma} K_u F + \int_{\ptl \Sigma} u \ptl_n K_{u} - K_{u} \ptl_n u.
\end{align}
The formula \eqref{gen-Poisson-representation-formula} will be useful for us throughout our work, in different contexts. We can tailor this general formula for both Dirichlet and Neumann boundary value problems. In the following, we shall discuss two configurations that would be particularly relevant for us.

The orbit space $\Sigma$ geometry of Kerr black hole spacetime resembles that of the complement of a half-disk (with the boundary representing the horizon) in a half plane. We solve the Dirichlet problem using the method of images. The regularity and compatibility conditions for our problem imply that the `image' is reflection antisymmetric. This applies for the image charge as well as the Dirichlet data. Thus, with this picture, we have the complement of a full disk in a full plane, with reflection (with respect to the axes) antisymmetric data at the disk. It follows that the asymptotic decay rate for this problem is $\mathcal{O} \bigl(\frac{1}{R}\bigr)$ for this Dirichlet problem. This decay rate can be independently verified using the separation of variables. It may be noted that this decay rate is faster than that of the Poisson equation in a plane \big(i.e., $\frac{1}{2 \pi} \log R$ asymptotic behaviour\big). This faster decay rate plays a fundamental role in our problem.

Likewise, consider the Poisson equation with Neumann boundary conditions in the orbit space~$\Sigma$. The regularity at the axes implies that the Neumann data in the extended picture is reflection anti-symmetric. We thus recover the $\mathcal{O} \bigl(\frac{1}{R}\bigr)$ decay rate for the solution with appropriate decay conditions for the source function $f$, e.g.,
\begin{align*}
&-\frac{1}{R^2} \ptl_ \theta \biggl(\bar{\mu}_q q^{aR} q^{b \theta} \vert \Phi \vert \biggl(\bar{g}'_{ab} - \halb q_{ab} q^{ef} \bar{g}'_{ef}\biggr)\biggr), \\
&- \frac{1}{R} \ptl_R \biggl(\frac{1}{R} \bar{\mu}_q q^{ac} q^{bd} \vert \Phi \vert \biggl(\bar{g}' - \halb q_{ab} q^{ef} \bar{g}'_{ef}\biggr)\biggr)
\end{align*}
is compactly supported (or with appropriate decay rate consistent with asymptotically flat conditions) for our problem.
It follows that the solution of the Dirichlet problem $({\rm D-BVP})_{Y^R}$ decays as
\smash{$\frac{Y^R}{R} \sim \frac{1}{R}$} asymptotically, for large $R$. It follows from analogous arguments that the solutions for the boundary value problem
\begin{equation*}
\left.\begin{aligned}
&\Delta_0 Y^\theta= - \frac{1}{R} \ptl_R \biggl(\frac{1}{R} \bar{\mu}_q q^{Rc} q^{\theta d} \vert \Phi \vert \biggl(\bar{g}' - \halb q_{ab} q^{ef} \bar{g}'_{ef}\biggr)\biggr)
\notag\\
&\hphantom{\Delta_0 Y^\theta=}{}
+ \frac{1}{R^3} \ptl_\theta \biggl(\frac{1}{R} \bar{\mu}_q q^{Rc} q^{Rd} \vert \Phi \vert \biggl(\bar{g}' - \halb q_{ab} q^{ef} \bar{g}'_{ef}\biggr)\biggr) \qquad \text{on} \ (\Sigma)\\
&Y^\theta= Y^\theta_{\mathcal{H}^+} \qquad \text{on} \ \bigl(\mathcal{H}^+\bigr)
\end{aligned}
 \right\}
\tag*{${\rm (D-BVP)}_{Y^\theta}$}
\end{equation*}
are unique, regular (well-posed) and decay \smash{$Y^\theta \sim \frac{1}{R}$} asymptotically for large $R$.
On the other hand, due to the regularity conditions $Y$ admits the expansion:
	\begin{align*}
	&Y^{\theta}= \sum^{\infty}_{n=1} Y^{\theta}_n \sin (n \theta), \qquad
	Y^{R}= \sum^{\infty}_{n=0} Y^R_n \cos(n \theta)
	\end{align*}
for the solutions of the conformal Killing vector $Y$. Likewise, for regularity reasons, the behaviour of inhomogeneities in the boundary value problems mentioned above is restricted on the axes. In particular, written in explicit form, they behave as follows
\begin{align*}
& \bar{\mu}_q q^{RR} q^{\theta \theta} \vert \Phi \vert \bigl(\bar{g}' \bigl(\ptl_R, \ptl_\theta\bigr)\bigr)\sim \sin \theta ,\\
&\ptl_R \biggl(\bar{\mu}_q q^{RR} q^{RR} \vert \Phi \vert \biggl(\bar{g}' (\ptl_R, \ptl_R) - \halb q_{RR} tr_q \bigl(\bar{g}'\bigr)\biggr)\biggr) \sim \sin \theta
\end{align*}
close to the axes $\Gamma$. Now that we clarified the structure of the source terms, let us introduce the notation, for $(R, \theta)$ coordinates
\begin{align*}
&\cal{M}^{ R \theta }:= \bar{\mu}_q q^{RR} q^{\theta \theta} \vert \Phi \vert \bgg'(\ptl_R, \ptl_\theta), \notag\\
&\cal{M}^{R R}:= \bar{\mu}_q q^{R R} q^{RR} \vert \Phi \vert \biggl(\bgg' (\ptl_R, \ptl_R) - \halb q_{RR} {\rm tr}_q \bgg'\biggr).
\end{align*}

 As a consequence of the above arguments, they must admit a Fourier decomposition of the~form:
\begin{align*}
&- \frac{1}{R} \mathcal{M}^{RR}= \sum^\infty_{n = 0} I_n (R, t) \cos n \theta , \notag\\
&-R \mathcal{M}^{R \theta}= \sum^\infty_{n=1} J_n (R, t) \sin n \theta.
\end{align*}

Now, plugging in these decompositions in the first order equations \eqref{Y-theta-transport}, we
get
\begin{align*}
\ptl_R Y^R_0 (R, t) - \frac{1}{R} Y^R_0 (R, t) = I_0 (R, t),
\end{align*}
which admits the solution,
\begin{align*}
Y^R_0 (R, t) = \frac{R}{R_+} Y_0^R (R_+, t) + \int^R_{R_+} \frac{I_0 (R', t)}{R'} {\rm d}R'
\end{align*}
for the lowest frequency quantity $Y^R_0$. Now for higher frequencies,
\begin{align*}
&\ptl_R Y_n^R (R, t) - \frac{1}{R} Y^R_0 (R, t) - n Y^\theta_n (R, t)= I_n (R, t), \notag\\
& R^2 \ptl_R Y^\theta _n (R, t) - n Y^R_n (R, t) = J_n (R, t),
\end{align*}
which again follow from the first-order equations respectively. These equations can be decoupled~as
\begin{align} \label{Y-first-order-modes}
\frac{1}{R} \ptl_R \bigl(R \ptl_R Y^\theta_n (R, t)\bigr) - \frac{n^2 Y^\theta_n(R, t)}{R^2} = \frac{n I_n (R, t)}{R^2} + \frac{1}{R}
\ptl_R \frac{J_n (R, t)}{R}.
\end{align}

The characteristic equation admits two real roots and it may be noted that the fundamental set of solutions is given by
\begin{align*}
 \text{fundamental solutions set for $Y^\theta$ in \eqref{Y-first-order-modes}} = \{ R^n, R^{-n} \}.
\end{align*}
The corresponding Wronskian is \smash{$ \frac{2n}{R} \neq 0 \ \forall R \in (R_+, \infty)$}, $\to 0$ as $R \to \infty$ and \smash{$\to \frac{2n}{R_+}$} for $R \to R_+$. It follows that \smash{$Y^\theta = A(t) R^{-n} + B(t) R^{-n}$}.

We get the following asymptotic behaviour of the Wronskian near the horizon $\mathcal{H}^{+}$:
\begin{align*}
&W_{\mathcal{H}^+} \bigl(Y_n^\theta\bigr)= \frac{2n}{R_+}, \notag\\
&W_{\mathcal{H}^+} \biggl(\frac{Y_n^R}{R}\biggr)= \frac{2n}{R_+} \qquad \text{as} \ R \to R_+,
\end{align*}
and near the outer asymptotic region,
\begin{align*}
&W_{\bar{\iota}^0} \bigl(Y_n^\theta\bigr)= \frac{2n}{R}, \notag\\
&W_{\bar{\iota}^0} \biggl(\frac{Y_n^R}{R}\biggr)= \frac{2n}{R} \to 0 \qquad \text{as} \ R \to \infty,
\end{align*}
so that, for the conformal Killing vector $Y$, the asymptotic behaviour is
\begin{subequations} \label{Y-theta-asym}
\begin{align}
	Y_n^ \theta (R, t)&{}= Y_n^\theta (t) \bigl(R^n - R^{2n}_+ R^{-n}\bigr), \qquad R \ \text{near} \ R_+, \ n \geq 1, \notag \\
	&{} \to 0 \qquad \text{as} \ R \to R_+, \\
	Y_n^\theta (R, t)&= Y_n^\theta (t) R^{-n} \qquad \text{for large} \ R, \ n \geq 1, \notag\\
	&{} \to 0 \qquad \text{as} \ R \to \infty,
\end{align}
\end{subequations}
and
\begin{subequations} \label{Y-R-asym}
	\begin{align}
	Y_n^ R (R, t)&{}= Y_n^R (t) \bigl(R^{n+1} - R^{2n}_+ R^{-n+1}\bigr), \qquad R \ \text{near} \ R_+, \ n \geq 1, \notag \\
	&{} \to 0 \qquad \text{as} \ R \to R_+ ,\\
	Y_n^R (R, t)&= Y_n^R (t) R^{-n+1} \qquad \text{for large} \ R, \ n \geq 1, \notag\\
\frac{Y_n^R}{R}	 &{}\to 0 \qquad \text{as} \ R \to \infty.
	\end{align}
\end{subequations}

For the estimates in this work, the quantities $Y^\theta (t)$, $Y^R (t) $ in the right-hand sides of \eqref{Y-theta-asym} and~\eqref{Y-R-asym} are treated as constants (in each $\Sigma_t$) and thus there is a slight abuse of notation. The behaviour of $Y^R_0 (R, t)$ is a bit subtle and it is directly related to the regularity issues of our problem. This will be studied separately later.

\subsection*{Wave map phase space $\boldsymbol{X}$}
The general gauge transforms of the quantities
look like
\[
\vert \Phi \vert' = \vert \mbo{\Phi} \vert' + \mathcal{L}_{\bar{\text{Y}}}\vert \Phi \vert
\]
subsequently the wave map canonical pairs
\begin{gather*}
U'^A = \mbo{U}'^A + \mathcal{L}_{\bar{\text{Y}}} U^A, \\
p'_A = \mbo{p}'_A + \mathcal{L}_{\bar{\text{Y}}} p_A \qquad \forall A
\end{gather*}
likewise, the (spacetime) gauge transform of the metric on the target
\[
h'_{AB} (U)= \mbo{h}'_{AB} (\mbo{U}) + \mathcal{L}_{\bar{\text{Y}}} h_{AB}(U) \qquad \forall A, B,
\]
which is analogous to the transformation of a scalar.
In the case of the axially symmetric and stationary Kerr black hole spacetime the operator $\mathcal{L}_{\bar{Y}} = \mathcal{L}_{\bar{Y}} \vert_{\Sigma}$. As a consequence, the formulas above reduce to
\begin{align*}
& U'^A= \mbo{U}'^A +\mathcal{L}_{\bar{\text{Y}}} U^A, \\
& p'_A= \mbo{p}'_A + \mathcal{L}_{\bar{\text{Y}}} p_A.
 \end{align*}
In the asymptotic regions we have the above formulas reduce to
\begin{align*}
&U'^A= \mathcal{L}_{\bar{\text{Y}}}\vert_{\Sigma} U^A, \\
&p'_A= \mathcal{L}_{\bar{\text{Y}}}\vert_{\Sigma} p_A.
\end{align*}
After noting that in the $(R, \theta)$ coordinates for $(\Sigma, q)$
\begin{align*}
&\ptl_R \vert \Phi \vert \sim \mathcal{O}(R), \qquad \ptl_\theta \vert \Phi \vert \sim \mathcal{O}\bigl(R^2\bigr) \qquad \text{for large $R$} , \\
& \ptl_R \vert \Phi \vert \sim \mathcal{O}(1), \qquad \ptl_\theta \vert \Phi \vert \sim \mathcal{O}(1) \qquad \text{for $R$ close to $R_+$}.
\end{align*}
As a consequence, we have
\begin{align*}
&\vert \Phi \vert' \sim \mathcal{O} \biggl(\frac{1}{R}\biggr) \qquad \text{for large $R$}, \notag\\
&\vert \Phi \vert' \sim \mathcal{O}(1) \qquad \text{for $R$ close to $R_+$}.
\end{align*}
Next, the other component of the wave map $U \colon (M, g) \to (N, h)$ is constituted by the twist potential. It follows from background Kerr geometry that
\begin{align*}
	&\ptl_R \omega \sim \mathcal{O}\biggl(\frac{1}{R^3}\biggr), \qquad \ptl_\theta \omega \sim \mathcal{O}(1) \qquad \text{for large $R$} , \\
	& \ptl_R \omega\sim \mathcal{O}(1), \qquad \ptl_\theta \omega \sim \mathcal{O}(1) \qquad \text{for $R$ close to $R_+$}
\end{align*}
and
\begin{align*}
&\vert \omega \vert' \sim \mathcal{O} \biggl(\frac{1}{R}\biggr) \qquad \text{for large $R$}, \\
&\vert \omega \vert' \sim \mathcal{O}(1) \qquad \text{for $R$ close to $R_+$}.
\end{align*}
Now, then let us turn to the conjugate momenta, we have from the Hamiltonian equation,
\begin{align*}
\frac{N}{ \bar{\mu}_q}p_A &{}= h_{AB} (U) \ptl_t U'^B - h_{AB}(U) \mathcal{L}_{N'} U^B
\intertext{in the asymptotic regions}
 &{}= h_{AB} (U) \ptl_t \mathcal{L}_Y U^B - h_{AB}(U) \mathcal{L}_{N'} U^B.
\end{align*}

\subsection*{Lagrange multipliers}
 Now let us turn to the remanding quantities that occur in the ADM formalism, Lagrange multipliers $\{ N', N'^a \}$ in our Weyl--Papapetrou gauge, as constructed using a gauge transform from harmonic coordinates. It may noted that, for our background Kerr metric,
 \begin{gather}
\bar{\gg}' (\ptl_\phi, \ptl_\phi) = \vert \Phi \vert',\qquad 	 N= \vert \Phi \vert^\halb (- \bar{g} ({\rm d}t, {\rm d}t))^{-\halb} \nonumber
 	\intertext{noting that}
 	 \vert \Phi \vert '= \bgg' (\ptl_\phi, \ptl_\phi) + Y^\a \ptl_\a \vert \Phi \vert \nonumber
 	\intertext{and}
 \bar{g}' (\ptl_t, \ptl_t)= 	\bgg({\rm d} t, {\rm d} x^\a) \bgg \bigl({\rm d} t, {\rm d} x^\b\bigr) \bigl(\bgg' (\ptl_a, \ptl_\b) + \mathcal{L}_Y \bgg(\ptl_\a, \ptl_\b)\bigr) \nonumber
 	\intertext{we get}
 	 N'= \halb \vert \Phi \vert^{-1} N (\bgg (\ptl_\phi, \ptl_\phi) + Y^\a \ptl_\a \vert \Phi \vert) \notag\\
 \hphantom{N'=}{}
 	 - \halb \vert \Phi \vert^\halb N^3 \bigl(\bgg({\rm d} t, {\rm d} x^\a) \bgg \bigl({\rm d} t, {\rm d} x^\b\bigr) \bigl(\bgg' (\ptl_a, \ptl_\b) + \mathcal{L}_Y \bgg(\ptl_\a, \ptl_\b)\bigr) \bigr). \label{pert-lapse}
 	\end{gather}
In the asymptotic regions, we have
\begin{align*}
	N' = Y^a \ptl_a N + N \ptl_t Y^t.
	\end{align*}
Likewise, for the shift vector, we have
 \begin{align*}
 \bigl(\vert \Phi \vert^{-1} q_{ab}+ \vert \Phi \vert \mathcal{A}_a \mathcal{A}_b\bigr)\bar{N}^b = \vert \Phi \vert q_{ab} N^b + \vert \Phi \vert \mathcal{A}_t \mathcal{A}_a
 \end{align*}
consequently
\[
 \vert \Phi \vert^{-1} \bar{N}'^b = \vert \Phi \vert^{-1} N'^b + q^{ab} \vert \Phi\vert \mathcal{A}_t \mathcal{A}'_a.
\]
Recall
\[
\bar{N}^b = \bar{q}^{ab} \bar{g}_{0a} = \vert \Phi \vert q^{ab} \bar{g}_{0a},
\]
then
\begin{align*}
 N'^b&{}= \bar{N}'^b + q^{ab} \vert \Phi \vert^2 \mathcal{A}_t \mathcal{A}'_a \notag\\
 &{}= \vert \Phi\vert q^{ab} \bar{g}'_{ta} + q^{ab} \vert \Phi \vert \mathcal{A}_t \bar{g}' _{a \phi }\notag\\
 &{}= \vert \Phi\vert q^{ab} \bigl(\bgg'_{ta} + (\mathcal{L}_Y \bar{g})_{ta}\bigr) + q^{ab} \vert \Phi \vert \mathcal{A}_t \bigl(\bgg'_{a \phi} + (\mathcal{L}_Y \bar{g})_{a \phi}\bigr).
\end{align*}
 Now then using
 \begin{gather*}
 	(\mathcal{L}_Y \bar{g})_{0a} =\ptl_0 Y^b \bar{g}_{ba} + \ptl_a Y^t \bigl(-N^2 + N_\phi N^\phi\bigr) + \ptl_a Y^\phi N_\phi \notag\\
 	\hphantom{(\mathcal{L}_Y \bar{g})_{0a} }{}
 = \ptl_0 Y^b \bar{g}_{ba} + \ptl_a Y^t \bigl(- N^2 + \vert \Phi\vert\mathcal{A}^2_0\bigr) + \vert \Phi \vert \ptl_a Y^\phi \mathcal{A}_0,
\\
 	(\mathcal{L}_{Y} \bar{g})_{a \phi} = \vert \Phi \vert \ptl_a Y^\phi + \ptl_a Y^t \bar{N}_\phi \notag\\
 	\hphantom{(\mathcal{L}_{Y} \bar{g})_{a \phi}}{}
 = \vert \Phi\vert \bigl(\ptl_a Y^\phi + \ptl_a Y^t \mathcal{A}_0\bigr)
 	\end{gather*}
 in the asymptotic regions, we have
 \begin{align} \label{pert-shift}
 	N'^b = \ptl_t Y^b - N^2 q^{ab} \ptl_a Y^t.
 \end{align}
The results obtained above are summarized in the following lemma.

 \begin{Lemma}[Lagrange multipliers $\{ N', N'^a \}$]
 	Suppose $\bar{Y}$ is a gauge transform from the harmonic coordinates $\bigl(\bar{M}', \bar{\gg}'\bigr)$ to the Weyl Papapetrou gauge $\bigl(\bar{M}', \bar{g}'\bigr)$.
 	\begin{itemize}\itemsep=0pt
 \item In the Weyl--Papapetrou gauge, the Lagrange multipliers $\left\{ N, N'^a\right\}$ in the asymptotic regions are given by \eqref{pert-lapse} and \eqref{pert-shift} respectively

\item The vector field $N'^a$ so constructed is regular at the axes and behaves as
\[
N'^R = \mathcal{O}(1), \qquad N'^\theta = \mathcal{O}\biggl(\frac{1}{R}\biggr) \qquad \text{near the spatial infinity},
\]
and
\[
N'^R = \mathcal{O}(1), \qquad N'^\theta =0 \qquad \text{at the horizon}.
\]
\end{itemize}
\end{Lemma}
\begin{proof}
In our gauge, we have
\begin{align*}
\ptl_t Y^t = - \frac{1}{N} Y^a\ptl_a N, \qquad \forall\, t
\end{align*}
in the asymptotic regions.
Thus,
\begin{align*}
\ptl_b Y^t = - \ptl_b \biggl(\frac{1}{N} \int^t_{0} Y^a \ptl_a N {\rm d}t' \biggr).
\end{align*}
Now define a quantity $\mathcal{D}^a $ occurring in \eqref{pert-shift} as
\begin{align*}
\mathcal{D}^a := N^2 q^{ab} \ptl_a Y^t,
\end{align*}
then
\begin{align*}
\mathcal{D}^a = - N^2 q^{ab} \ptl_b \biggl(\frac{1}{N} \int^t_{0} Y^a \ptl_a N \,{\rm d}t'\biggr).
\end{align*}

Recall that
\begin{subequations}
\begin{align}
	&{\rm e}^{2 \nu}= \sin^2 \theta \frac{\bigl(r^2 +a^2\bigr)^2 - a^2 \Delta \sin^2 \theta }{R^2}, \\
	&\frac{\ptl_R N}{N}= \frac{1+ \frac{R^2_+}{R^2}}{R \bigl(1- \frac{R^2_+}{R^2} \bigr) } \sim \biggl(1-\frac{R_+}{R}\biggr)^{-1} \qquad \text{as} \ R \to R_+ \notag\\
	&\hphantom{\frac{\ptl_R N}{N}}{}\to 0 \qquad \text{as} \ \mathcal{O} \biggl(\frac{1}{R}\biggr) \qquad \text{as} \ R \to \infty, \label{NR-asym}\\
	&\frac{\ptl_\theta N}{N}= \cot \theta, \label{Ntheta-asym}
		\intertext{and denote}
		&\mathcal{Y}^a = \int^t_0 Y^a \,{\rm d}t .
\end{align}
\end{subequations}
Let us now compute the behaviour of $\mathcal{D}$ at various boundaries. We have the following expressions for the components of $\mathcal{D}$,
\begin{align*}
	D^R &{}= N^2 {\rm e}^{-2 \mbo{\nu}} (q_0)^{RR} \ptl_R Y^t \notag \\
	&{}= \biggl(1 - \frac{R^2_+}{R^2} \biggr)^2 \cdot \frac{R^4}{ \bigl(\bigl(r^2+a^2\bigr)^2 -a^2 \Delta \sin^2 \theta\bigr)} \notag
	 \ptl_R \Biggl(Y^t(t=0) - \mathcal{Y}^R \frac{1+ \frac{R^2_+}{R^2}}{1-\frac{R^2_+}{R^2}} - \mathcal{Y}^\theta \cot \theta \Biggr) \notag\\
	& \to 0 \qquad \text{as} \ R \to \infty
\end{align*}
at the rate of $\mathcal{O}\bigl(\frac{1}{R^2}\bigr)$; and
\[
D^R	= \mathcal{O}(1) \qquad \text{as} \ R \to R_+.
\]
Likewise, we have
\begin{align*}
	D^\theta &{}= N^2 {\rm e}^{-2 \mbo{\nu}} (q_0)^{\theta \theta} \ptl_\theta Y^t \notag\\
	&{}= \biggl(1- \frac{R^2_+}{R^2}\biggr)^2 \cdot \frac{R^4}{ \bigl(\bigl(r^2+a^2\bigr)^2 -a^2 \Delta
		\sin^2 \theta \bigr)} \ptl_\theta \Biggl(Y^t (t=0) - \mathcal{Y}^R \frac{1 + \frac{R^2_+}{R^2}} {{1-\frac{R^2_+}{R^2}} } - \mathcal{Y}^ \theta \cot \theta \Biggr)\\
	&{} \to 0 \qquad \text{as} \ R \to \infty
\end{align*}
at the rate of $\mathcal{O}\bigl(\frac{1}{R^3}\bigr)$
\[
	D^\theta = \mathcal{O} \biggl(1-\frac{R_+}{R}\biggr) \to 0 \qquad \text{as} \ R \to R_+.
\tag*{\qed}
\]
\renewcommand{\qed}{}
\end{proof}

It now follows that the conjugate momenta, are given by
\begin{align*}
N \frac{p'_A}{\bar{\mu}_q} &{}= h_{AB} (U) \ptl_t \mathcal{L}_Y U'^A - h_{AB} (U) \mathcal{L}_{N'} \ptl_a U^b
\intertext{in the asymptotic regions}
&{}= h_{AB}(U) N^2 q^{ab} \ptl_b Y^t \ptl_a U^B \notag\\
&{}= h_{AB} (U) R^2 \sin^2 \theta \biggl(1-\frac{R^2_+}{R^2}\biggr)^2 \cdot \frac{R^2}{ \sin^2 \theta \bigl(\bigl(r^2+a^2\bigr)^2-a^2 \Delta \sin^2 \theta\bigr)} \notag\\
&\quad{} \times \biggl(-\ptl_R U^B \ptl_R \biggl(\mathcal{Y}^a \frac{\ptl_a N}{N}\biggr) - \frac{1}{R^2} \ptl_\theta U^B \ptl_\theta \biggl(\mathcal{Y}^a \frac{\ptl_a N}{N}\biggr)\biggr).
\end{align*}
As a consequence, we can estimate the values of the conjugate momenta at various boundaries (cf.\ equations~\eqref{NR-asym} and \eqref{Ntheta-asym}).
\begin{align*}
&N \frac{p'^A}{ \bar{\mu}_q}= \mathcal{O} \biggl(\frac{1}{R^4}\biggr) \qquad \text{for large $R$}, \\
&N \frac{p'^A}{\bar{\mu}_q}= \mathcal{O} \biggl(1-\frac{R^2_+}{R^2}\biggr)\qquad \text{for $R$ close to $R_+$},\ A= 1, 2.
\end{align*}
It can be verified that the decay rates and the boundary behaviour of the shift vector field agree from separate analysis using the momentum constraint. From the momentum constraint, we~have
\begin{align*}
- \leftexp{(q_0)}{\grad}_b \varrho^b_a + \halb p'_A \ptl_a U^A =0
\end{align*}
after taking into account that the transverse-traceless tensors vanish for our geometry, where
\begin{align*}
\varrho^a_c = \bar{\mu}_{q_0} \bigl(\leftexp{(q_0)}{\grad}_c Y^a + \leftexp{(q_0)}{\grad}^a Y_c - \delta^a_c \leftexp{(q_0)}{\grad}_b Y^b \bigr),
\end{align*}
which can be reduced to an elliptic equation for the shift vector.
Now, let us turn our attention to the Hamiltonian constraint
\begin{align*}
H &{}= \bar{\mu}^{-1}_{q_0} \biggl({\rm e}^{-2 \mbo{\nu}} \Vert \varrho \Vert^2_{q_0} - \halb \mbo{\tau}^2 {\rm e}^{2 \mbo{\nu}} \bar{\mu}^2_{q_0} + \halb p_A p^A\biggr) \notag\\
&\quad{}+ \bar{\mu}_{q_0} \biggl(2 \Delta_0 \mbo{\nu} + \halb h_{AB} q_0^{ab} \ptl_a U^A \ptl_b U^B\biggr), \qquad (\Sigma, q_0),
\end{align*}
where
\[
\Delta_0 \mbo{\nu} := \frac{1}{\bar{\mu}_{q_0}} \ptl_b\bigl(q^{ab}_0 \bar{\mu}_{q_0} \ptl_b \mbo{\nu}\bigr).
\]
There are a few delicate aspects of the Hamiltonian constraint in our dynamical axisymmetric problem. For the special case of the Kerr metric,
\begin{align*}
H = \bar{\mu}_{q_0} \biggl(2 \Delta_0 \mbo{\nu} + \halb h_{AB} q_0^{ab} \ptl_a U^A \ptl_b U^B\biggr), \qquad (\Sigma, q_0).
\end{align*}
If we consider the quantity \smash{$\int H =0$}, it may be noted that the inner boundary term involves the~$\ptl_R \mbo{\nu}$ at the horizon \smash{$\mathcal{H}^+$}. This quantity vanishes for the Kerr black hole metric and we recover the positive mass theorem of Schoen--Yau \cite{schoen-yau-1, schoen-yau-2}. This is consistent with the (inner boundary) horizon being the minimal surface, which is the case with the Kerr black hole spacetime.
Let us now analyze the linearized Hamiltonian constraint
\begin{align*}
H' &{}= \bar{\mu}_{q_0} \bigl(2 \Delta_0 \mbo{\nu}'\bigr) + \halb \bar{\mu}_{q_0} \ptl_{U^C} h_{AB} q_0^{ab} \ptl_a U^A \ptl_b U^B U'^C + \bar{\mu}_{q_0} q_0^{ab} h_{AB} \ptl_a U'^A \ptl_b U^B, \qquad (\Sigma),
\end{align*}
which is an elliptic PDE.
We would like to construct $\mbo{\nu}'$ such that the following proposition holds.

\begin{Proposition} \label{well-posedness-nu'}
	A boundary value problem for $\nu'$ admits a unique, regular and bounded solution that decays at the rate of $\mathcal{O}\bigl(\frac{1}{R}\bigr)$ for large $R$.
\end{Proposition}

In order to set up a well-posed boundary value problem for this elliptic PDE, we need to specify appropriate boundary conditions.

Suppose, $\Sigma$ is a 2-surface that embeds into the Cauchy hypersurface of the Kerr metric \smash{$(\Sigma \hookrightarrow) \olin{\Sigma}$}, the Gauss curvature of $\Sigma$ is given by
\begin{align*}
K (\Sigma) = - \frac{1}{{\rm e}^{2 \mbo{\nu}}} \Delta_0 \mbo{\nu}
\end{align*}
and the mean curvature $H$ of $\Sigma \hookrightarrow \olin{\Sigma}$
\[
H_\Sigma = 0
\]
on account of the fact that the inner boundary of $\Sigma$ is a minimal surface. In the perturbative theory, we need to find an expression of the mean curvature.
\begin{align*}
\vec{H}_{\Sigma} = H_{\Sigma} \mbo{n}, \qquad \text{mean curvature vector, $\Sigma \hookrightarrow \olin{\Sigma}$}.
\end{align*}
For the conformal transformation $q= \Omega^2 q_0$, we have
\begin{align*}
H_\Sigma = \frac{1}{\Omega} \biggl(H_0 + \frac{2}{\Omega} \frac{ \ptl \Omega}{\ptl \mbo{n}} \biggr).
\end{align*}
It then follows that
\begin{align*}
H_{\Sigma} = \vert \Phi \vert^{1/2} {\rm e}^{-\mbo{\nu}} \biggl(\frac{1}{R} + \ptl_R \mbo{\nu}
\biggr).
\end{align*}
It may be noted that the preservation of the minimal surface condition at the inner boundary horizon \smash{$\mathcal{H}^+$} implies a Neumann boundary condition for $\mbo{\nu}'$ at the horizon \smash{$\mathcal{H}^+$}.
Let us now turn to the issue of regularity of $\mbo{\nu}'$ at the axes $\Gamma$. Firstly note that, from the form of the Kerr metric, the quantity $\mbo{\nu}'$ has the form
\begin{align*}
\mbo{\nu} = 2\gamma + c - 2 \log \rho
\end{align*}
in the Weyl--Papapetrou gauge, near the axes. For the reasons of regularity of the Kerr metric at the axes, we have $c \equiv 0$. Thus,
\begin{align*}
\mbo{\nu} = 2 \gamma - 2 \log \rho.
\end{align*}
Consequently, we have the condition
\begin{align} \label{gamma-nu condition}
\mbo{\nu}' = 2 \gamma'
\end{align}
for the linear perturbation theory, thus suggesting the Dirichlet boundary conditions for $\mbo{\nu}'$ at the axes
\begin{equation*}
\left.\begin{aligned}
&\Delta_0 \mbo{\nu}'=\frac{1}{4} \ptl_{U^C} h_{AB}(U) q^{ab}_{0} \ptl_a U^A U'^C + q_{ab} h_{AB} \ptl_a U^a \ptl_b U^B \qquad \text{on}\ (\Sigma)\\
&\mbo{\nu}' = 2\gamma' \qquad \text{on} \ (\Gamma) \\
&\ptl_n \mbo{\nu}'=0 \qquad \text{on} \ \bigl(\mathcal{H}^{+}\bigr)
\end{aligned}
 \right\}
\qquad {\rm (M-BVP)}_{\mbo{\nu'}}.
\end{equation*}

It must be pointed out that for the regularity of $\mbo{\nu'}$ itself, at the axes, we need to impose the condition $\ptl_{\mbo{n}} \mbo{\nu}'=0$ at the axes. A priori, the apparent over-determined boundary data --~both Dirichlet and Neumann~-- at the axes is a significant issue for the well-posedness of the boundary value problem for $\mbo{\nu}'$. It may be recalled that an elliptic problem with both Dirichlet and Neumann data, specified at the same boundary, is typically ill posed. Importantly, we can prove that the boundary conditions $\mbo{\nu}' = 2 \gamma'$ and $\ptl_{\mbo{n}} \mbo{\nu}'=0$ are equivalent. In other words, imposition of one condition automatically satisfies the other and vice versa. This saves us from the aforementioned issue.

However, we should ask a more fundamental question: Can the set-up of our problem result in a well-posed boundary value problem for $\mbo{\nu}'$ (as proposed in Proposition~\ref{well-posedness-nu'}) in the first place? and what about the regularity of $\mbo{\nu}'$ at the corners $\Gamma \cap \mathcal{H}^+$?
Suppose, $\mbo{\nu}' = 2 \gamma'$, then a~computation shows that
\begin{align*}
\ptl_{\mbo{n}} \mbo{\nu}' = \ptl_R \mbo{\nu}'= 2 \ptl_R \gamma' =0
\end{align*}
at the corners, in a limiting sense. On the other hand, suppose $\mbo{\nu}'$ were a scalar, then $\mbo{\nu}' = (\mbo{\nu}')_{{\rm HG}} + Y^a \ptl_a \nu = Y^a \ptl_a \mbo{\nu}$ in the asymptotic regions. Then a computation shows that
\begin{align*}
\ptl_R \mbo{\nu}' = \ptl_R Y^a \ptl_a \mbo{\nu} \neq 0.
\end{align*}
In the above, there is an inconsistency for two reasons. Firstly, the quantity $\mbo{\nu}'$ is not regular at the corners. Secondly, the Neumann boundary condition at the horizon is not satisfied at the end points, the corners. The following lemma is crucial and comes to our rescue. In this lemma, we show that $\mbo{\nu}'$ does not transform like a scalar.

\begin{Lemma}\samepage
	Suppose the quantity $\mbo{\nu}'$ is compactly supported in the harmonic gauge, then
	\begin{enumerate}\itemsep=0pt
	 \item[$(1)$] $\mbo{\nu}'$ has the following structure
	\begin{align}
	\mbo{\nu}' = Y^a \ptl_a \mbo{\nu} + \halb \leftexp{(q_0)}{\grad}_a Y^a \label{gauge-transform-nu'}
	\end{align}
	in the asymptotic regions.
	\item[$(2)$] Furthermore, $\mbo{\nu}'$ is regular at the corners $\Gamma \cap \mathcal{H}^+$.
	\end{enumerate}
\end{Lemma}
\begin{proof}
	Consider $D \cdot \bar{\mu}_q$ in a conformally flat form, we have
	\begin{align*}
	D \cdot \bar{\mu}_q = {\rm e}^{2 \mbo{\nu}} D \cdot \bar{\mu}_{q_0} + 2 {\rm e}^{2 \mbo{\nu}} \mbo{\nu}' \bar{\mu}_{q_0}.
	\end{align*}
	Now, then for our choice of the gauge condition, where we hold the flat metric $q_0$ fixed,
	\begin{align*}
	2 {\rm e}^{2 \mbo{\nu}} \mbo{\nu}' \bar{\mu}_{q_0} &{}= (\mbo{\nu}')_{{\rm HG}} + \mathcal{L}_Y \bar{\mu}_q, \notag\\
\intertext{in the asymptotic regions}
	&{}= 2 {\rm e}^{2 \mbo{\nu}} Y^a \ptl_a \mbo{\nu} \bar{\mu}_{q_0} + {\rm e}^{2 \mbo{\nu}} \ptl_a \bar{\mu}_{q_0} + {\rm e}^{2 \mbo{\nu}} \ptl_a Y^a \bar{\mu}_{q_0},
	\intertext{the right-hand side can be expressed equivalently as}
	&{}= 2 {\rm e}^{2 \mbo{\nu}} Y^a \ptl_a \mbo{\nu} \bar{\mu}_{q_0} + {\rm e}^{2 \mbo{\nu}} \mathcal{L}_Y \bar{\mu}_{q_0} \\
	&{}= 2 {\rm e}^{2 \mbo{\nu}} Y^a \ptl_a \mbo{\nu} \bar{\mu}_{q_0} + {\rm e}^{2 \mbo{\nu}} \ptl_a (Y^a \bar{\mu}_{q_0}) \\
	&{}= 2 {\rm e}^{2 \mbo{\nu}} Y^a \ptl_a \mbo{\nu} \bar{\mu}_{q_0} + {\rm e}^{2 \mbo{\nu}} \bar{\mu}_{q_0} \leftexp{(q_0)}{\grad}_a Y^a.
	\end{align*}
	The result \eqref{gauge-transform-nu'} follows.
	Now let us show that the `correction term' is harmonic in the asymptotic regions. We have
	\begin{align*}
	\leftexp{(q_0)}{\grad}_a \leftexp{(q_0)}{\grad}_b Y_c - \leftexp{(q_0)}{\grad}_b \leftexp{(q_0)}{\grad}_a Y_c = \leftexp{(q_0)} R^{\quad d}_{abc} Y_d
	\end{align*}
	on account of the fact that $Y$ is a conformal Killing vector field in the asymptotic regions, ${\rm CK}(Y, q_0)=0$, it follows that
	\begin{align*}
	\leftexp{(q_0)}{\grad}^a \leftexp{(q_0)}{\grad}_a \bigl(\leftexp{(q_0)}{\grad}_c Y^c\bigr) =0.
	\end{align*}
	
	With the correction term in \eqref{gauge-transform-nu'}, the assertion (2) can be also be verified explicitly. It can also be verified that the Neumann boundary condition for $\mbo{\nu}'$ at the horizon $\mathcal{H}^+$ is satisfied.
\end{proof}

Firstly, we note that the mixed boundary value problem, with regular boundary data, is well-posed. We would like to remark that `uniqueness up to a constant' which is usually the case for a Neumann boundary value problem is not sufficient for our problem because we need the decay of $\mbo{\nu}'$ (consistent with the regularity on the axes) to establish the needed boundary behaviour of our fields.
In view of the uniqueness, we claim that the solution to the mixed boundary value problem above can be decomposed into the following two parts: solution for a inhomogeneous Dirichlet boundary value problem $\mbo{\nu'}_D$ and a homogeneous Neumann problem $\mbo{\nu}'_N$
\begin{align*}
\mbo{\nu}' = \mbo{\nu}'_D + \mbo{\nu}'_N
\end{align*}
represented in the same domain. In the following, we shall elaborate on our construction that ensures that the needed boundary conditions for $\mbo{\nu}'$ are satisfied. We specify the (regular) data for Dirichlet problem for $\mbo{\nu}'_D$ in a half plane such that it is consistent with the data for $\mbo{\nu}'$. This problem is well-posed on account of the regularity conditions of the source (involves $U'$, $h(U)$, $h'(U)$) and admits a regular solution. We shall use this to obtain a suitable choice of data for the Neumann problem for $\mbo{\nu}'_N$ at the horizon, so that the total sum $\mbo{\nu}'_D + \mbo{\nu}'_N$ satisfies the needed Neumann boundary condition for $\mbo{\nu}'$. The well-posedness and regularity for $\mbo{\nu}'$ at the axes implies that it should vanish on the axes, thus implying that $\mbo{\nu}'$ satisfies the needed Dirichlet condition on the axes. This construction also ensures that the regularity of $\mbo{\nu}'$ at the corners is also satisfied. Consider
\begin{equation} \label{Dirichlet-nu'}
\left.\begin{aligned}
&\Delta_0 \mbo{\nu}_D'=\frac{1}{4} \ptl_{U^C} h_{AB}(U) q^{ab}_{0} \ptl_a U^A U'^C\\
& \hphantom{\Delta_0 \mbo{\nu}_D'=}{} + q_{ab} h_{AB} \ptl_a U^a \ptl_b U^B \qquad \text{on}\ (\Sigma)\\
&\mbo{\nu}'=2 \gamma' \qquad \text{on} \ (\Gamma) \\
&\mbo{\nu}'=2 \gamma' \qquad \text{on} \ (\{ \rho=0\} \setminus \Gamma)
\end{aligned}
 \right\}
\qquad {\rm (D-BVP)}_{\mbo{\nu'}}.
\end{equation}
It may be noted that, for our construction, the data at the `cut' between the axes~${\{ \rho =0\} \setminus \Gamma}$ can be specified arbitrarily so that its smooth but for convenience we choose $ \mbo{\nu}'_D = 2 \gamma'$. It~follows that well posed Dirichlet boundary value problem \eqref{Dirichlet-nu'} is well posed and that there exists a regular solution. We shall use this to construct Neumann data for $\mbo{\nu}'_N$. Note that we can compute $\ptl_{\mbo{n}} \mbo{\nu}'_D$ at the horizon (boundary) using the `Dirichlet to Neumann map'
\begin{align*}
\Lambda \colon\ \mbo{\nu'} \to \ptl_{\mbo{n}} \mbo{\nu}'
\end{align*}
 and the conformal inversion map discussed previously. We then formulate the homogeneous Neumann boundary value problem
\begin{equation*}
\left.\begin{aligned}
&\Delta_0 \mbo{\nu}_N'=0 \qquad \text{on} \ (\Sigma)\\
&\ptl_n \mbo{\nu}'_N=0 \qquad \text{on} \ (\Gamma) \\
&\ptl_n \mbo{\nu}'_N= \ptl_{\mbo{n}} \mbo{\nu}'_D \big\vert_{\mathcal{H}^+} \qquad \text{on} \ \bigl(\mathcal{H}^{+}\bigr)
\end{aligned}
 \right\}
\qquad \text{(N-BVP)}_{\mbo{\nu'}}.
\end{equation*}
It follows from the use of the representation formulas \eqref{gen-Poisson-representation-formula} for both Dirichlet and Neumann problems, that $\mbo{\nu}'$ decays at the rate of $\mathcal{O}\bigl(\frac{1}{R}\bigr)$ for large $R$.
\begin{Proposition}
	The integral invariant quantity $Y_0(\mathcal{H}^+)$ at the horizon is finite for all times if and only if it vanishes.
\end{Proposition}

\begin{proof}
	Recall the expression for $Y_0^R (R_+)$
	\begin{align}
	Y^R_0 (R_+) &{}= \frac{R_+}{ 2 \pi} \int^\infty_{R_+} \frac{1}{R'^2} \int^{2 \pi}_{0} \frac{R'}{2} {\rm e}^{2 \gamma - 2 \mbo{\nu}} \biggl(\bgg' (\ptl_R, \ptl_R) - \frac{1}{R'^2} \bgg' (\ptl_\theta, \ptl_\theta) \biggr) {\rm d} \theta {\rm d}R', \nonumber
	\intertext{which can be reexpresed as}
	&{}= \frac{R_+}{ 2 \pi} \int^\infty_{R_+} \frac{1}{R'^2} \int^{2 \pi}_{0} \frac{R'}{2} \biggl(\mbo{q}'_0 (\ptl_R, \ptl_R) - \frac{1}{R'^2} \mbo{q}'_0 (\ptl_\theta, \ptl_\theta) \biggr) {\rm d} \theta {\rm d}R', \label{integral-invariant-expression}
		\end{align}
		where $q_0$ is the metric perturbation in harmonic gauge.

The quantities, $\mbo{q}'_0 (\ptl_R, \ptl_R)$, $\mbo{q}'_0 (\ptl_R, \ptl_\theta)$ and $ \mbo{q}'_0 (\ptl_\theta, \ptl_\theta)$ admit the decomposition
\begin{align*}
&\mbo{q}'_0 (\ptl_R, \ptl_R)= \sum^{\infty}_{n=0} \bigl\{ \mbo{q}'_0 (\ptl_R, \ptl_R) \bigr\}_n \cos n \theta, \\
&\mbo{q}'_0 (\ptl_R, \ptl_\theta)= \sum^{\infty}_{n=1} \bigl\{ \mbo{q}'_0 (\ptl_R, \ptl_\theta) \bigr\}_n \sin n \theta, \\
&\mbo{q}'_0 (\ptl_\theta, \ptl_\theta)= \sum^{\infty}_{n=0} \bigl\{ \mbo{q}'_0 (\ptl_\theta, \ptl_\theta) \bigr\}_n \cos n \theta.
\end{align*}
Now then \eqref{integral-invariant-expression} can be simplified as
\begin{align}
Y^R_0 (R_+) &{}= \frac{R_+}{2} \int^\infty_{R_+} \frac{1}{R'} \biggl( \bigl\{ \mbo{q}'_0 (\ptl_R, \ptl_R) \bigr\}_0 (R') - \frac{1}{R'^2} \bigl\{ \mbo{q}'_0 (\ptl_\theta, \ptl_\theta) \bigr\}_0 (R') \biggr) {\rm d}R'
\notag\\
&{}= \frac{R_+}{2} \int^\infty_{R_+} \frac{1}{R'} \biggl( \bigl\{\mbo{q}'_0 (\ptl_R, \ptl_R) \bigr\}_0 (R') - \ptl_{R'} \biggl(\frac{1}{R'} \bigl\{ \mbo{q}'_0 (\ptl_\theta, \ptl_\theta) \bigr\}_0 (R') \biggr) \biggr)
{\rm d}R' \notag\\
&\quad{}+\frac{R_+}{2} \int^\infty_{R_+} \ptl_{R'} \biggl(\frac{1}{R'^2} \bigl\{ \mbo{q}'_0 (\ptl_\theta, \ptl_\theta) \bigr\}_0 (R') \biggr) {\rm d}R'. \label{integral-invariant-II}
\end{align}
Now consider \eqref{flatness-condition} and note that the operator on the left-hand side is a linear differential operator. In~polar coordinates
\begin{align*}
& -\frac{1}{R} \ptl^2_\theta \mbo{q}_0 (\ptl_R, \ptl_ R) + \ptl_R \mbo{q}_0 (\ptl_R, \ptl_R) + \frac{2}{R} \ptl^2 R \theta \mbo{q}_0 (\ptl_R, \ptl_ \theta) - \frac{2}{R^3} \mbo{q}_0 (\ptl_\theta, \ptl_\theta) \notag\\
&\qquad{} + \frac{2}{R^2} \ptl_R \mbo{q}'_0 (\ptl_\theta, \ptl_\theta) - \frac{1}{R} \ptl^2_{RR} \mbo{q}_0 (\ptl_\theta, \ptl_\theta) =0.
\end{align*}

It follows that
\begin{align*}
&\ptl_R \bigl\{ \mbo{q}'_0 (\ptl_R, \ptl_R) \bigr\}_0 - \frac{2}{R^3} \bigl\{ \mbo{q}'_0 (\ptl_\theta, \ptl_\theta) \bigr\}_0 \notag\\
& \qquad{} + \frac{2}{R^2} \ptl_R \bigl\{ \mbo{q}'_0 (\ptl_\theta, \ptl_\theta) \bigr\}_0 -
\frac{1}{R} \ptl^2_R \bigl\{ \mbo{q}'_0 (\ptl_\theta, \ptl_\theta) \bigr\}_0 =0,
\end{align*}
which is a pure divergence,
\begin{align*}
\ptl_R \biggl(\bigl\{ \mbo{q}'_0 (\ptl_R, \ptl_R) \bigr\}_0 - \ptl_R \biggl( \frac{1}{R} \bigl\{ \mbo{q}'_0 (\ptl_\theta, \ptl_\theta) \bigr\}_0 \biggr) \biggr) = 0,
\end{align*}
uniformly in $\Sigma$ for all times. We get the condition that the $R$ derivative of the quantity in the first integral of \eqref{integral-invariant-II} vanishes. As a result this quantity can make a finite contribution only if it vanishes. The result follows.
\end{proof}

\section{Strict conservation of the regularized Hamiltonian \texorpdfstring{$\boldsymbol{H^{{\rm Reg}}}$}{H\string^\{Reg\}}}\label{section8}

In this section, we shall establish that the regularized Hamiltonian energy $H^{{\rm Reg}}$ is strictly conserved in time. In particular, we shall establish the following.

\begin{Theorem}
	Suppose we have the initial value problem of Einstein's equations for general relativity. Then
	\begin{enumerate}\itemsep=0pt
		\item[$(1)$] There exists a $(C^\infty$-$)$diffeomorphism from harmonic coordinates to the Weyl--Papapetrou gauge of the maximal development $(M', g')$ of the perturbative theory of Kerr black hole spacetimes, under axial symmetry.
		
		\item[$(2)$] The positive-definite Hamiltonian $H^{{\rm Reg}}$ is strictly-conserved forwards and backwards in time.
	\end{enumerate}
\end{Theorem}

 Consider the vector field density
 \begin{align*}
 	(J^b)^{{\rm Reg}}&{}= N^2 {\rm e}^{-2 \mbo{\nu}} q^{ab}_0 p'_A \ptl_a U'^A + U'^A \mathcal{L}_{N'} \bigl(N \bar{\mu}_{q_0} q^{ab}_0 h_{AB} \ptl_b U^B\bigr) \notag\\
 	&\quad{}+ \mathcal{L}_{N'} (N) \bigl(2 \bar{\mu}_{q_0} q^{ab}_0 \ptl_a \mbo{\nu}'\bigr) + 2 N'^a \ptl_a \mbo{\nu}' \bar{\mu}_{q_0} q^{ab}_0 \ptl_a N
 - 2 N'^b \bar{\mu}_{q_0} q^{bc}_0 \ptl_a \mbo{\nu}' \ptl_c N.
 	\end{align*}
 We shall classify the terms in \smash{$\bigl(J^b\bigr)^{{\rm Reg}}$} into kinematic, dynamical and conformal terms:
 \begin{alignat*}{3}
 	&J_1^b := N^2 {\rm e}^{-2 \mbo{\nu}} q^{ab}_0 p'_A \ptl_a U'^A, \qquad && \text{(dynamical terms)}& \\
 	&J_2^b := U'^A \mathcal{L}_{N'} \bigl(N \bar{\mu}_{q_0} q^{ab}_0 h_{AB} \ptl_b U^B\bigr), \qquad && \text{(kinematic terms)}& \\
 	&J_3^b := \mathcal{L}_{N'} (N) \bigl(2 \bar{\mu}_{q_0} q^{ab}_0 \ptl_a \mbo{\nu}'\bigr), &&& \\
 	&J_4^b := 2 N'^a \ptl_a \mbo{\nu}' \bar{\mu}_{q_0} q^{ab}_0 \ptl_a N, &&&\\
 	&J_5^b := - 2 N'^b \bar{\mu}_{q_0} q^{bc}_0 \ptl_a \mbo{\nu}' \ptl_c N. \qquad && \text{(conformal terms)}
 	\end{alignat*}

The asymptotic and decay rates of the fluxes and the associated integrands can be explicitly evaluated with a specific choice of gauge on the target. In the following, we choose, $h = 4 {\rm d}\gamma ^2 + {\rm e}^{-4 \gamma } {\rm d} \omega^2$. Analogous computations can be performed for other gauges on the target manifold. An important aspect of these flux estimates is that most of the integrands of these flux terms vanish pointwise at the boundaries of the orbit space. 

\subsection*{`Dynamical' boundary terms}
Consider the dynamical flux term, ${\rm Flux} (J_1, \Gamma)$, where
\[
(J_1)^b := N^2 {\rm e}^{-2 \mbo{\nu}} q^{ab}_0 p'_A \ptl_a U'^A.
\]
Let us start by considering the terms
\[
{\rm Flux} \bigl(N^2 {\rm e}^{-2 \mbo{\nu}} (q_0)^{ab} \bigl(4 \mbo{p} \ptl_a \gamma' \bigr), \Gamma\bigr)
\]
 and
 \[
 {\rm Flux} \bigl(N^2 {\rm e}^{-2 \mbo{\nu}} (q_0)^{ab} \bigl({\rm e}^{-4 \gamma} \mbo{r} \ptl_a \gamma' \bigr), \Gamma\bigr).
 \]
We have
\begin{gather*}
{\rm Flux} \bigl(N^2 {\rm e}^{-2 \mbo{\nu}} (q_0)^{ab} \bigl(4p \ptl_a \gamma' \bigr), \Gamma\bigr) \\
\qquad{} = \int^t_0 \int_{ (-\infty, -R_+) \, \cup\, (R_+, \infty) } \lim_{\theta \to 0, \pi} \bigl(4 N^2 {\rm e}^{-2 \mbo{\nu}} (q_0)^{ab} \mbo{p} \ptl_a \gamma'\bigr)^\theta {\rm d}R {\rm d}t \notag\\
\qquad{}= \int^t_0 \int_{ (-\infty, -R_+) \, \cup\, (R_+, \infty) } \lim_{\theta \to 0, \pi} \biggl(4 N \biggl(\frac{ N\mbo{p}'}{\bar{\mu}_q} \biggr) \bar{\mu}_{q_0} (q_0)^{ab} \ptl_a \gamma' \biggr)^\theta {\rm d}R {\rm d}t. \notag
 \end{gather*}
If we consider the integrand, within the domain of integration, we have
 \begin{gather*}
 \lim_{\theta \to 0, \pi}4 \frac{N}{R^2} \cdot \frac{N \mbo{p}'}{\bar{\mu}_q} \bar{\mu}_{q_0} \ptl_\theta \gamma' = 4 \frac{\Delta^\halb \sin \theta}{R^2} \frac{N \mbo{p}'}{\bar{\mu}_q} \bar{\mu}_{q_0} \ptl_\theta \gamma'
 \to 0 \qquad \text{as} \ \theta \to 0, \pi,
\end{gather*}
and
\begin{gather*}
 {\rm Flux} \bigl(N^2 {\rm e}^{-2 \mbo{\nu}} (q_0)^{ab} \bigl({\rm e}^{-4 \gamma} \mbo{r} \ptl_a \omega' \bigr), \Gamma\bigr) \\
 \qquad{}= \int^t_0 \int_{ (-\infty, -R_+) \, \cup\, (R_+, \infty) } \lim_{\theta \to 0, \pi} \bigl( {\rm e}^{-4\gamma } N^2 {\rm e}^{-2 \mbo{\nu}} (q_0)^{ab} \mbo{r} \ptl_a \omega' \bigr)^\theta {\rm d}R {\rm d}t \notag\\
 \qquad{}= \int^t_0 \int_{ (-\infty, -R_+) \, \cup\, (R_+, \infty) } \lim_{\theta \to 0, \pi} \biggl( N^2 (q_0)^{ab} \frac{{\rm e}^{-4\gamma }\mbo{r}}{\bar{\mu}_q} \bar{\mu}_{q_0}\ptl_a \omega' \biggr)^\theta {\rm d}R {\rm d}t.
 \end{gather*}
The integrand
\begin{gather*}
 \lim_{\theta \to 0, \pi} \frac{N^2}{R^2} {\rm e}^{-4\gamma} \frac{\mbo{r}'}{\bar{\mu}_q} \bar{\mu}_{q_0} \ptl_\theta \omega' \notag\\
 \qquad{}= \lim_{\theta \to 0, \pi} \frac{\Delta}{R} \sin^2 \theta\frac{\Sigma^2}{\sin^4 \theta \bigl(\bigl(r^2+a^2\bigr)^2 -a^2 \Delta \sin^2 \theta\bigr)} \frac{\mbo{r}'}{\bar{\mu}_q} \ptl_\theta \omega'
 \to 0 \qquad \text{as} \ \theta \to 0, \pi,
\end{gather*}
due to the rapid decay of $\ptl_\theta \omega'$ as $\theta \to 0$ and $\pi$.
Next consider, ${\rm Flux} \bigl(J_1, \mathcal{H}^+\bigr)$.

We have the term, ${\rm Flux} \bigl(N^2 {\rm e}^{-2 \mbo{\nu}}q^{ab}_0 \mbo{p} \ptl_a \gamma', \mathcal{H}^+\bigr)$, which can be estimated near the future horizon $\mathcal{H}^+$ as
\begin{gather*}
 {\rm Flux} \bigl(N^2 {\rm e}^{-2 \mbo{\nu}}q^{ab}_0 \mbo{p} \ptl_a \gamma', \mathcal{H}^+\bigr) \notag\\
\qquad{} =
\int^t_0 \int^\pi_0 \lim_{R \to R_+} \bigl(
N^2 {\rm e}^{-2 \mbo{\nu}}q^{ab}_0 \mbo{p}' \ptl_a \gamma' \bigr)^R {\rm d} \theta {\rm d}t \notag\\
\qquad{} = \int^t_0 \int^\pi_0 \lim_{R \to R_+} \biggl( N \cdot q^{ab}_0 \frac{N \mbo{p}'}{\bar{\mu}_q} \bar{\mu}_{q_0} \ptl_a \gamma' \biggr)^R {\rm d}\theta {\rm d}t \notag\\
\qquad{} = \int^t_0 \int^\pi_0 \lim_{R \to R_+} \biggl( R \sin \theta\biggl(1- \frac{R^2_+}{R^2} \biggr) \frac{N \mbo{p}' }{\bar{\mu}_q} \bar{\mu}_{q_0} \ptl_R \gamma' \biggr) {\rm d}\theta {\rm d}t.
\end{gather*}
It may be recalled that the behaviour of the canonical pair $(\gamma',\mbo{p}') $ near the future horizon is
\begin{align*}
&4 \ptl_R \gamma'= \biggl(1- \frac{R^2_+}{R^2} \biggr) \qquad \text{and} \qquad \frac{N \mbo{p}'}{\bar{\mu}_q} = \mathcal{O} \biggl(1- \frac{R^2_+}{R^2} \biggr), \notag\\
&{\rm Flux} \bigl(N^2 {\rm e}^{-2 \mbo{\nu}}q^{ab}_0 \mbo{p} \ptl_a \gamma', \mathcal{H}^+\bigr) \to 0 \qquad \text{as} \ R \to R_+.
\end{align*}
Likewise,
\begin{align*}
{\rm Flux} \bigl(N^2 {\rm e}^{-2 \mbo{\nu}}q^{ab}_0 \mbo{r}' \ptl_a \omega', \mathcal{H}^+\bigr)
&{}=
\int^t_0 \int^\pi_0 \lim_{R \to R_+} \big( N^2 {\rm e}^{-2 \mbo{\nu}}q^{ab}_0 \mbo{r}' \ptl_a \omega' \big)^R {\rm d}\theta {\rm d} t \notag\\
&{}= \int^t_0 \int^\pi_0 \lim_{R \to R_+} \left( N {\rm e}^{-4\gamma} \cdot q^{ab}_0 \frac{N {\rm e}^{4\gamma} \mbo{r}'}{\bar{\mu}_q} \bar{\mu}_{q_0} \ptl_a \omega' \right)^R {\rm d}\theta {\rm d}t. \notag
\end{align*}
Again recall that the behaviour of the canonical pair $(\omega', \mbo{r}')$ near the future horizon is
\begin{align*}
\ptl_R \omega' &{}= \biggl(1- \frac{R^2_+}{R^2} \biggr) \qquad \text{and} \qquad \frac{N {\rm e}^{4 \gamma} \mbo{r}'}{\bar{\mu}_q} = \mathcal{O}\biggl(1- \frac{R^2_+}{R^2} \biggr), \end{align*}
the integrand in ${\rm Flux} \bigl(N^2 {\rm e}^{-2 \mbo{\nu}}q^{ab}_0 \mbo{p} \ptl_a \gamma', \mathcal{H}^+\bigr)$ is
\begin{gather*}
= \lim_{R \to R_+} R^2 \sin \theta \biggl(1- \frac{R^2_+}{R^2} \biggr) \frac{\Sigma^2}{ \sin^4 \theta \bigl(\bigl(r^2+a^2\bigr)^2-a^2 \Delta \sin^2 \theta \bigr)^2}
 c \biggl(1-\frac{R^2_+}{R^2} \biggr) \cdot \biggl(1-\frac{R^2_+}{R^2} \biggr) \notag\\
 \to 0 \qquad \text{as} \ R \to R_+.
\end{gather*}

\subsubsection*{Flux of $\boldsymbol{J_1}$ at the outer boundary, $\boldsymbol{{\rm Flux} \bigl(J_1, \bar{\iota}^0\bigr)}$}
 Consider the flux of the term
\begin{align*}
{\rm Flux} \biggl(\frac{N^2}{ \bar{\mu}_q} \mbo{p}' \bar{\mu}_q q^{ab} \ptl_a \gamma', \bar{\iota}^0 \biggr) = \int^t_0 \int^\pi_0 \biggl(\lim_{R \to \infty} \frac{N^2}{ \bar{\mu}_q} \mbo{p}' \bar{\mu}_{q_0} \ptl_R \gamma' \biggr) {\rm d}\theta {\rm d}t,
\end{align*}
we have
\begin{align*}
\frac{N p}{\bar{\mu}_q}' = \mathcal{O} \biggl(\frac{1}{R^3} \biggr), \qquad \gamma' = \frac{1}{R}, \qquad \ptl_R \gamma' = \mathcal{O} \biggl(\frac{1}{R^2} \biggr),
\end{align*}
and recall that
\begin{align*}
N = \bigl(r(R)^2 - 2Mr(R) + a^2\bigr)^{1/2} \sin^2 \theta = \mathcal{O}(R) \qquad \text{for large} \ R.
\end{align*}
Therefore, in the region under consideration, we have for the integrand of
\[
{\rm Flux} \biggl(\frac{N^2}{ \bar{\mu}_q} p' \bar{\mu}_q q^{ab} \ptl_a \gamma', \bar{\iota}^0 \biggr) \sim \mathcal{O}(R) \cdot \mathcal{O} \biggl(\frac{1}{R^3}\biggr) \cdot \mathcal{O} \biggl(\frac{1}{R^2}\biggr) \cdot R \to 0 \qquad\text{as}\ R \to \infty.
\]
Likewise, consider the following flux terms at the outer boundary:
\begin{align*}
&{\rm Flux} \biggl(-\frac{4 N^2}{\bar{\mu}_q} \mbo{r}' \gamma' \bar{\mu}_{q_0} q_0^{ab} \ptl_a \omega, \bar{\iota}^0 \biggr), \qquad {\rm Flux} \biggl(\frac{N^2}{\bar{\mu}_q} \mbo{r}' \bar{\mu}_{q_0} q_0^{ab} \ptl_a \omega', \bar{\iota}^0 \biggr),
\\
&{\rm Flux} \biggl(\frac{N^2}{\bar{\mu}_q} \mbo{r}' \bar{\mu}_{q_0} q^{ab} \ptl_a \omega', \bar{\iota}^0 \biggr) = \int^t_0 \int^\pi_0 \biggl( \lim_{R \to \infty} \frac{N^2}{\bar{\mu}_q} \mbo{r}' \bar{\mu}_{q_0} \ptl_R \omega' \biggr) {\rm d} \theta {\rm d}t.
\end{align*}
It may be recalled that
\begin{align*}
\frac{N {\rm e}^{4 \gamma} \mbo{r}'}{\bar{\mu}_q} = \mathcal{O} \biggl(\frac{1}{R^3} \biggr),
\end{align*}
and
\[
 {\rm e}^{-2\gamma} = \biggl(\frac{r^2 + a^2 \cos^2 \theta}{ \sin^2 \theta}\biggr) \frac{1}{ \bigl(\bigl(r^2 +a^2\bigr)^2 -a^2 \Delta \sin^2 \theta\bigr)} \sim \mathcal{O} \biggl(\frac{1}{R^2}\biggr) \qquad \text{for large} \ R.
\]
Therefore, for the integrand within, we have
\[
{\rm Flux} \biggl(\frac{N^2}{\bar{\mu}_q} \mbo{r}' \bar{\mu}_{q_0} q^{ab} \ptl_a \omega', \bar{\iota}^0 \biggr) = \mathcal{O}(R) \cdot \mathcal{O}\biggl(\frac{1}{R^3} \biggr) \cdot \mathcal{O}\biggl(\frac{1}{R^4} \biggr) R^2 \cdot \frac{1}{R^2} \to 0 \qquad \text{as} \ R \to \infty.
\]
 Next, consider
\begin{align*}
{\rm Flux} \biggl(-\frac{4 N^2}{\bar{\mu}_q} \mbo{r}' \gamma' \bar{\mu}_{q_0} q_0^{ab} \ptl_a \omega, \bar{\iota}^0 \biggr)
&{}= \int^t_0 \int^\pi_0 \biggl(\lim_{R \to \infty} -\frac{4 N^2}{\bar{\mu}_q} \mbo{r}' \gamma' \bar{\mu}_{q_0} \ptl_R \omega \biggr) {\rm d}\theta {\rm d}t \notag\\
&{}= \int^t_0 \int^\pi_0 \biggl(\lim_{R \to \infty} - 4N \biggl(\frac{ N {\rm e}^{4 \gamma}\mbo{r'}}{\bar{\mu}_q} \biggr) {\rm e}^{-4\gamma} \gamma' \bar{\mu}_{q_0} \ptl_R \omega \biggr) {\rm d}\theta {\rm d}t.
\end{align*}
Using the above estimates again and noting that $\ptl_R \omega = \mathcal{O}\bigl(\frac{1}{R^3}\bigr)$,
\begin{align*}
\frac{4 N^2}{\bar{\mu}_q} \mbo{r}' \gamma' \bar{\mu}_{q_0} \ptl_R \omega &{}= \mathcal{O}(R) \cdot \mathcal{O}\biggl(\frac{1}{R^3} \biggr) \cdot \mathcal{O}\biggl(\frac{1}{R^4} \biggr) \cdot \mathcal{O}\biggl(\frac{1}{R} \biggr) R \mathcal{O}\biggl(\frac{1}{R^3} \biggr) \notag\\
&{}= 0 \qquad \text{as} \ R \to \infty,
\end{align*}
from which it follows that ${\rm Flux} \bigl(-\frac{4 N^2}{\bar{\mu}_q} \mbo{r}' \gamma' \bar{\mu}_{q_0} q_0^{ab} \ptl_a \omega, \bar{\iota}^0 \bigr) =0$.

\subsection*{`Kinematic' boundary terms}
Consider
\begin{align*}
(J_2)^b &{}:= U'^A \mathcal{L}_{N'} \bigl(N \bar{\mu}_q q^{ab} h_{AB} (U) \ptl_b U^B\bigr).
\intertext{In a special gauge, we have}
&{}= \gamma' \mathcal{L}_{N'} \bigl(4N \bar{\mu}_q q^{ab} \ptl_a \gamma\bigr) + \omega' \mathcal{L}_{N'} \bigl(N \bar{\mu}_q q^{ab} {\rm e}^{-4 \gamma} \ptl_a \omega\bigr).
\end{align*}

Let us start with ${\rm Flux} (J_2, \Gamma)$. We have
\begin{align*}
{\rm Flux}(J_2, \Gamma)= \int^t_0 \int_{ (-\infty, R_+) \cup (R_+, \infty) } \lim_{\theta \to 0, \pi} \bigl( U'^A \mathcal{L}_{N'} \bigl(N \bar{\mu}_{q_0} q^{ab}_0 h_{AB}(U) \ptl_b U^B\bigr) \bigr)^\theta {\rm d}R {\rm d}t .\notag
\end{align*}
We will expand the expression $U'^A \mathcal{L}_{N'} \bigl(N \bar{\mu}_{q_0} q_0^{ab} h_{AB}(U) \ptl_a U^B\bigr)$ as follows:
\begin{gather}
U'^A \mathcal{L}_{N'} \bigl(N q^{ab}_0\bar{\mu}_{q_0} h_{AB}(U) \ptl_\theta U^B \bigr)= U'^A \bigl(N \ptl_c N'^c \bar{\mu}_{q_0} q^{ab}_0 h_{AB}(U) \ptl_a U^B \notag\\
\qquad{}+ N'^c \ptl_c \bigl(N \bar{\mu}_{q_0} q_0^{ab} h_{AB}(U) \ptl_a U^B\bigr) - \ptl_c N'^b N \bar{\mu}_{q_0} q^{ab}_0 h_{AB} \ptl_a U^B\bigr).\label{Kinematic-expansion}
\end{gather}
We would like to point out that $\omega'$ vanishes on the axes whereas $\gamma'$ does not. Likewise, $\ptl_\theta \omega$ decays rapidly on the axes whereas $\ptl_\theta \gamma$ does not. This causes a few subtleties for terms involving~$\gamma'$ but as will see later, we will have few fortuitous cancellations involving this quantity.
Each of the terms in \eqref{Kinematic-expansion} can in turn be decomposed as
\begin{align*}
U'^A \bigl(N \ptl_c N'^c \bar{\mu}_{q_0} q^{ab}_0 h_{AB}(U) \ptl_a U^B\bigr) &{}= \gamma' \bigl(4 N \ptl_c N'^c \bar{\mu}_{q_0} q^{ab} _0 \ptl_a \gamma\bigr)
 + \omega' \bigl(N \ptl_c N'^c \bar{\mu}_{q_0} q^{ab} _0 {\rm e}^{-4 \gamma} \ptl_a \omega\bigr).
\end{align*}
As a consequence, we have
\begin{gather}
{\rm Flux} \bigl(\gamma' \bigl(4 N \ptl_c N'^c \bar{\mu}_{q_0} q^{ab} _0 \ptl_a \gamma\bigr), \Gamma\bigr) \notag\\
\qquad{}= \int^t_0 \int_{ (-\infty, -R_+) \cup (R_+, \infty) } \lim_{\theta \to 0, \pi} \bigl(\gamma' \bigl(4 N \ptl_c N'^c \bar{\mu}_{q_0} q^{ab} _0 \ptl_a \gamma\bigr) \bigr)^\theta {\rm d}R {\rm d}t \notag\\
\qquad{}= \int^t_0 \int_{ (-\infty, R_+) \cup (R_+, \infty) } \biggl( \lim_{\theta \to 0, \pi} 4 \frac{N}{R} \gamma' \ptl_\theta \gamma \bigl(\ptl_\theta N'^\theta + \ptl_R N'^R\bigr) \biggr) {\rm d}R {\rm d}t.\label{kin-flux-axes-1}
\end{gather}
Let us note that $\lim_{\theta \to 0, \pi} \ptl_\theta \gamma' = \cot \theta$. Likewise,
\begin{gather*}
{\rm Flux} \bigl(\omega' \bigl(\ptl_c N'^c N \bar{\mu}_{q_0} q^{ab} _0 {\rm e}^{-4 \gamma} \ptl_a \omega\bigr), \Gamma\bigr) \notag\\
\qquad{}= \int^t_0 \int_{ (-\infty, R_+) \cup (R_+, \infty) } \lim_{\theta \to 0, \pi} \bigl(\omega' \bigl(\ptl_c N'^c N \bar{\mu}_{q_0} q^{ab} _0 {\rm e}^{-4\gamma} \ptl_a \omega \ptl_b\bigr) \bigr)^\theta {\rm d}R {\rm d}t \notag\\
\qquad{}{}= \int^t_0 \int_{ (-\infty, R_+) \cup (R_+, \infty) } \biggl(\lim_{\theta \to 0, \pi} \frac{N}{R} {\rm e}^{-4\gamma} \omega' \ptl_\theta \omega \ptl_c \bigl(N'^c\bigr) \biggr) {\rm d}R {\rm d}t
=0,
\end{gather*}
where we took into account the fact that the $\ptl_\theta \omega$ term has rapid decay as $\theta \to 0, \pi$.
Next,
\begin{gather*}
U'^A N'^c \ptl_c \bigl(N \bar{\mu}_{q_0} q^{ab}_0 h_{AB} (U) \ptl_a U^A\bigr) = 4\gamma' N'^c\ptl_c \bigl(N \bar{\mu}_{q_0} q^{ab}_0 \ptl_a \gamma\bigr)
+ \omega' N'^c \ptl_c \bigl(N \bar{\mu}_{q_0} {\rm e}^{-4\gamma} \ptl_a \omega\bigr),
\\
{\rm Flux} \bigl(4\gamma' N'^c\ptl_c \bigl(N \bar{\mu}_{q_0} q^{ab}_0 \ptl_a \gamma\bigr), \Gamma\bigr) \notag\\
\qquad{}= \int^t_0 \int_{ (-\infty, R_+) \cup (R_+, \infty) } \lim_{\theta \to 0, \pi} \biggl(
 4\gamma' N'^c\ptl_c \bigl(N \bar{\mu}_{q_0} q^{ab}_0 \ptl_a \gamma\bigr), \frac{1}{R} \ptl_\theta \biggr)^\theta {\rm d}R {\rm d}t \\
\qquad{}= \int^t_0 \int_{ (-\infty, R_+) \cup (R_+, \infty) } \biggl(\lim_{\theta \to 0, \pi} 4 \gamma' N'^c \ptl_c \biggl(N\frac{\bar{\mu}_{q_0}}{R^2} \ptl_\theta \gamma\biggr) \biggr) {\rm d}R {\rm d}t.
\end{gather*}
Consider the integrand
\begin{align} \label{N'^R terms}
 R\gamma' N'^c \ptl_c \biggl(\frac{N}{R} \ptl_\theta \gamma\biggr) = 4 \gamma' \biggl(N'^ \theta \ptl_\theta \biggl(\frac{N}{R} \ptl_\theta \gamma\biggr) + N'^R \ptl_R \biggl(\frac{N}{R} \ptl_\theta \gamma\biggr) \biggr).
\end{align}
The \smash{$N'^R$} term remains while the \smash{$N'^\theta$} term vanishes at the axes (as $\theta \to 0$ or $\pi$). For the sake of brevity, let us combine the $N'^R$ terms in \eqref{N'^R terms} and \eqref{kin-flux-axes-1}. We shall revisit this later. Now then we have
\begin{align} \label{kin-flux-comb-nu'}
\int^t_0 \int_{ (-\infty,-R_+) \cup (\infty, R_+)} \biggl(4 \gamma' \ptl_R \biggl(N'^R \frac{N}{R} \ptl_\theta \gamma\biggr) \biggr) {\rm d}R {\rm d}t.
\end{align}
Similarly,
\begin{align*}
& {\rm Flux} \bigl(\omega' N'^c \ptl_c \bigl(N \bar{\mu}_{q_0} q^{ab}_0 {\rm e}^{-4\gamma} \ptl_a \omega\bigr), \Gamma\bigr) \notag\\
&{}\qquad{}=
\int^t_0 \int_{ (-\infty, R_+) \cup (R_+, \infty) } \lim_{\theta \to 0, \pi} \biggl( \omega' N'^c \ptl_c \bigl(N \bar{\mu}_{q_0} q^{ab}_0 {\rm e}^{-4\gamma} \ptl_a \omega\bigr) \ptl_b \biggr)^\theta {\rm d}R {\rm d}t \\
&{}\qquad{}= \int^t_0 \int_{ (-\infty, R_+) \cup (R_+, \infty) } \biggl( \lim_{\theta \to 0, \pi} \omega' N'^c \ptl_c \biggl(\frac{N \bar{\mu}_{q_0}}{R^2} {\rm e}^{-4\gamma} \ptl_\theta \omega\biggr) \biggr) {\rm d}R {\rm d}t.
\end{align*}
The integrand
\begin{align*}
 \omega' N'^c \ptl_c \biggl(\frac{\bar{\mu}_{q_0}}{R^2} {\rm e}^{-4\gamma} \ptl_\theta \omega\biggr)= \omega' \biggl(N'^\theta \ptl_\theta \biggl(N \frac{\bar{\mu}_{q_0}}{R^2} {\rm e}^{-4\gamma} \ptl_\theta \omega\biggr) + N'^ R \ptl_R \biggl(N \frac{\bar{\mu}_{q_0}}{R^2} {\rm e}^{-4\gamma} \ptl_\theta \omega\biggr)\biggr)
\end{align*}
vanishes at the axes due to $\omega' =0$ on $\Gamma$ and the rapid decay of $\ptl_\theta \omega$ at the axes.
Finally,
\begin{align*}
-U'^A \ptl_c N'^b \bar{\mu}_{q_0} q^{ac}_0 h_{AB} \ptl_a U^B = -4\gamma' \ptl_c N'^b \bar{\mu}_{q_0} \bar{q_0}^{ac} \ptl_a \gamma
-\omega' \ptl_c N'^b \bar{\mu}_{q_0} \bar{q_0}^{ac} {\rm e}^{-4\gamma} \ptl_a \omega,
\end{align*}
so that the fluxes
\begin{align*}
& {\rm Flux} \bigl(-4\gamma' \ptl_c N'^b \bar{q_0}^{ac} \ptl_a \gamma, \Gamma\bigr) \notag\\
&\qquad{}= \int^t_0 \int_{ (-\infty, R_+) \cup (R_+, \infty) } \lim_{\theta \to 0, \pi} \bigl( -4N \gamma' \ptl_c N'^b \bar{\mu}_{q_0} \bar{q_0}^{ac} \ptl_a \gamma \bigr)^\theta {\rm d}R {\rm d}t \notag\\
&\qquad{}= \int^t_0 \int_{ (-\infty, R_+) \cup (R_+, \infty) }\biggl(\lim_{\theta \to 0, \pi} \biggl(-4 \gamma' N \biggl(R \ptl_R N'^\theta \ptl_R \gamma + \frac{1}{R} \ptl_\theta N'^\theta \ptl_\theta \gamma\biggr)\biggr) \biggr) {\rm d}R {\rm d}t.
\end{align*}
If we consider the integrand
\begin{align*}
\lim_{\theta \to 0, \pi} \biggl(4 \gamma' N \biggl(N\ptl_R N'^\theta \ptl_R \gamma + \frac{1}{R} \ptl_\theta N'^\theta \ptl_\theta \gamma\biggr)\biggr),
\end{align*}
we note that the term involving $ \ptl_R N'^R \ptl_R \gamma $ vanishes because $\ptl_R N'^\theta$ vanishes on the axes \big($N'^\theta $ is a constant along the axes\big)
and in the second term involving \smash{$\ptl_\theta N'^\theta \ptl_\theta \gamma$} cancels with a flux term in \eqref{kin-flux-axes-1},
\begin{align*}
& {\rm Flux} \bigl(-\omega' \ptl_c N'b \bar{q_0}^{ac} {\rm e}^{-4\gamma} \ptl_a \omega, \Gamma\bigr) \notag\\
&\qquad{}= \int^t_0 \int_{ (-\infty, R_+) \cup (R_+, \infty) } \lim_{\theta \to 0, \pi} \bigl(
-N \omega' \ptl_c N'^b \bar{q_0}^{ac} {\rm e}^{-4\gamma} \ptl_a \omega \bigr)^\theta {\rm d}R {\rm d}t \notag\\
&\qquad{}= \int^t_0 \int_{ (-\infty, R_+) \cup (R_+, \infty) } \Bigl(\lim_{\theta \to 0, \pi} \bigl(- N \omega' (\ptl_R N'^ \theta {\rm e}^{-4\gamma} \ptl_R \omega + \ptl_\theta N'^ R {\rm e}^{-4 \gamma} \ptl_\theta \omega\bigr) \Bigr) {\rm d}R {\rm d}t
=0
\end{align*}
on account of rapid decay of $\ptl_\theta \omega$ at the axes, vanishing of $N \sim \sin \theta $ and $\omega'$ at the axes.

Now then for
\begin{align*}
{\rm Flux} \bigl(J_2, \mathcal{H}^+\bigr),
\end{align*}
we have
\begin{align*}
{\rm Flux} \bigl(\gamma' \bigl(4 \ptl_c N'^c \bar{\mu}_{q_0} q^{ab} _0 \ptl_a \gamma, \mathcal{H}^+\bigr)\bigr)
= \int^t_0 \int^\pi_0 \Bigl(\lim_{R \to R_+} 4 R N \gamma' \ptl_R \gamma \ptl_c N'^c\Bigr) {\rm d} \theta {\rm d}t.
\end{align*}

Using
\begin{align*}
\lim_{R \to R_+} \ptl_R \gamma \leq c, \qquad \lim_{R \to R_+} \ptl_c N'^c \leq c,
\end{align*}
 we estimate the integrand
\begin{align*}
\lim_{R \to R_+} 4 R N \gamma' \ptl_R \gamma \ptl_c N'^c \leq c \biggl(1- \frac{R^2_+}{R^2}\biggr)^2.
\end{align*}

So we have flux
\begin{gather*}
{\rm Flux} \bigl(\gamma' \bigl(4 \ptl_c N'^c \bar{\mu}_{q_0} q^{ab} _0 \ptl_a \gamma, \mathcal{H}^+\bigr)\bigr) =0,
\\
{\rm Flux} \bigl(\omega' \bigl(\ptl_c N'^c \bar{\mu}_{q_0} q^{ab} _0 {\rm e}^{-4 \gamma} \ptl_a \omega\bigr), \mathcal{H}^+\bigr)
= \int^t_0 \int^\pi_0 \Bigl(\lim_{R \to R_+} \omega' \ptl_c N'^c N \bar{\mu}_{q_0} {\rm e}^{-4\gamma} \ptl_R \omega\Bigr) {\rm d} \theta {\rm d}t \notag\\
\hphantom{{\rm Flux} \bigl(\omega' \bigl(\ptl_c N'^c \bar{\mu}_{q_0} q^{ab} _0 {\rm e}^{-4 \gamma} \ptl_a \omega\bigr), \mathcal{H}^+\bigr)}{}
= \int^t_0 \int^\pi_0 \Bigl(\lim_{R \to R_+} R N {\rm e}^{-4\gamma} \omega' \ptl_R \omega \ptl_c N'^c\Bigr)R {\rm d}\theta {\rm d}t.
\end{gather*}
Noting that
\begin{align*}
\lim_{R \to R_+} \ptl_R \omega \leq c, \qquad \lim_{R \to R_+} {\rm e}^{-4\gamma} \leq c
\end{align*}
the integrand
\begin{align*}
R N {\rm e}^{-4\gamma} \omega' \ptl_R \omega \ptl_c N'^c \leq c \biggl(1- \frac{R_+}{R^2}\biggr)^2.
\end{align*}
Thus, again
\begin{align*}
{\rm Flux} \bigl(\omega' \bigl(\ptl_c N'^c \bar{\mu}_{q_0} q^{ab} _0 {\rm e}^{-4 \gamma} \ptl_a \omega\bigr)\bigr) =0.
\end{align*}
Next,
\begin{align*}
& {\rm Flux} \bigl(4\gamma' N'^c\ptl_c \bigl(N \bar{\mu}_{q_0} q^{ab}_0 \ptl_a \gamma\bigr), \mathcal{H}^+\bigr) \notag \\
&\qquad{}= \int^t_0 \int^\pi_0 \Bigl(\lim_{R \to R_+} 4 \gamma' N'^c (\ptl_c N R \ptl_R \gamma + N \ptl_c (R \ptl_R \gamma) ) \Bigr) {\rm d}R {\rm d}t
\end{align*}
the integrand can be estimated close to the horizon as
\begin{align*}
4 \gamma' N'^c (\ptl_c N R \ptl_R \gamma + N \ptl_c (R \ptl_R \gamma) ) \leq c \biggl(1- \frac{R^2_+}{R^2}\biggr) + c \biggl(1- \frac{R^2_+}{R^2}\biggr)^2,
\end{align*}
so we have
\begin{gather*}
{\rm Flux} \bigl(4\gamma' N'^c\ptl_c \bigl(N \bar{\mu}_{q_0} q^{ab}_0 \ptl_a \gamma\bigr), \mathcal{H}^+\bigr) =0,
\\
{\rm Flux} \bigl(\omega' N'^c \ptl_c \bigl(\bar{\mu}_{q_0} {\rm e}^{-4\gamma} \ptl_a \omega\bigr), \mathcal{H}^+\bigr) \notag\\
\qquad{}=
\int^t_0 \int^\pi_0 \Bigl(\lim_{R \to R_+} \omega' N'^c \ptl_c\bigl(N R {\rm e}^{-4\gamma} \ptl_R \omega\bigr) \Bigr) {\rm d}\theta {\rm d}t \notag\\
\qquad{}= \int^t_0 \int^\pi_0 \Bigl(\lim_{R \to R_+} \omega' N'^c \bigl(\ptl_c N R {\rm e}^{-4\gamma} \ptl_R \omega + N \ptl_c \bigl({\rm e}^{-4\gamma} R \ptl_R \omega\bigr)\bigr) \Bigr) {\rm d} \theta {\rm d}t
\end{gather*}
the integrand can be estimated as
\begin{align*}
\omega' N'^c \bigl(\ptl_c N R {\rm e}^{-4\gamma} \ptl_R \omega + N \ptl_c \bigl({\rm e}^{-4\gamma} R \ptl_R \omega\bigr)\bigr) \leq c \biggl(1- \frac{R^2_+}{R^2}\biggr) + c \biggl(1- \frac{R^2_+}{R^2}\biggr)^2,
\end{align*}
so we have
\begin{align*}
{\rm Flux} \bigl(\omega' N'^c \ptl_c \bigl(\bar{\mu}_{q_0} {\rm e}^{-4\gamma} \ptl_a \omega\bigr), \mathcal{H}^+\bigr) =0.
\end{align*}
Next,
\begin{align*}
& {\rm Flux} \bigl(-4 N \gamma' \ptl_c N'^b \bar{q_0}^{ac} \ptl_a \gamma, \mathcal{H}^+\bigr) \notag\\
&{}= \int^t_0 \int^\pi_0 \biggl(\lim_{R \to R_+} -4 N \gamma' \biggl(\ptl_R N'^R \ptl_R \gamma + \frac{1}{R^2} \ptl_\theta N'^ \theta \ptl_ \theta \gamma\biggr) \biggr) {\rm d}\theta {\rm d}t
 =0
\end{align*}
in view of the behaviour of the integrand
\begin{align*}
 4 N \gamma' \biggl(\ptl_R N'^R \ptl_R \gamma + \frac{1}{R^2} \ptl_\theta N'^ \theta \ptl_ \theta \gamma\biggr) \qquad \text{behaves as} \ c\biggl(1- \frac{R^2_+}{R^2}\biggr)^2 \ \text{as} \ R \to R_+.
\end{align*}
Finally,
\begin{align*}
& {\rm Flux} \bigl(-\omega' \ptl_c N'^b \bar{q_0}^{ac} N {\rm e}^{-4\gamma} \ptl_a \omega, \mathcal{H}^+\bigr) \notag \\
 &{}\qquad{}= \int^t_0 \int^\pi_0 \biggl(\lim_{R \to R_+} - N {\rm e}^{-4 \gamma} \omega' \biggl(\ptl_R N'^R \ptl_R \omega + \frac{1}{R^2} \ptl_\theta N'^\theta \ptl_\theta \omega\biggr) \biggr) {\rm d}\theta {\rm d}t
= 0,
\end{align*}
where again the integrand can be estimated such that
\[
N {\rm e}^{-4 \gamma} \omega' \biggl(\ptl_R N'^R \ptl_R \omega + \frac{1}{R^2} \ptl_\theta N'^\theta \ptl_\theta \omega\biggr)
\]
behaves as
\[
 c\biggl(1- \frac{R^2_+}{R^2}\biggr)^2 \qquad \text{for $R$ close to $R_+$}.
\]
Next, for the kinematic fluxes at the spatial infinity
\begin{align*}
{\rm Flux}\bigl(J_2, \iota^0\bigr)
\end{align*}
let us start with estimating the divergence term $\displaystyle \ptl_c N'^c$ near the spatial infinity. We have
\begin{align*}
\ptl_c N'^c &{}= \ptl_\theta N'^ \theta + \ptl_R N'^R \notag\\
&{}= \ptl_R \Biggl(\ptl_t Y^R - N^2 {\rm e}^{-2 \nu} \ptl_R \Biggl(Y^t(t=0) - \mathcal{Y}^ \theta \cot \theta - \mathcal{Y}^R \frac{ 1 + \frac{R^2_+}{R^2}}{ R \bigl(1- \frac{R^2_+}{R^2}\bigr)} \Biggr) \Biggr) \notag \\
&\quad{} + \ptl_\theta \Biggl(\ptl_t Y^\theta - \frac{N^2 {\rm e}^{-2\mbo{\nu}'}}{R^2} \Biggl(Y^t(t=0) - \mathcal{Y}^\theta \cot \theta - \mathcal{Y}^R \frac{ 1 + \frac{R^2_+}{R^2}}{ R \bigl(1- \frac{R^2_+}{R^2}\bigr)}\Biggr)\Biggr).
\end{align*}
We have the following behaviour for large $R$:
\begin{align*}
& \ptl_R \bigl(N^2 {\rm e}^{-2 \mbo{\nu}}\bigr) \sim \mathcal{O}\biggl(\frac{1}{R}\biggr),
\qquad \ptl_\theta \frac{N^2 {\rm e}^{-2 \mbo{\nu}'}}{R^2} \sim \mathcal{O} (1)
\end{align*}
and
\begin{align*}
& \ptl_R \frac{ 1 + \frac{R^2_+}{R^2}}{ R \bigl(1- \frac{R^2_+}{R^2}\bigr)} \sim \mathcal{O} (\frac{1}{R^2}), \qquad \ptl^2_R \frac{ 1 + \frac{R^2_+}{R^2}}{ R \bigl(1- \frac{R^2_+}{R^2}\bigr)} \sim \mathcal{O} (\frac{1}{R^3})
\end{align*}
from somewhat lengthy but standard computations. It now follows that the coordinate divergence of $N'$ is then
\begin{gather*}
\ptl_c N'^c = \mathcal{O} \biggl(\frac{1}{R}\biggr),
\\
{\rm Flux} \bigl(\gamma' \bigl(4 \ptl_c N'^c \bar{\mu}_{q_0} q^{ab} _0 \ptl_a \gamma, \iota^0\bigr)\bigr) =
\int^t_0 \int^\pi_0 \Bigl(\lim_{R \to \infty} 4 R \gamma' \ptl_R \gamma \ptl_c N'^c \bar{\mu}_{q_0}\Bigr) {\rm d} \theta {\rm d}t,
\\
4 R \gamma' \ptl_R \gamma \ptl_c N'^c \bar{\mu}_{q_0} = 4 R \cdot \mathcal{O} \biggl(\frac{1}{R}\biggr) \cdot \mathcal{O} \biggl(\frac{1}{R}\biggr) \cdot \mathcal{O} \biggl(\frac{1}{R}\biggr)
 \to 0 \qquad \text{as} \ R \to \infty.
\end{gather*}
Thus,
\begin{gather*}
{\rm Flux} \bigl(\gamma' \bigl(4 \ptl_c N'^c \bar{\mu}_{q_0} q^{ab} _0 \ptl_a \gamma, \iota^0\bigr)\bigr) =0,
\\
{\rm Flux} \bigl(\omega' \bigl(\ptl_c N'^c \bar{\mu}_{q_0} q^{ab} _0 {\rm e}^{-4 \gamma} \ptl_a \omega\bigr), \iota^0\bigr) \notag\\
 \qquad{} = \int^t \int^\pi_0 \Bigl(\lim_{R \to \infty} R {\rm e}^{-4\gamma} \omega' \ptl_R \omega \ptl_c\bigl(\ptl_t Y^c - N^2 {\rm e}^{-2 \mbo{\nu}} q^{ac}_0 \ptl_a Y^t\bigr)\Bigr) {\rm d} \theta {\rm d}t,
\\
R {\rm e}^{-4\gamma} \omega' \ptl_R \omega \ptl_c\bigl(\ptl_t Y^c - N^2 {\rm e}^{-2 \mbo{\nu}} q^{ac}_0 \ptl_a Y^t\bigr) = R \cdot \mathcal{O} \biggl(\frac{1}{R^4}\biggr) \cdot \mathcal{O} \biggl(\frac{1}{R}\biggr) \cdot \mathcal{O} \biggl(\frac{1}{R^3}\biggr) \cdot \mathcal{O} \biggl(\frac{1}{R}\biggr) \notag\\
\hphantom{R {\rm e}^{-4\gamma} \omega' \ptl_R \omega \ptl_c\bigl(\ptl_t Y^c - N^2 {\rm e}^{-2 \mbo{\nu}} q^{ac}_0 \ptl_a Y^t\bigr)}{}
\to 0 \qquad \text{as} \ R \to \infty,
\\
{\rm Flux} \bigl(\omega' \bigl(\ptl_c N'^c \bar{\mu}_{q_0} q^{ab} _0 {\rm e}^{-4 \gamma} \ptl_a \omega\bigr), \iota^0\bigr)\bigr) =0.
\end{gather*}
Analogously, noting that $\gamma' = \mathcal{O}\bigl(\frac{1}{R}\bigr)$ and $\omega' = \mathcal{O} \bigl(\frac{1}{R}\bigr)$ for large $R$ we have the following:
\begin{gather*}
{\rm Flux} \bigl(4\gamma' N'^c\ptl_c \bigl(\bar{\mu}_{q_0} q^{ab}_0 \ptl_a \gamma\bigr), \iota^0\bigr) \\
\qquad{} = \int^t_0 \int^\pi_0 \Bigl(\lim_{R \to \infty} 4 \gamma' N'^c \ptl_c (R \ptl_R \gamma)\Bigr) {\rm d}\theta {\rm d}t \notag\\
\qquad{} = \int^t_0 \int^\pi_0 \bigl(4 \gamma' \bigl(\ptl_t Y^c - N^2 {\rm e}^{-2 \mbo{\nu}} q^{ac} \ptl_a Y^t\bigr) \ptl_c (R \ptl_R \gamma) \bigr) {\rm d}\theta {\rm d}t
=0,
\\
 {\rm Flux} \bigl(\omega' N'^c \ptl_c \bigl(\bar{\mu}_{q_0} {\rm e}^{-4\gamma} \ptl_a \omega\bigr), \iota^0\bigr) \notag\\
\qquad{} = \int^t_0 \int^\pi_0 \Bigl(\lim_{R \to \infty} \omega' \bigl(\ptl_t Y^c - N^2 {\rm e}^{-2 \mbo{\nu}} q^{ac} \ptl_a Y^t\bigr) \ptl_c \bigl(R {\rm e}^{-4\gamma} \ptl_R \omega\bigr)\Bigr) {\rm d}\theta {\rm d}t
 =0,
\\
 {\rm Flux} \bigl(-4\gamma' \ptl_c N'^b \bar{q_0}^{ac} \ptl_a \gamma, \iota^0\bigr) \notag\\
\qquad{} = \int^t_0 \int^\pi_0 \biggl(\lim_{R \to \infty} -4 \gamma' \biggl(\ptl_R N'^R \ptl_R \gamma + \frac{1}{R^2} \ptl_\theta N'^R \ptl_\theta \gamma \biggr)\biggr) {\rm d}\theta {\rm d}t
 = 0,
\\
 {\rm Flux} \bigl(-\omega' \ptl_c N'^b \bar{q_0}^{ac} {\rm e}^{-4\gamma} \ptl_a \omega, \iota^0\bigr) \notag\\
\qquad{} = \int^t_0 \int^\pi_0 \biggl(\lim_{R \to \infty} \omega' {\rm e}^{-4\gamma} \biggl(\ptl_R N'^R \ptl_R \omega + \frac{1}{R^2} \ptl_\theta N'^R \ptl_\theta \omega \biggr)\biggr) {\rm d}\theta {\rm d}t
 =0.
\end{gather*}

\subsection*{`Conformal' boundary terms $\boldsymbol{
(J_3)^b := \mathcal{L}_{N'} (N) \bigl(2 \bar{\mu}_q q^{ab} \ptl_a \mbo{\nu}'\bigr)}$.}

Let us start with the flux of this conformal term at the axes:
\begin{gather*}
 {\rm Flux} (J_3, \Gamma) \\
 \qquad{} = \int^t_0 \int_{(-\infty, -R_+) \cup (R_+, \infty)} \lim_{\theta \to 0, \pi} \bigl(\mathcal{L}_{N'} (N) \bigl(2 \bar{\mu}_{q_0}q^{ab}_0 \ptl_a \mbo{\nu}\bigr) \bigr)^\theta {\rm d}R {\rm d}t \notag\\
 \qquad{} = \int^t_0 \int_{(-\infty, -R_+) \cup (R_+, \infty)} \bigl( 2 \ptl_\theta \mbo{\nu}' \bigl(N'^R \ptl_R N + N'^\theta \ptl_\theta N\bigr) \bigr) {\rm d}R {\rm d}t \notag\\
 \qquad{} = \int^t_0 \int_{(-\infty, -R_+) \cup (R_+, \infty)} \biggl(\lim_{\theta \to 0, \pi} 2 \ptl_\theta \mbo{\nu}' \biggl(N'^R \frac{\sin \theta R^2_+}{R^2} + N'^\theta R \cos \theta\biggl(1-\frac{R^2_+}{R^2} \biggr) \biggr)\biggr) {\rm d}R {\rm d}t.
\end{gather*}
Expanding out the integrand within ${\rm Flux} (J_3, \Gamma) $, near the axes $\Gamma$, we get
\begin{gather*}
 \lim_{\theta \to 0, \pi} 2 \ptl_\theta \mbo{\nu}' \Biggl\{ \Biggl(\ptl_t Y^R -\biggl(1 - \frac{R^2_+}{R^2} \biggr)^2 \cdot \frac{R^4}{ \bigl(\bigl(r^2+a^2\bigr)^2 -a^2 \Delta \sin^2 \theta\bigr)} \notag\\
 \hphantom{\lim_{\theta \to 0, \pi} 2 \ptl_\theta \mbo{\nu}' \Biggl\{ \Biggl(}{}
 \times \ptl_R \Biggl(Y^t(t=0) - \mathcal{Y}^R \frac{1+ \frac{R^2_+}{R^2}}{1-\frac{R^2_+}{R^2}} - \mathcal{Y}^\theta \cot \theta \Biggr) \Biggr)\frac{\sin \theta R^2_+}{R^2} \notag\\
\hphantom{ \lim_{\theta \to 0, \pi} 2 \ptl_\theta \mbo{\nu}' \Biggl\{}{} + \Biggl(\ptl_t Y^\theta - \biggl(1- \frac{R^2_+}{R^2}\biggr)^2 \cdot \frac{R^4}{ \bigl(\bigl(r^2+a^2\bigr)^2 -a^2 \Delta
	\sin^2 \theta \bigr)} \notag\\
\hphantom{ \lim_{\theta \to 0, \pi} 2 \ptl_\theta \mbo{\nu}' \Biggl\{ +\Biggl(}{} \times \ptl_\theta \Biggl(Y^t (t=0) - \mathcal{Y}^R \frac{1 + \frac{R^2_+}{R^2}} {{1-\frac{R^2_+}{R^2}} } - \mathcal{Y}^ \theta \cot \theta \Biggr) \Biggr) R \cos \theta \biggl(1-\frac{R^2_+}{R^2}\biggr) \Biggr\}
\end{gather*}
from which it follows that
\[
{\rm Flux} (J_3, \Gamma)= 0
\]
due to the decay rate of the terms in the parenthesis.

Now then, consider the flux of $J_3$ at the horizon $H^+$:
\begin{gather*}
{\rm Flux} \bigl(J_3, \mathcal{H}^+\bigr) = \int^t_0 \int^\pi _0 \Bigl(\lim_{R \to R_+} \mathcal{L}_{N'} (N) \bigl(2 \bar{\mu}_q q^{ab} \ptl_a \mbo{\nu}'\bigr) \Bigr) {\rm d}\theta {\rm d}t \notag\\
 \hphantom{{\rm Flux} \bigl(J_3, \mathcal{H}^+\bigr) =}{}
\times{}\int^t_0 \int^\pi_0 \Bigl(\lim_{R \to R_+} 2 R \mathcal{L}_{N'} N \ptl_R \mbo{\nu}' \Bigr) {\rm d}\theta {\rm d}t \notag\\
 \hphantom{{\rm Flux} \bigl(J_3, \mathcal{H}^+\bigr)}{}
= \int^t_0 \int^\pi _0 \biggl(\lim_{R \to R_+} 2 R \ptl_R \mbo{\nu}' \biggl(\bigl(\ptl_t Y^R - N^2 {\rm e}^{-2 \mbo{\nu}} \ptl_R Y^t\bigr) \sin \theta \frac{R^2_+}{R^2} \notag\\
 \hphantom{{\rm Flux} \bigl(J_3, \mathcal{H}^+\bigr)= \int^t_0 \int^\pi _0 \biggl(}{}
 + \biggl(\ptl_t Y^\theta - \frac{N^2}{R^2} {\rm e}^{-2 \mbo{\nu}} \ptl_\theta Y^t\biggr) R \cos \theta \biggl(1-\frac{R^2_+}{R^2}\biggr) \biggr) \biggr) {\rm d}\theta {\rm d}t
\end{gather*}
expanding out the integrand within $ {\rm Flux} \bigl(J_3, \mathcal{H}^+\bigr)$ close to the horizon $\mathcal{H}^+$, we have
\begin{gather*}
 \lim_{R \to R_+} 2R \ptl_R \mbo{\nu}' \Bigg\{ \Biggl(\ptl_t Y^R - \biggl(1 - \frac{R^2_+}{R^2} \biggr)^2 \cdot \frac{R^4}{ \bigl(\bigl(r^2+a^2\bigr)^2 -a^2 \Delta \sin^2 \theta\bigr)} \\
 \hphantom{\lim_{R \to R_+} 2R \ptl_R \mbo{\nu}' \Bigg\{\Biggl(}{}
 \times \ptl_R \Biggl(Y^t(t=0) - \mathcal{Y}^R \frac{1+ \frac{R^2_+}{R^2}}{1-\frac{R^2_+}{R^2}} - \mathcal{Y}^\theta \cot \theta \Biggr) \Biggr)\sin \theta \frac{R^2_+}{R^2} \notag\\
\hphantom{\lim_{R \to R_+} 2R \ptl_R \mbo{\nu}' \Bigg\{ }{}
 + \Biggl( \ptl_t Y^\theta - \biggl(1- \frac{R^2_+}{R^2}\biggr)^2 \cdot \frac{R^4}{ \bigl(\bigl(r^2+a^2\bigr)^2 -a^2 \Delta
	\sin^2 \theta \bigr)} \notag\\
\hphantom{\lim_{R \to R_+} 2R \ptl_R \mbo{\nu}' \Bigg\{ + \Biggl(}{}
 \ptl_\theta \Biggl(Y^t (t=0) - \mathcal{Y}^R \frac{1 + \frac{R^2_+}{R^2}} {{1-\frac{R^2_+}{R^2}} } - \mathcal{Y}^ \theta \cot \theta \Biggr) \Biggr) \Bigg\} \to 0 \qquad \text{as} \ R \to R_+
\end{gather*}
after plugging in the expression of $Y$ and $\mathcal{Y}$ for $R$ near $R_+$.
Likewise, at spatial infinity, we have
\begin{align*}
{\rm Flux} (J_3, \bar{\iota}^0)= \int^t_0 \int^\pi _0 \biggl(&\lim_{R \to \infty} 2 R \ptl_R \mbo{\nu}' \biggl(\bigl(\ptl_t Y^R - N^2 {\rm e}^{-2 \mbo{\nu}} \ptl_R Y^t\bigr) \sin \theta \frac{R^2_+}{R^2} \notag\\
&{} + \biggl(\ptl_t Y^\theta - \frac{N^2}{R^2} {\rm e}^{-2 \mbo{\nu}} \ptl_\theta Y^t\biggr) R \cos \theta \biggl(1-\frac{R^2_+}{R^2}\biggr) \biggr) \biggr) {\rm d}\theta {\rm d}t.
\end{align*}
Again the integrand occurring above can be estimated, after plugging in the relevant quantities:
\begin{gather*}
= \lim_{R \to \infty} 2R \ptl_R \mbo{\nu}' \Bigg\{ \Biggl(\ptl_t \sum^\infty_{n=1} Y^R_n R^{-n+1} \cos n \theta - \biggl(1 - \frac{R^2_+}{R^2} \biggr)^2 \cdot \frac{R^4}{ \bigl( \bigl(r^2+a^2\bigr)^2 -a^2 \Delta \sin^2 \theta\bigr)} \\
 \qquad\quad{} \times \ptl_R \Biggl(Y^t(t=0) + \sum^\infty_{n =1} \mathcal{Y}^R_n R^{-n+1} \frac{1+ \frac{R^2_+}{R^2}}{1-\frac{R^2_+}{R^2}} - \sum^\infty_{n=1} \mathcal{Y}^\theta_n R^{-n+1} \sin n \theta \cot \theta \Biggr) \Biggr)\sin \theta \frac{R^2_+}{R^2} \notag\\
 \qquad{} + \Biggl( \ptl_t \sum^{\infty}_{n=1} Y^\theta_n R^{-n} \sin n \theta - \biggl(1- \frac{R^2_+}{R^2}\biggr)^2 \cdot \frac{R^4}{ \bigl(\bigl(r^2+a^2\bigr)^2 -a^2 \Delta
	\sin^2 \theta \bigr)} \notag\\
 \qquad\quad{} \times \ptl_\theta \Biggl(Y^t (t=0) + \sum^\infty_{n=1} \mathcal{Y}^R_n R^{-n+1} \cos n \theta \frac{1 + \frac{R^2_+}{R^2}} {{1-\frac{R^2_+}{R^2}} } - \sum^\infty_{n=1} \mathcal{Y}^\theta_n R^{-n} \sin n \theta \cot \theta \Biggr) \Biggr) \Bigg\} \notag\\
 \to 0 \qquad \text{as} \ R \to \infty.
\end{gather*}
We would like to point out that the choice of $Y^t(t=0)$ is at our discretion.
Next, let us consider~${\rm Flux} (J_4, \Gamma)$.

We have the flux expression at the axes:
\begin{gather*}
{\rm Flux} (J_4, \Gamma) = \int^t_0 \int_{ (-\infty,-R_
	+) \cup (R_+, \infty)} \lim_{\theta \to 0, \pi} 2 R \biggl(\mathcal{L}_{N'} \mbo{\nu}' \frac{1}{R^2} \ptl_\theta N \biggr) {\rm d}R {\rm d}t
\notag\\
\quad{}= \int^t_0 \int_{ (-\infty,-R_
	+) \cup (R_+, \infty)}\biggl(\lim_{\theta \to 0, \pi} 2 \cos \theta \biggl(1 - \frac{R^2_+}{R^2} \biggr) \bigl(\ptl_t Y^c - N^2 {\rm e}^{-2 \mbo{\nu}} q^{ac} \ptl_a Y^t \bigr) \ptl_c \mbo{\nu}' \biggr) {\rm d}R {\rm d}t,
\end{gather*}
expanding out the integrand above, we have
\begin{align*}
\quad&{} = \lim_{\theta \to 0, \pi} 2 \cos \theta \biggl(1 - \frac{R^2_+}{R^2} \biggr)
\Biggl\{ \ptl_t Y^R \ptl_R \mbo{\nu}' - \biggl (1- \frac{R^2_+}{R^2} \biggr)^2 \cdot \frac{R^4 \ptl_R \mbo{\nu}'}{ \bigl(\bigl(r^2+a^2\bigr)^2 -a^2 \Delta
	\sin^2 \theta \bigr)} \notag\\
&\hphantom{ = \lim_{\theta \to 0, \pi} 2 \cos \theta \biggl(1 - \frac{R^2_+}{R^2} \biggr)\Biggl\{}{}
\times \ptl_R \Biggl(Y^t (t=0) - \mathcal{Y}^R \frac{1 + \frac{R^2_+}{R^2}} {{1-\frac{R^2_+}{R^2}} } - \mathcal{Y}^ \theta \cot \theta \Biggr) \notag\\
&\hphantom{ = \lim_{\theta \to 0, \pi} 2 \cos \theta \biggl(1 - \frac{R^2_+}{R^2} \biggr)\Biggl\{}{}
+ \ptl_t Y^\theta \ptl_\theta \mbo{\nu}' - \biggl(1 - \frac{R^2_+}{R^2} \biggr)^2 \cdot \frac{R^4 \ptl_\theta \mbo{\nu}'}{ \bigl(\bigl(r^2+a^2\bigr)^2 -a^2 \Delta \sin^2 \theta\bigr)} \notag\\
&\hphantom{ = \lim_{\theta \to 0, \pi} 2 \cos \theta \biggl(1 - \frac{R^2_+}{R^2} \biggr)\Biggl\{}{}
\times \ptl_\theta \Biggl(Y^t(t=0) - \mathcal{Y}^R \frac{1+ \frac{R^2_+}{R^2}}{1-\frac{R^2_+}{R^2}} - \mathcal{Y}^\theta \cot \theta \Biggr) \Biggr\}.
\end{align*}
It may be noted that $N'^\theta$ now vanishes using the expansion of $\sin n \theta $.
A term related to $\mathcal{Y}^R$ and thus $N'^R$ remains. Let us re-compress this term and represent the integrand as
\[
2N'^R \ptl_R \mbo{\nu}' \cos \theta\biggl(1- \frac{R^2_+}{R^2}\biggr),
\]
which can be re-expressed as
\[
\ptl_R \biggl(2 N'^R \ptl_R \mbo{\nu}' \cos \theta \biggl(1- \frac{R^2_+}{R^2} \biggr) \biggr) - 2 \mbo{\nu'}\ptl_R \biggl(N'^R \cos \theta \biggl(1-\frac{R^2_+}{R^2} \biggr) \biggr).
\]
Now, the flux corresponding to the second term above fortuitously combines with the remaining flux term in \eqref{kin-flux-comb-nu'} to yield $0$ total flux, in view of the regularity condition \eqref{gamma-nu condition}. Further, the first term is a total divergence term and converges to $0$ at the boundaries of the axes ($R \to \infty$), in view of the asymptotic behaviour of $N'^R$ and $\mbo{\nu'}$.

Subsequently,
\begin{align*}
{\rm Flux} (J_4, \mathcal{H}^+) &{}= \int^t_0 \int^\pi _0 \Bigl(\lim_{R \to R_+} 2 \mathcal{L}_{N'} \mbo{\nu}' \bar{\mu}_q q^{ab} \ptl_a N \Bigr) {\rm d}\theta {\rm d}t \notag\\
&{}= \int^t_0 \int^\pi _0 \Bigl(\lim_{R \to R_+} 2 \mathcal{L}_{N'} \mbo{\nu}' \bar{\mu}_{q_0} \ptl_R N \Bigr) R {\rm d}\theta {\rm d}t \notag\\
&{}= \int^t_0 \int^\pi _0 \biggl(2 \sin \theta \biggl(1+ \frac{R^2_+}{R^2} \biggr) \bigl(\ptl_t Y^c - N^2 {\rm e}^{-2 \mbo{\nu}} q^{ac} \ptl_a Y^t\bigr) \ptl_c \mbo{\nu}' \biggr) {\rm d}\theta {\rm d}t,
\end{align*}
the integrand is
\begin{align*}
\hphantom{{\rm Flux} (J_4, \mathcal{H}^+)}&{}= 2 \sin \theta \biggl(1+ \frac{R^2_+}{R^2} \biggr) \Biggl\{ \ptl_t Y^R \ptl_R \mbo{\nu}' - \biggl(1- \frac{R^2_+}{R^2} \biggr)^2 \cdot \frac{R^4 \ptl_R \mbo{\nu}'}{ \bigl( \bigl(r^2+a^2\bigr)^2 -a^2 \Delta
	\sin^2 \theta \bigr)} \notag\\
& \hphantom{= 2 \sin \theta \biggl(1+ \frac{R^2_+}{R^2} \biggr) \Biggl\{}{}
\times \ptl_R \Biggl(Y^t (t=0) - \mathcal{Y}^R \frac{1 + \frac{R^2_+}{R^2}} {{1-\frac{R^2_+}{R^2}} } - \mathcal{Y}^ \theta \cot \theta \Biggr) \notag\\
&\hphantom{= 2 \sin \theta \biggl(1+ \frac{R^2_+}{R^2} \biggr) \Biggl\{}{}
+ \ptl_t Y^\theta \ptl_\theta \mbo{\nu}' - \biggl(1 - \frac{R^2_+}{R^2} \biggr)^2 \cdot \frac{R^4 \ptl_\theta \mbo{\nu}'}{ \bigl(\bigl(r^2+a^2\bigr)^2 -a^2 \Delta \sin^2 \theta\bigr)} \notag\\
& \hphantom{= 2 \sin \theta \biggl(1+ \frac{R^2_+}{R^2} \biggr) \Biggl\{}{}
\times \ptl_\theta \Biggl(Y^t(t=0) - \mathcal{Y}^R \frac{1+ \frac{R^2_+}{R^2}}{1-\frac{R^2_+}{R^2}} - \mathcal{Y}^\theta \cot \theta \Biggr) \Biggr\}.
\end{align*}

The second term in the curly brackets rapidly vanishes due to the decay rate of $Y^ \theta$ close to $\mathcal{H}^+$. The non vanishing factor in the second line above is accounted for by the vanishing boundary condition of $\ptl_R \mbo{\nu}'$ at the horizon,
\begin{align*}
{\rm Flux} \bigl(J_4, \bar{\iota}^0) &{}= \int^t_0 \int^\pi _0 \Bigl(\lim_{R \to \infty} 2 \mathcal{L}_{N'} \mbo{\nu}' \bar{\mu}_{q_0} \ptl_R N \Bigr) {\rm d}\theta {\rm d}t \notag \\
&{}= \int^t_0 \int^\pi _0 \biggl( \lim_{R \to \infty} 2 \sin \theta \frac{R^2_+}{R} \bigl(\ptl_t Y^c - N^2 {\rm e}^{-2 \mbo{\nu}} q^{ac} \ptl_a Y^t\bigr) \ptl_c \mbo{\nu}' \biggr) {\rm d} \theta {\rm d}t,
\end{align*}
so that the integrand involves
\begin{align*}
\qquad{} &{}= \lim_{R \to \infty} 2 \sin \theta \biggl(1+ \frac{R^2_+}{R^2} \biggr) \Biggl\{ \ptl_t Y^R \ptl_R \mbo{\nu}' - \biggl(1- \frac{R^2_+}{R^2} \biggr)^2 \cdot \frac{R^4 \ptl_R \mbo{\nu}'}{ \bigl(\bigl(r^2+a^2\bigr)^2 -a^2 \Delta
 	\sin^2 \theta \bigr)} \notag\\
 & \hphantom{= \lim_{R \to \infty} 2 \sin \theta \biggl(1+ \frac{R^2_+}{R^2} \biggr) \Biggl\{}{}
 \times \ptl_R \Biggl(Y^t (t=0) - \mathcal{Y}^R \frac{1 + \frac{R^2_+}{R^2}} {{1-\frac{R^2_+}{R^2}} } - \mathcal{Y}^ \theta \cot \theta \Biggr) \notag\\
 & \hphantom{= \lim_{R \to \infty} 2 \sin \theta \biggl(1+ \frac{R^2_+}{R^2} \biggr) \Biggl\{}{}
 + \ptl_t Y^ \theta \ptl_\theta \mbo{\nu}' - \biggl(1 - \frac{R^2_+}{R^2} \biggr)^2 \cdot \frac{R^4 \ptl_\theta \mbo{\nu}'}{ \bigl(\bigl(r^2+a^2\bigr)^2 -a^2 \Delta \sin^2 \theta\bigr)} \notag\\
 & \hphantom{= \lim_{R \to \infty} 2 \sin \theta \biggl(1+ \frac{R^2_+}{R^2} \biggr) \Biggl\{}{}
 \times \ptl_R \Biggl(Y^t(t=0) - \mathcal{Y}^R \frac{1+ \frac{R^2_+}{R^2}}{1-\frac{R^2_+}{R^2}} - \mathcal{Y}^\theta \cot \theta \Biggr) \Biggr\}.
\end{align*}
In the second term in the curly brackets decay only after using the decay rate of $\ptl_ \theta \mbo{\nu}'$ \big(because $ \ptl_\theta Y^R$ does not decay near $\iota^0$\big)
\begin{align*}
(J_5)^b := -2N'^b \bar{\mu}_q q^{ac} \ptl_a \mbo{\nu}' \ptl_c N.
\end{align*}
Now let us look at flux
${\rm Flux} (J_5, \Gamma)$.
We have the flux of $J_5$ at the axes given by
\begin{align*}
&{\rm Flux} (J_5, \Gamma)= \int^t_0 \int_{(-\infty, -R_+) \cup (R_+, \infty)} \Bigl(\lim_{\theta \to 0, \pi}-2 R^2 N'^ \theta q^{ac}_0 \ptl_a \mbo{\nu}' \ptl_c N\Bigr) {\rm d}R {\rm d}t \\
&\qquad{}= \int^t_0 \int_{(-\infty, -R_+) \cup (R_+, \infty)} \biggl(\lim_{\theta \to 0, \pi} -2 R^2 \biggl(\ptl_ t Y^\theta - \frac{N^2}{R^2} {\rm e}^{-2 \mbo{\nu}'} \ptl_\theta Y^t \biggr) q^{ac}_0 \ptl_a \mbo{\nu}' \ptl_c N\biggr) {\rm d}R {\rm d}t.
\end{align*}
Now consider the integrand
\begin{align*}
\qquad&{}= \lim_{\theta \to 0, \pi}-2 R^2 \biggl(\ptl_ t Y^\theta - \frac{N^2}{R^2} {\rm e}^{-2 \mbo{\nu}'} \biggr) \cdot \notag\\
&\hphantom{= \lim_{\theta \to 0, \pi}}{} \times \biggl( \ptl_R \mbo{\nu}' \sin \theta \biggl(1 + \frac{R^2_+}{R^2}\biggr) + \frac{1}{R^2} \ptl_ \theta \mbo{\nu}' R \cos \theta \biggl(1- \frac{R^2_+}{R^2}\biggr) \biggr)
 \to 0 \qquad \text{as} \ \theta \to 0, \pi,
\end{align*}
and
\begin{align*}
{\rm Flux} \bigl(J_5, \mathcal{H}^+\bigr) &{}= \int^t_0 \int^\pi _0 \Bigl(\lim_{R \to R_+} q_0 \bigl(-2N'^b \bar{\mu}_q q^{ac} \ptl_a \mbo{\nu}' \ptl_c N, \ptl_R\bigr) \Bigr) {\rm d}\theta {\rm d}t \notag \\
&{}= \int^t_0 \int^\pi _0 \Bigl(\lim_{R \to R_+} - 2 N'^R \bar{\mu}_q q^{ac} \ptl_a \mbo{\nu}' \ptl_c N \Bigr) {\rm d}\theta d t \notag \\
&{}= \int^t_0 \int^\pi_0 \Bigl( \lim_{R \to R_+} -2\bigl(\ptl_t Y^R - N^2 {\rm e}^{-2 \mbo{\nu}} \ptl_R Y^t\bigr) \bar{\mu}_q q^{ac} \ptl_a \mbo{\nu}' \ptl_c N \Bigr) {\rm d}\theta {\rm d}t,
\end{align*}
the integrand contains the terms
\begin{align*}
&
=- 2 \Biggl\{ \ptl_t \Biggl(\sum^\infty_{n \to 1} Y^R_n R R^+_n \biggl(\frac{R^n}{R^n_+}- \frac{R_+^n}{R^n} \cos \theta\biggr) - R^2 \biggl(1- \frac{R^2_+}{R^2}\biggr)^2 \frac{1}{ \bigl(r^2 +a^2\bigr)^2 - a^2 \Delta \sin \theta} \Biggr) \notag\\
&\hphantom{=- 2 \Biggl\{}{}
\times\ptl_R \Biggl(Y^t(t=0)- \frac{1+ \frac{R^2_+}{R^2}}{R \bigl(1- \frac{R_+^2}{R^2}\bigr)} \sum^\infty_{n = 1} \mathcal{Y}^R_n R R^+_n \biggl(\frac{R^n}{R^n_+} \frac{R_+^n}{R^n} \cos \theta \biggr) \Biggr) \notag \\
&\hphantom{=- 2 \Biggl\{}{}
- \cot \theta \sum^\infty_{n=1} \mathcal{Y}_n^\theta R^n_+\biggl(\frac{R^n}{R^n_+} + \frac{R^n_+}{R^n} \biggr) \cos n\theta \Biggr\} \notag\\
&\hphantom{=}{}
 \times \biggl(\ptl_R \mbo{\nu}' \sin \theta \biggl(1 + \frac{R^2_+}{R^2}\biggr) + \frac{1}{R^2} \ptl_\theta \mbo{\nu}' R \cos \theta \biggl(1- \frac{R^2_+}{R^2}\biggr) \biggr)
 \to 0 \qquad \text{as} \ R \to R_+
\end{align*}
due to vanishing term in the final bracket.

We are now left with the final flux term
\begin{align*}
&{\rm Flux} (J_5, \bar{\iota}^0) = \int^t_0 \int^\pi_0 \Bigl(\lim_{R \to \infty} -2\bigl(\ptl_t Y^R - N^2 {\rm e}^{-2 \mbo{\nu}} \ptl_R Y^t\bigr) \bar{\mu}_q q^{ac} \ptl_a \mbo{\nu}' \ptl_c N \Bigr) {\rm d}\theta {\rm d}t \notag\\
&\qquad{} = \int^t_0 \int^\pi_0 \biggl(\lim_{R \to \infty} -2 \bigl(\ptl_t Y^R - N^2 {\rm e}^{-2 \mbo{\nu}} \ptl_R \bigl(Y^t (t=0) - \mathcal Y^R \ptl_R N/N -\mathcal{Y}^\theta \ptl_\theta N/N\bigr) \bigr) \notag\\
& \qquad\hphantom{= \int^t_0 \int^\pi_0 \biggl(}{} \times\biggl(R \ptl_R \mbo{\nu}' \sin \theta \biggl(1+ \frac{R^2_+}{R^2}\biggr) + \frac{1}{R} \ptl_ \theta \mbo{\nu}' \cos \theta \biggl(1- \frac{R^2_+}{R^2}\biggr) \biggr) \biggr) {\rm d}\theta {\rm d}t.
\end{align*}
Again the integrand can be estimated as
\begin{align*}
\qquad&{}= \lim_{R \to \infty} -2 \Biggl\{- \ptl_t \Biggl(\sum^{\infty}_{n=1} Y^R_n R^{-n+1} \cos n \theta \Biggr) - R^4 \biggl(1- \frac{R^2_+}{R^2}\biggr)^2 \frac{1}{ \bigl(r^2 +a^2\bigr)^2 - a^2 \Delta \sin \theta} \notag\\
& \qquad{}
\times \Biggl(Y^t(t=0)- \frac{1+ \frac{R^2_+}{R^2}}{R \bigl(1- \frac{R_+^2}{R^2}\bigr)} \sum^{\infty}_{n=1} \mathcal{Y}^R_n R^{-n+1} \cos n \theta - \cot \theta \sum^\infty_{n=1} \mathcal{Y}^\theta_n R^{-n} \sin n\theta \Biggr) \Biggr\} \notag\\
& \quad{}\times \biggl\{ R \cdot \mathcal{O} \biggl(\frac{1}{R}\biggr) \mathcal{O}(1) + \frac{1}{R} \cdot \mathcal{O}\biggl(\frac{1}{R}\biggr) \cdot \mathcal{O}(1) \biggr\} \to 0 \qquad \text{as} \ R \to \infty,
\end{align*}
and behaves like $\mathcal{O}\bigl(\frac{1}{R}\bigr)$ for large $R$.
Thus it follows that
\begin{align*}
{\rm Flux}\bigl(J_5, \iota^0\bigr) = 0.
\end{align*}

\subsection*{Acknowledgements}
I acknowledge the gracious hospitality of Institut des Hautes \'Etudes Scientifiques (IHES) at Bures-sur-Yvette in Fall 2016.
In a previous work \cite{GM17_gentitle}, the coauthor Vincent Moncrief had fittingly paid tribute to A.~Taub, J.~Marsden and S.~Dain for their influence on him and for making fundamental contributions in this direction. I take this opportunity to pay tribute to his own outstanding contributions to general relativity and Hamiltonian methods. On an individual note, for his mentorship, encouragement, and enjoyable interactions, I am indebted to him.

This project has spanned several years, and my gratitude is also due to all the colleagues who have continually supported me and to the institutions that hosted me. Special thanks are due to my host Hermann Nicolai at the Albert Einstein Institute at Golm (DFG grant no. GU 1513/2-1), where significant aspects of this work were completed. A part of this work were done while I was a Postdoc at the Department of Mathematics, Yale University, where I benefited from conducive working conditions. Finally, I thank the referees for the feedback.

This work was supported by the German Research Foundation (Grant numbers GU 1513/1-1 and GU 1513/2-1).

\LastPageEnding

\end{document}